\documentclass[12pt]{article}
\usepackage{amsmath, amsthm}
\usepackage{graphicx}
\usepackage{enumerate}
\usepackage{natbib}
\usepackage{url} 
\usepackage{amssymb,amsfonts}
\usepackage{subcaption}
\usepackage{xcolor}
\usepackage{mathrsfs}
\usepackage{float}
\usepackage{mathtools}
\usepackage{booktabs}
\newcommand{\blind}{0}
\usepackage[hidelinks]{hyperref}
\usepackage[nameinlink]{cleveref}
\creflabelformat{equation}{#1#2#3}
\crefname{equation}{eq.}{eqs.}
\Crefname{equation}{Eq.}{Eqs.}
\crefname{figure}{fig.}{figs.}
\Crefname{figure}{Fig.}{Figs.}
\Crefname{table}{Tab.}{Tabs.}
\crefname{table}{tab.}{tabs.}
\crefname{appendix}{app.}{apps.}
\Crefname{appendix}{App.}{Apps.}
\Crefname{sec}{Section}{Sections}
\Crefname{subsec}{Subsection}{Subsections}
\Crefname{chapt}{Chap.}{Chapters}
\crefname{thm}{thm.}{thms.}
\Crefname{thm}{Thm.}{Thms.}
\Crefname{lem}{Lem.}{Lems.}
\Crefname{cor}{Cor.}{Cors.}
\Crefname{prop}{Prop.}{Props.}
\Crefname{ass}{Assumption}{Assumptions}

\newcommand{\supp}{\mbox{supp}}

\newcommand{\be}{\begin{equation}}
\newcommand{\ee}{\end{equation}}
\newcommand{\ben}{\begin{equation*}}
\newcommand{\een}{\end{equation*}}
\newcommand{\bea}{\begin{eqnarray}}
\newcommand{\eea}{\end{eqnarray}}
\newcommand{\bean}{\begin{eqnarray*}}
\newcommand{\eean}{\end{eqnarray*}}
\newcommand{\bi}{\begin{itemize}}
\newcommand{\ei}{\end{itemize}}
\newcommand{\bc}{\begin{center}}
\newcommand{\ec}{\end{center}}
\newcommand{\bcf}{\begin{figure}\begin{center}}
\newcommand{\ecf}{\end{center}\end{figure}}

\newcommand{\E}{\mathbb{E}}

\newcommand{\R}{\mathbb{R}}

\newcommand{\bfbb}{\mbox{\bf I\kern-.45em B}}

\newcommand{\bfbeta}{\boldsymbol{\beta}}

\newcommand{\calb}{{\mathcal B}}

\newcommand{\calx}{{\mathcal X}}
\newcommand{\caly}{{\mathcal Y}}

\newcommand{\calp}{{\mathcal P}}
\newcommand{\calf}{{\mathcal F}}

\newcommand{\calg}{{\mathcal G}}

\newcommand{\calm}{{\mathcal M}}

\newcommand{\calt}{{\mathcal T}}

\newcommand{\argmin}{\mbox{argmin}}

\newcommand{\leftarray}{\left\{\begin{array}{lll}}
\newcommand{\rightarray}{\end{array}\right.}

\numberwithin{equation}{section}

\newtheorem{prop}{Proposition}[section]
\newtheorem{theorem}{Theorem}[section]
\newtheorem*{theorem*}{Theorem}

\newtheorem{lemma}[theorem]{Lemma}
\newtheorem{corollary}[theorem]{Corollary}
\newtheorem*{corollary*}{Corollary}
\newtheorem{ass}{Assumption}[section]

\theoremstyle{remark}
\newtheoremstyle{boldexample}
  {.5\topsep}{.5\topsep}
  {\normalfont}{0pt}
  {\bfseries}{.}{5pt plus 1pt minus 1pt}
  {\thmname{#1}\thmnumber{ #2}\thmnote{ (#3)}}
\theoremstyle{boldexample}
\newtheorem{definition}[theorem]{Definition}
\newtheorem*{example}{Example}
\newtheorem*{remark}{Remark}
\theoremstyle{remark}

\graphicspath{{Figs/}}

\usepackage{libertine}
\usepackage{cancel}

\usepackage[all]{nowidow}

\begin{document}

\def\spacingset#1{\renewcommand{\baselinestretch}%
{#1}\small\normalsize} \spacingset{1}

\if1\blind
{
  \title{\bf Title}
  \author{Yidan Xu\\
    Department of Statistics, University of Michigan\\
    yidanxu@umich.edu\\
    and\\
    Yixin Wang\\
    Department of Statistics, University of Michigan\\
    yixinw@umich.edu\\
    and\\
    XuanLong Nguyen\\
    Department of Statistics, University of Michigan\\
    xuanlong@umich.edu}
  \maketitle
} \fi

\if0\blind
{
  \title{Stable Causal Estimation with Optimal Transport}
  \author{
    Yidan Xu, Yixin Wang, and XuanLong Nguyen\\
    {\normalsize Department of Statistics, University of Michigan}\\
    {\normalsize\texttt{\{yidanxu, yixinw, xuanlong\}@umich.edu}}
  }
  \date{}
  \maketitle
} \fi

\begin{abstract}
\noindent
In observational studies, the overlap assumption, also called positivity, is frequently strained in practice. Limited overlap destabilizes widely used estimators such as inverse propensity weighting (IPW), whose variance depends on density ratios that explode as propensity scores approach zero or one. To address this challenge, we propose the \emph{Stable Transport Estimator (STE)}, an optimal transport-based approach for estimating average causal effects when overlap holds but strict overlap fails. 
Unlike IPW, the variance of STE contains no density-ratio factor; it is governed by the transport map itself and stays finite under bounded moments, even when density ratios diverge.
By replacing reweighting with a geometry-aware transport map, STE constructs stable counterfactual comparisons without requiring density estimation or propensity score modeling. STE further extends to settings where the overlap condition is violated: when covariate supports overlap only partially, a linear interpolation transport map enables principled extrapolation with explicit bias-variance trade-offs. We establish consistency, asymptotic normality, and a variance characterization that makes explicit the contribution of the estimated transport map to the asymptotic variance. Across synthetic and real datasets, STE demonstrates improved robustness, reduced bias, and markedly more stable variance compared to existing estimators.
\end{abstract}

\noindent%

\section{Introduction}
\label{sec:introduction}

Causal inference from observational data is central to policy evaluation in
economics, public health, and social science \citep{rubin1974estimating,holland1986statistics,pearl2009causality,
rosenbaum2002observational,imbens2015causal}.  When randomized experiments are
infeasible, identification typically relies on three assumptions: SUTVA (Stable Unit Treatment Value Assumption), unconfoundedness (also called ignorability), and overlap (also called positivity) \citep{rubin1974estimating,rosenbaum1983central,imbens2015causal}.  Among these, overlap is the assumption most frequently strained in practice.  Treated and control units may share the same support in principle, yet their covariate distributions can be so far apart that standard estimators become numerically unstable and exhibit exploding variance \citep{crump2009dealing, khan2010irregular, petersen2012diagnosing, damour2021overlap, rothe2017robust}. 
Two examples below illustrate the two regimes of overlap failure that this paper addresses.

\textbf{Two canonical failure cases.} 
\Cref{fig:meps} illustrates a near-structural failure in the Medical
Expenditure Panel Survey (MEPS) \citep{meps}. After excluding disability-based enrollees,
virtually every person aged 65 and older is covered by Medicare, leaving only a
small and highly atypical group without coverage. The covariate distributions of the two groups are therefore nearly disjoint: there are essentially no non-Medicare analogs for a typical 65-and-older beneficiary. Although overlap may hold formally, propensity scores are so close to $0$ or $1$ for most treated units that the density ratios of the distributions associated with the two groups are effectively unbounded. The causal estimand of primary interest is the Average Treatment Effect on the Treated (ATT), which measures how outcomes for Medicare beneficiaries would differ absent coverage. As a result, it cannot be credibly identified by inverse propensity weighting (IPW) \citep{horvitz1952generalization,hirano2003efficient} or by any other estimator that relies on finite pointwise density ratios.

A second, qualitatively different failure mode arises even when the two groups technically share the same covariate support.  In the LaLonde-PSID comparison \citep{lalonde1986evaluating,dehejia1999causal}, NSW trainees are predominantly young, low-income individuals with sparse prior labor-force histories, while PSID respondents are older, more educated, and have substantially higher and more stable earnings histories.  Although neither group is deterministically excluded from any region of covariate space, their empirical densities differ sharply. Some covariate regions are populated almost exclusively by treated units, others almost exclusively by controls, so propensity scores degenerate near $0$ or $1$ and IPW estimators become highly unstable. 

\begin{figure}
    \centering
    \includegraphics[width=\linewidth]{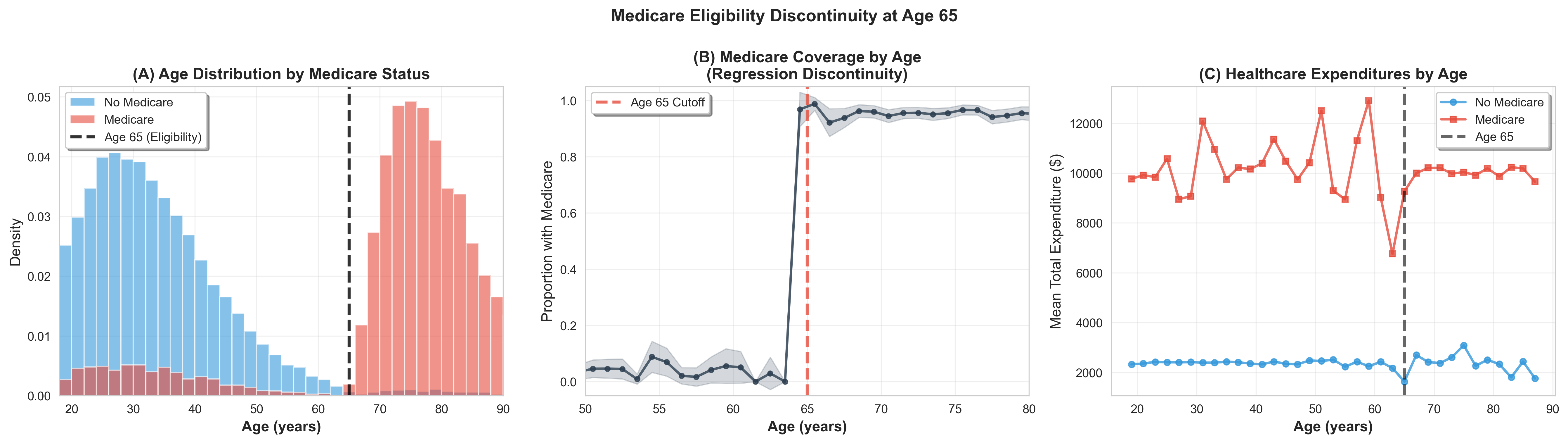}
    \caption{Medicare eligibility discontinuity at age 65 in MEPS dataset. The peak in Medicare-covered distribution below age 65 reflects a separate
    sub-population enrolled on the basis of disability or ESRD.}
    \label{fig:meps}
\end{figure}

\textbf{Why standard weighting breaks down.}
Both pathologies share a common root: the variance of IPW estimators scales with
the pointwise density ratio. Even when the overlap
assumption holds formally, strict overlap (\Cref{ass:strong_positivity})
frequently fails. Whenever the treated and control covariate distributions
diverge, because their supports barely overlap or because their tails vanish
at different rates, this ratio becomes arbitrarily large. When overlap is violated outright, as in MEPS, the ratio ceases to exist on part of the treated support. The resulting weights explode or become undefined, inflating variance and rendering inference sensitive to tail behavior, modeling choices, and small perturbations in the data \citep{busso2014new, damour2021overlap, rothe2017robust}.

\textbf{Stable causal estimation via optimal transport.}
Rather than reweighting control units to resemble treated units, we propose aligning covariate distributions through the geometry of the covariate space itself.  Optimal transport (OT) \citep{mccann1995existence, kantorovich1942translocation, rachev1998mass, villani2003topics, villani2009optimal} provides a natural set of tools for such alignment: instead of comparing treatment and control covariate distributions pointwise via density ratios, OT seeks the least costly way to move mass from one distribution to the other. When the transport cost is a power $p \geq 1$ of the Euclidean distance, this discrepancy is the Wasserstein-$p$ distance \citep{villani2009optimal}. It remains finite whenever $p$-th moments are bounded, even when density ratios explode. This is the key stability advantage we exploit. When strict overlap fails, the transport cost can remain well behaved even though IPW weights do not.

Motivated by this observation, we propose the \emph{Stable Transport Estimator (STE)}. The STE transports the control covariate distribution toward the treated distribution, and then imputes untreated potential outcomes via the transported control regression:
\begin{equation}
\label{eq:att_definition}
\begin{aligned}
\tau_{\text{ATT}} &= \mathbb{E}[Y(1)|A=1] - \mathbb{E}[Y(0)\mid A=1]\\
&= \mathbb{E}[Y|A=1] - \mathbb{E}_{X\sim P(X\mid A=0)}\bigl[\mathbb{E}[Y\mid {T}(X),A=0]\bigr],
\end{aligned}
\end{equation}
where we denote $A$ as the treatment assignment, $X$ as the observed covariates, $Y$ as the outcome, and $T$ as the transport map from controls to treated units satisfying $T_{\sharp} P_X^0 = P_X^1$. 
Crucially, STE involves neither propensity scores nor density ratios; it evaluates the outcome regression at transported control covariates, assigning each control unit a geometry-aware counterfactual position in the treated covariate space.  For $T$ to exist and be well defined, we require only finite transport cost, a bounded-moment condition, rather than bounded density ratios.  When the outcome regression is Lipschitz, the composition with $T$ varies smoothly, so the estimator avoids the numerical instability of multiplying outcomes by large weights. STE is therefore valid throughout the strict-overlap-failure regime, as in the LaLonde example, where overlap holds but the density ratio is unbounded.

\textbf{Extrapolation under overlap violation.} The transport approach also suggests an immediate extension to the harder regime where the overlap condition is violated (\Cref{def:lack_overlap_population}).
When the treated and control supports overlap only partially, transporting all the way to the target population forces severe extrapolation of the outcome regression into unseen covariate regions. We therefore introduce an interpolating proxy distribution, $P_X^{\eta} = (\eta\,\mathrm{Id} + (1-\eta) T)_{\sharp} P_X^0$, a linear interpolation along constant-speed geodesics in the Wasserstein space \citep{ambrosioGradientFlowsMetric2008}. 
The proxy coincides with the target at $\eta = 0$ and stays on the source support at $\eta = 1$, retaining overlap with the training distribution. Indexing the interpolation by the hyperparameter $\eta$ creates an explicit bias--variance trade-off: moving toward $\eta=0$ better approximates the target population at the cost of extrapolation, while moving toward $\eta=1$ reduces extrapolation error at the cost of residual covariate imbalance. The resulting algorithm provides a coherent way to reason about causal estimation even when standard overlap-based methods cease to apply.

\textbf{Contributions.}
This paper advances causal inference methodology along four dimensions.
\emph{(i)~Identification:} We develop a transport-based identification strategy that replaces density-ratio adjustment with geometry-aware change-of-variables under SUTVA, unconfoundedness, and overlap.
\emph{(ii)~Estimation with theoretical guarantees:} We propose the Stable Transport Estimator (STE) for ATT, ATC, and ATE estimation, and show that its variance contains no density-ratio factor. We prove consistency whenever the plug-in map and the outcome regression are $L_2$ consistent. We establish a central limit theorem under explicit conditions that characterize the contribution of the estimated transport map to the asymptotic variance. For parametric exact-map plug-ins, this contribution admits an exact variance identity.
\emph{(iii)~Principled extrapolation:} We extend STE to settings where overlap is violated through partial support overlap by introducing a linear interpolation along the transport path. The interpolated estimator is consistent under the same conditions as the STE, and satisfies a central limit theorem in which the interpolation bias is explicit to first order, estimable, and removable, with the asymptotic variance unchanged.
\emph{(iv)~Empirical evidence:} We demonstrate through extensive synthetic and real data empirical studies that STE delivers markedly more stable variance and improved robustness relative to standard alternatives under failure of strict overlap and violation of overlap.

\subsection{Related Work}

From unbounded density ratios to outright support mismatch, failures of overlap pose a fundamental challenge in observational study.

\textbf{Trimming and estimand redefinition.}
One natural response to overlap failures is to restrict the inferential target. Trimming and truncation methods discard units outside the region of common support \citep{crump2009dealing, lee2011weight}, while alternative estimand approaches redefine the causal quantity to reduce or eliminate dependence on the overlap assumption \citep{li2019addressing, kennedy2019nonparametric}. Although practically useful, both strategies alter the target population and thereby limit the ability to draw conclusions about the original population of interest.

\textbf{Balancing weights and population mixing.}
A separate line of work stabilizes propensity score weighting through balancing weight methods that constrain extreme weights arising from model misspecification \citep{ben2021augmented, hainmueller2012entropy, imai2014covariate}. While empirically effective in moderate settings, these methods primarily correct for propensity score model error rather than genuine overlap failures, and their constrained optimization procedures can become computationally infeasible when overlap is severely violated. More recently, \cite{jang2024mixing} proposed constructing a synthetic population by mixing the treated and control distributions through convex combinations, improving overlap and enabling more efficient inverse propensity weighting. Although conceptually related to our work in that both approaches construct alternative distributions with better overlap properties, the two approaches differ fundamentally: \cite{jang2024mixing} rely on population mixing, whereas we employ optimal transport maps that preserve the metric structure of the covariate space. This distinction enables us to provide formal bias-variance characterizations under failure of strict overlap and to extend to settings where no mixing proportion yields adequate overlap.

\textbf{Optimal transport in causal inference.}
A growing literature applies optimal transport to causal inference, though with objectives distinct from ours. \cite{gunsilius2021matching} developed a probabilistic matching estimator using unbalanced optimal transport plans to improve covariate balance in matched samples. \cite{torous2024optimal} employed the Brenier map to construct a changes-in-changes estimator for panel data under parallel trends assumptions. \cite{lin2023causal} proposed Wasserstein barycenter methods for distributional treatment effect estimation, where the estimand is itself a probability distribution. These contributions address matching, panel data, and distributional estimands, respectively, rather than the overlap problem in cross-sectional observational studies.

\textbf{Our work.}
In contrast to all of the above, our work uses transport maps directly to define identification strategies and construct bias-corrected estimators under failure of strict overlap and, in its extension, under violation of overlap. By leveraging the metric structure of the covariate space, our approach improves estimation stability, reduces bias and variance, and broadens the applicability of causal inference in settings where existing methods either fail or require abandoning the original estimand.

\subsection{Organization}

The remainder of the paper is organized as follows. \Cref{sec:identification} introduces the formal setup and assumptions, distinguishes overlap, strict overlap, and violation of overlap, and revisits why standard weighting estimators become unstable as overlap deteriorates. \Cref{sec:monge_OR} develops the transport-based identification strategy and the Stable Transport Estimator. \Cref{sec:theory} establishes its theoretical properties: an exact bias and variance decomposition, consistency, and asymptotic normality with an explicit contribution of the estimated transport map to the variance. \Cref{sec:interpolation} extends STE to the overlap-violation setting through interpolation along the transport path. \Cref{sec:sim} presents simulations and empirical illustrations, and \Cref{sec:conclusion} concludes.

\section{From Density Ratios to Transport Maps: Covariate Adjustment under Limited Overlap}
\label{sec:identification}
\subsection{The Overlap Condition in Causal Inference}
 
We consider the standard observational causal inference setup with binary treatment. Given i.i.d.\ draws
$$
(Y_i(1),Y_i(0),A_i,X_i)\sim \mathcal P,
$$
where $A_i\in\{0,1\}$ is the treatment indicator, $X_i\in\mathcal{X}$ is the covariate vector, and $Y_i(1),Y_i(0)$ are the potential outcomes. We take $\mathcal{Y}\subset\mathbb{R}$ and $\mathcal{X}$ to be a $d$-dimensional complete separable metric space (i.e., Polish space). We suppress the index $i$ when no confusion arises. 
The realized outcome is linked to the potential outcomes through SUTVA: $Y = AY(1)+(1-A)Y(0)$.

Identifying average causal effects from observational data often requires three assumptions: Stable Unit Treatment Value Assumption (SUTVA), unconfoundedness, and overlap \citep{rosenbaum1983central, imbens2015causal}.

\begin{ass}[Stable Unit Treatment Value Assumption (SUTVA)]
\label{ass:SUTVA}
SUTVA assumes no interference between units and that the treatment applied is consistent across all units.
\end{ass}
 
\begin{ass}[Unconfoundedness (a.k.a.\ Ignorability)]
\label{ass:confoundedness}
The potential outcomes are independent of treatment assignment conditional on covariates:
$$
Y(1),Y(0)\perp A\mid X.
$$
\end{ass}
 
\begin{ass}[Overlap (a.k.a.\ Positivity)]
\label{ass:positivity}
Every unit has a strictly positive probability of receiving either treatment level, conditional on its covariates:
$$
0<e(x):= \mathcal P(A=1\mid X=x)<1,\qquad x\in\mathcal{X},
$$
where $e(x)$ denotes the propensity score.
\end{ass}

Overlap must be distinguished from the stronger condition of \emph{strict overlap}, sometimes called \emph{strong overlap}, which requires the propensity score to be uniformly bounded away from~$0$~and~$1$.

\begin{ass}[Strict Overlap]
\label{ass:strong_positivity}
There exist constants $0<\alpha<\beta<1$ such that
$$
\alpha<e(x)<\beta,\qquad x\in\mathcal X.
$$
\end{ass}

Overlap is sufficient for identification, but strict overlap is what many weighting arguments effectively rely on for stable estimation. In practice, overlap often comes close to failing, and this typically manifests as failure of strict overlap: propensity scores approach $0$ or $1$, and hence propensity-score weights can become arbitrarily large, even though overlap may still hold formally. This is precisely the regime where classical inverse-probability weighting becomes unstable \citep{busso2014new,damour2021overlap}.

To see this, we first introduce the treated and control covariate distributions
$$
P_X^1:= \mathcal P(X\mid A=1),\qquad P_X^0:=\mathcal P(X\mid A=0),
$$
and the conditional mean functions
$$
h_a(x):=\mathbb E[Y\mid X=x,A=a],\qquad a\in\{0,1\}.
$$
Under \Cref{ass:SUTVA}-\Cref{ass:positivity}, the average treatment effect on the treated (ATT) can be written as
$$
\begin{aligned}
\tau_{\mathrm{ATT}}
&=
\mathbb E[Y(1)\mid A=1]-\mathbb E[Y(0)\mid A=1]\\
&=
\mathbb E_{P_X^1}[\E[Y \mid X=x, A=1]]-\mathbb E_{P_X^1}[\E[Y \mid X=x, A=0]]\\
&=
\mathbb E_{P_X^1}[h_1(X)]-\mathbb E_{P_X^1}[h_0(X)].
\end{aligned}
$$
The second term in the last line poses the key difficulty: the counterfactual mean $h_0(x) := \E[Y \mid X=x, A=0]$ must be evaluated under the treated covariate distribution $P_X^1$, not the control distribution $P_X^0$ under which the control outcomes are naturally observed. This requires a change of measure from $P_X^0$ to $P_X^1$.

\subsection{IPW as a Change of Measure via Density Ratios}

A standard solution to estimate the second term is inverse-probability weighting (IPW), which performs this change of measure through a pointwise density ratio, reweighting the control distribution $P_X^0$, under which the outcomes are observed, toward the treated distribution $P_X^1$, under which the expectation is required. Assuming both $P_X^1$ and $P_X^0$ admit densities $p_X^1$ and $p_X^0$ with respect to a common dominating measure, the second term can be rewritten as
$$
\mathbb E_{P_X^1}[h_0(X)]
=
\mathbb E_{P_X^0}\!\left[h_0(X)\frac{p_X^1(X)}{p_X^0(X)}\right].
$$
By Bayes' rule, the density ratio can be expressed in terms of the propensity score:
$$
\frac{p_X^1(x)}{p_X^0(x)}
=
\frac{e(x)(1-r)}{(1-e(x))r},
\quad \text{where } r:=\mathcal P(A=1).
$$
so that IPW upweights precisely the control units whose covariates are more representative of the treated population.

The validity of IPW depends on the density ratio being pointwise well-defined and finite. When the density ratio $p_X^1/p_X^0$ is unbounded, or, equivalently, when the propensity score $e/(1-e)$ approaches $1$, IPW weights become unstable and the resulting estimator suffers from large or infinite variance. Crucially, this instability can arise even when the formal overlap assumption (\Cref{ass:positivity}) holds. For instance, two distributions may have common support over all of $\R$, so that $e(x) \in (0,1)$ everywhere, yet the density ratio $p_X^1/p_X^0$ can still grow without bound in the tails. A small number of observations then receive enormous weights, variance inflates, and estimation becomes highly sensitive to tail behavior, model misspecification, and finite-sample noise \citep{busso2014new,damour2021overlap}. This is precisely the failure of strict overlap (\Cref{ass:strong_positivity}), a regime in which IPW-based estimators become unstable in practice.

The density ratio also makes the assumption hierarchy transparent: by the Bayes relation above, $p_X^1/p_X^0$ is finite a.e. exactly when overlap holds (\Cref{ass:positivity}), and is in addition uniformly bounded exactly when strict overlap holds (\Cref{ass:strong_positivity}). The two ways this degrades give the two modes of overlap failure below: the ratio may be finite almost everywhere but unbounded, so that overlap holds while strict overlap fails, or it may be undefined, so that overlap itself fails.

\begin{definition}[Violation of overlap]\label{def:lack_overlap_population}
Overlap (\Cref{ass:positivity}) is \emph{violated} if, for some $a \in \{0,1\}$, the ratio
$\dfrac{\mathcal P(\mathrm{d}x \mid A=a)}{\mathcal P(\mathrm{d}x \mid A \neq a)}$
is not well-defined; that is, $\mathcal P(\cdot \mid A=a) \not\ll \mathcal P(\cdot \mid A \neq a)$,
so there is a set $B$ with $\mathcal P(B \mid A=a)>0$ but $\mathcal P(B \mid A\neq a)=0$,
for example when the supports are disjoint or only partially overlapping.
\end{definition}

\Cref{def:lack_overlap_population} captures the case where overlap fails, for example when discrete covariate distributions have disjoint support.

\begin{definition}[Near-violation of overlap / failure of strict overlap]\label{def:lack_overlap_practical}
Suppose overlap (\Cref{ass:positivity}) holds, so the ratio
$\dfrac{\mathcal P(\mathrm{d}x \mid A=a)}{\mathcal P(\mathrm{d}x \mid A \neq a)}$
is well-defined and finite a.e. Fix a constant $C>0$. Overlap is \emph{$C$-nearly violated} if, for
some $a \in \{0,1\}$,
$$
\operatorname*{ess\,sup}_{x \in \calx} \frac{\mathcal P(\mathrm{d}x \mid A=a)}{\mathcal P(\mathrm{d}x \mid A \neq a)} > C,
$$
including the limiting case where this supremum is $+\infty$.
\end{definition}

\Cref{def:lack_overlap_practical} quantifies how close the design is to violating strict overlap: the larger the constant $C$ for which the ratio is $C$-nearly violated, the more severe the near-violation. Strict overlap
(\Cref{ass:strong_positivity}) requires the density ratio to admit a finite uniform bound.
It fails precisely in the limiting case of \Cref{def:lack_overlap_practical}, in which
the ratio exceeds every threshold $C$ on a set of positive probability, so that
this essential supremum equals $+\infty$.

\subsection{Transport Maps as an Alternative Change of Measure}
\label{subsec:transport-maps}

Taking a step back, IPW is just one way to perform the change of measure from $P_X^0$ to $P_X^1$: it reweighs each control unit at $x$ by the multiplicative factor $p_X^1(x)/p_X^0(x)$, measuring the discrepancy between the two distributions pointwise through the density ratio function.

But there are other ways to change variables and measure distances between distributions. Rather than comparing $P_X^1$ and $P_X^0$ pointwise, one can move mass from one distribution to the other using a transport map, and measure the discrepancy by the cost of this transport. To this end, we formally define the measurable space $(\calx,\calb(\calx))$, and let $\calp_{p}(\calx)$ be the space of probability measures with finite $p$-th moments, $\calp_{\mathrm{ac}}$ the space of probability measures that are absolutely continuous w.r.t. the Lebesgue measure.

Specifically, a \emph{Monge map} (or optimal transport map) $T : \calx \to \calx$ \citep{brenier1991polar, gangbo1996geometry, villani2009optimal} is a deterministic map such that if $X \sim P_X^0$, then $T(X) \sim P_X^1$; that is, $T$ pushes forward $P_X^0$ to $P_X^1$, written $T_{\sharp} P_X^0 = P_X^1$, and minimizes a cost of transportation. By Brenier's theorem when the transport cost for an unit of mass given by $c(x,y) = \|x - y\|^p$, for $p> 1$, and provided that $P_X^0$ is absolutely continuous w.r.t. the Lebesgue measure, the Monge map exists uniquely and minimizes
$$
\inf_{T:\, T_{\sharp} P_X^0 = P_X^1} \Big(\int \|x - T(x)\|^p \, \mathrm{d}P_X^0(x)\Big)^{1/p},
$$
which defines the Wasserstein-$p$ distance $W_p(P_X^0, P_X^1)$ \citep{villani2009optimal}.

The crucial observation is that the Wasserstein distance and the transport map that achieves it do not require the existence or boundedness of the density ratio $p_X^1/p_X^0$. It requires only that the distributions have finite $p$-th moments. Moreover, the transport cost remains bounded even when the density ratio explodes, making transport-based methods naturally robust to violations and near-violations of overlap.

\begin{remark}[IPW as a degenerate local transport]
\label{rem:ipw-special-case}
We can put the two methods into the same perspective of correcting covariate shift by a density ratio. IPW compares $p_X^1$ and $p_X^0$ at the same point $x$ and reweights control units in place. Under conditions where the Monge map and its inverse exist, by Monge--Amp\`ere theorem \citep{villani2009optimal}, the Monge map $T$ with inverse $S$ with $S_{\sharp} P_X^1 = P_X^0$ instead compares $p_X^1(x)$ against $p_X^0$ at the preimage $S(x)$, moving mass rather than rescaling it.
At the population level the two estimands coincide. We clarify the connection in \Cref{app:ipw-connection}, \Cref{prop:monge-ipw}. The per-unit weights, however, coincide only when the map is the identity. This happens only in the degenerate case of no covariate shift. Recovering IPW under genuine shift requires two ingredients: a localizing cost that suppresses transport, and a relaxation of mass conservation on the control side. In that limit the transport weight reduces to the density ratio exactly (\Cref{prop:entropic-ipw}\eqref{it:recover}).
Therefore IPW is a special case of the transport formulation. Where overlap holds, the localized coupling stays on the diagonal and the two procedures agree. Where overlap fails, there is no nearby control mass, so the diagonal coupling is infeasible. Transport remains well defined exactly where the density ratio $p_X^1/p_X^0$ diverges, and in this sense it strictly generalizes IPW.
\end{remark}

Before developing how transport maps can be used for causal identification and estimation in subsequent sections, we illustrate this contrast through concrete examples.

\subsection{Examples: Density ratios versus transport maps}
\label{sec:examples}

The examples below illustrate two distinct regimes. In the first, overlap still holds but is nearly violated, so strict overlap fails and density ratios can be arbitrarily large. In the second, overlap itself fails, so density-ratio adjustment is not even well defined. In both regimes, transport-based discrepancies remain well defined and stable.

\begin{figure}[t]
    \centering
    \includegraphics[width=\linewidth]{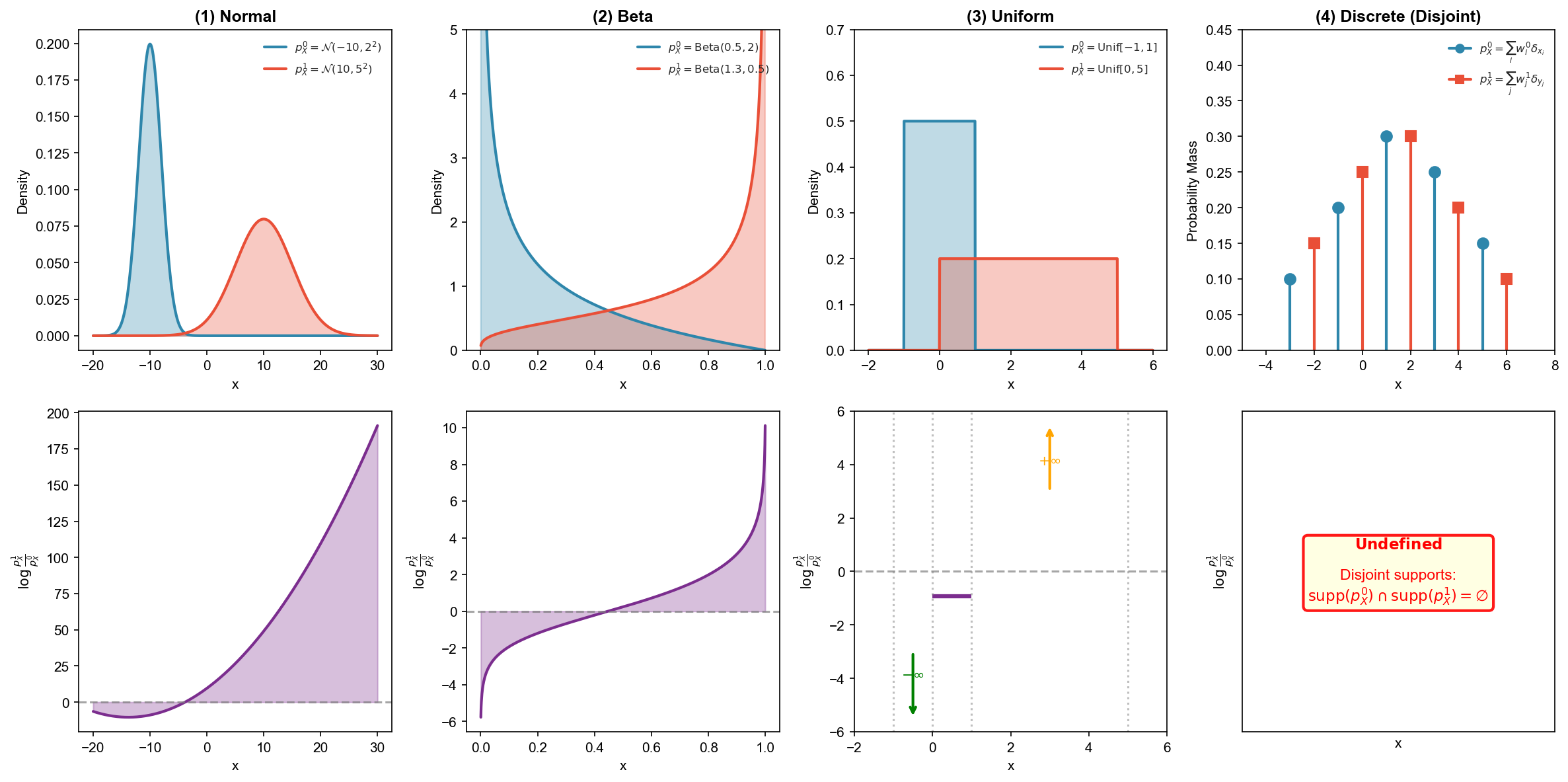}
    \caption{\footnotesize Examples illustrating violations of overlap. The first two columns satisfy overlap but nearly violate it in the sense of Definition \ref{def:lack_overlap_practical}; they violate strict overlap. The last two columns violate overlap entirely. Top row: covariate distributions. Bottom row: log density ratio $dP_X^1/dP_X^0$. The density ratio can become extremely large or undefined, while transport-based discrepancies remain well defined.}
    \label{fig:examples}
\end{figure}

\begin{example}[Near violation of overlap: Normal distributions]
\label{exp:vio_overlap}
Let $P_X^0=N(-10,2^2)$ and $P_X^1=N(10,5^2)$. Both distributions have full support on $\mathbb{R}$, so the formal overlap condition holds. However, the density ratio grows without bound in the tails: $\sup_{x\in\mathbb{R}}\, p_X^1(x)/p_X^0(x) = \infty$. Even at moderate values such as $x=30$, the ratio reaches ${\sim}10^{83}$. This is a near-violation of overlap: strict overlap fails, and IPW weights become extremely unstable for units in the tails even though overlap holds.

By contrast, because both distributions are Gaussian, the Monge map from $P_X^0$ to $P_X^1$ is simply an affine transformation, and the Wasserstein distance $W_2(P_X^0, P_X^1)$ is a closed-form function of the means and variances, finite and well-behaved regardless of the tail behavior of the density ratio. Thus the transport discrepancy remains controlled exactly where density-ratio adjustment becomes unstable.
\end{example}

 \begin{example}[Near violation of overlap: Beta Distributions]
\label{exp:unbounded}
Consider $P_X^0 = \mathrm{Beta}(0.5, 2)$ and $P_X^1 = \mathrm{Beta}(1.3, 0.5)$. As shown in \Cref{fig:examples}, the density ratio $p_X^1(x)/p_X^0(x) \to \infty$ as $x \to 1^-$. Near the right boundary, the ratio takes extreme values, causing IPW weights to blow up. This is a near-violation of overlap: strict overlap fails despite the density ratio being formally finite on the interval $(0,1)$.
 
On the other hand, because both distributions are supported on the compact set $[0,1]$, the transport cost $c(x,y) = |x-y|^2$ is uniformly bounded, and the Wasserstein distance remains finite. Transport-based adjustment thus remains well-defined and stable even where the density ratio~diverges.
\end{example}

\begin{example}[Violation of overlap: Uniform distributions]
Denote the following:
$$
P_X^0=\mathrm{Unif}[-1,1],\qquad P_X^1=\mathrm{Unif}[0,5].
$$
The supports overlap only on $[0,1]$; outside this interval the density ratio is undefined. Hence overlap fails and IPW cannot identify the ATT through density-ratio reweighting on the full treated population.

Nevertheless, both distributions have finite moments on compact supports, so the Wasserstein distances are well-defined and the corresponding Monge map exists. This illustrates that transport-based constructions continue to make sense even when the standard density-ratio adjustment breaks down completely.
\end{example}

\begin{example}[Violation of overlap: Discrete measures]
Consider
$
P_X^1=\sum_{i=1}^5 p_i\delta_{x_i},
\qquad
P_X^0=\sum_{i=1}^5 q_i\delta_{x_i+1}.
$
These two discrete distributions have disjoint supports. The density ratio is therefore undefined everywhere on the support of either distribution, so IPW is unavailable.

Transport-based discrepancies, however, remain well-defined: one may still transport mass from each support point of $P_X^0$ to points in the support of $P_X^1$, and the resulting transport cost is finite. This makes the contrast with density-ratio methods particularly stark in discrete settings.
\end{example}

These examples illustrate clearly that the density ratio can become unbounded or undefined even in simple settings, rendering IPW-based estimators unstable or infeasible. In contrast, Wasserstein distances and the associated transport mappings remain well-defined and bounded whenever the distributions have finite moments, without requiring any form of strict overlap. This motivates the transport-based approach to covariate adjustment that we develop in the remainder of the paper.

\section{Stable Causal Estimation via Optimal Transport}
\label{sec:monge_OR}

Having established that transport maps provide a well-defined change of measure even when density ratios become unbounded or undefined, we now turn to their use for causal identification and estimation. We focus on the deterministic transport map, whose plug-in estimator based on empirical distributions is stable and converges asymptotically to the population-level map as the sample size tends to infinity \citep{deb2021rates, manole2021plugin, goldfeld2024limit, barrio2023improved}. In this section, we establish the identifiability conditions required by the transport map approach and introduce the Stable Transport Estimator (STE) for the ATT, ATC and ATE. Its theoretical properties are developed in \Cref{sec:theory}.

\subsection{Identification via Transported Covariates}

We begin with the ATT, the Average Treatment Effect on the Treated. 
To correct the covariate shift between the control and treated populations, we use the Monge map of \Cref{subsec:transport-maps}, here taken in the direction that pushes the control covariate distribution forward to the treated one:
$$
T = \arg\inf_{T:  T_{\sharp} P_X^0 = P_X^1} \int c\bigl(x, T(x)\bigr)  dP_X^0(x).
$$
By Brenier's theorem, recalled in \Cref{subsec:transport-maps}, this map exists, is unique, and satisfies $T_{\sharp}  P_X^0 = P_X^1$ when $P_X^0 \in \calp_{p,\mathrm{ac}}(\calx)$ and $P_X^1 \in \calp_p(\calx)$.

\textbf{Existence of the Monge map.}
The existence requirement for the Monge map is mild. It suffices that both distributions have finite $p$-th moments and that $P_X^0$ admits a density. We take the cost $c(x,y) = \|x - y\|^p$ with the Euclidean norm and $p > 1$; the exponent $p>1$ makes the cost strictly convex, which is what Brenier's theorem requires. Under this cost, existence follows once the Wasserstein distance is finite,
$$
W_p(P_X^0, P_X^1)
 \leq 
\Big(\int \|x\|^p  \mathrm{d}P_X^0(x)\Big)^{1/p} + \Big(\int \|x\|^p  \mathrm{d}P_X^1(x)\Big)^{1/p}
 <  +\infty,
$$
or equivalently $P_X^0 \in \calp_{p,\mathrm{ac}}(\calx)$ and $P_X^1 \in \calp_p(\calx)$. Under this assumption the product coupling $P_X^0 \otimes P_X^1$ has finite cost, which guarantees existence. The general theory for Monge maps under other norms on $\R^d$ is discussed in \citet{champion2011monge}.

From now on we fix $p = 2$. Throughout the rest of the paper, $\calx \subseteq \R^d$ and the cost is squared Euclidean, $c(x,y) = \|x - y\|_2^2$. A brief recapitulation of optimal transport is provided in \Cref{sec:brief_recapitulation_to_optimal_transport}.

\textbf{Identification formula.}
Applying the change of variables $T_{\sharp}P_X^0=P_X^1$ to the ATT formula yields
$$
\begin{aligned}
    \tau_{ATT} &= \E_{P_X^1}[h_1(X)] - \E_{T_{\sharp}P_X^0}[h_0(X)]\\
    &= \E_{P_X^1}[h_1(X)] - \E_{P_X^0}[h_0(T(X))],\\
\end{aligned}
$$
where $h_a(x) := E[Y |X= x,A= a].$
The second line replaces integration under the treated covariate distribution with integration under the control covariate distribution, 
composed with the transport map. 

\textbf{Support overlap.}
For the transported identification formula to be valid, we need the conditional expectation $\E[Y \mid X = T(x),  A=0]$ 
to be well-defined. Since $T$ maps $\mathrm{supp}(P_X^0)$ into $\mathrm{supp}(P_X^1)$ for $P_X^0 \in \calp_{p,\mathrm{ac}}(\calx)$, 
this requires the outcome model under control to be evaluable at points in the treated support.
We formalize this requirement as follows.

\begin{ass}[Support Overlap]
\label{ass:supp_overlap}
For $P_X^0, P_X^1 \in \calp_{p,\mathrm{ac}}(\calx)$,
$$
\mathrm{supp}(P_X^0) = \mathrm{supp}(P_X^1).
$$
\end{ass}

For the ATT alone, only the inclusion $\mathrm{supp}(P_X^1) \subseteq \mathrm{supp}(P_X^0)$ is needed, so that $h_0$ is well-defined on the treated support. The ATC symmetrically needs only the reverse inclusion, and the ATE hence needs both inclusions. The equality in \Cref{ass:supp_overlap} covers all three estimands. This is exactly the overlap condition (\Cref{ass:positivity})
expressed in terms of supports.
Identification therefore requires no more than the overlap already assumed for IPW.
Stable IPW estimation, in contrast, relies on the strictly stronger condition of strict overlap (\Cref{ass:strong_positivity}),
which further demands a uniformly bounded density ratio.
Support overlap places no constraint on the relative magnitudes of the two densities, 
only that they share the same support. 
If the supports do not overlap, the regression model $h_0$ would need to extrapolate to covariate regions unobserved under control, 
and the validity of such extrapolation is not justifiable unless additional assumptions are made.
We defer to tackle the lack of overlap problem in \Cref{sec:interpolation}.

\textbf{Average Treatment Effect on the Control (ATC).} 
The same construction extends to other estimands like ATC. Let $S$ denote the inverse Monge map pushing $P_X^1$ back to $P_X^0$, $S_{\sharp} P_X^1=P_X^0$, which exists and is unique under the same conditions, since \Cref{ass:supp_overlap} requires $P_X^1$ to be absolutely continuous with finite $p$-th moments as well. Under \Cref{ass:SUTVA,ass:confoundedness,ass:supp_overlap}, the ATC is identified as
$$
\begin{aligned}
    \tau_{ATC} &=\E[Y(1) \mid A=0] - \E[Y(0) \mid A=0]\\
    &= \E_{S_{\sharp}P_X^1}[h_1(X)] - \E_{P_X^0}[h_0(X)]\\
    &= \E_{P_X^1}[h_1(S(X))] - \E_{P_X^0}[h_0(X)].
\end{aligned}
$$

\textbf{Average Treatment Effect (ATE).} The target population for the ATE is the entire population $P_X \in \mathcal{P}_{p,\mathrm{ac}}(\calx)$. Using both the transport map $T$ and its inverse $S$, and denoting $r = P(A=1)$, the ATE is identified under the assumptions for both ATT and ATC as
$$
\begin{aligned}
    \tau_{ATE} &=\E[Y(1)] - \E[Y(0)]\\
    &= \mathbb{E}[Y(1) \mid A=1]r + \mathbb{E}[Y(1) \mid A=0](1-r)\\
    &\quad - \mathbb{E}[Y(0) \mid A=1]r - \mathbb{E}[Y(0) \mid A=0](1-r) \\
    &= \mathbb{E}_{P_X^1}\mathbb{E}[Y\mid X, A=1]r + \mathbb{E}_{P_X^0}\mathbb{E}[Y \mid X,A=1](1-r)\\
    &\quad - \mathbb{E}_{P_X^1}\mathbb{E}[Y \mid X,A=0]r - \mathbb{E}_{P_X^0}\mathbb{E}[Y \mid X,A=0](1-r)\\
    &= \E_{P_X^1}[h_1(X)r + h_1(S(X))(1-r)] -
    \E_{P_X^0}[h_0(T(X))r + h_0(X)(1-r)]
\end{aligned}
$$

\subsection{Estimation with Potential Outcome Imputation}
We focus on the estimation of the ATT, as the same argument applies directly to both the ATC and the ATE.

From the identification formula derived above, we identify $\tau_{ATT}$ using
$$
\begin{aligned}
    \tau_{ATT} &= \mathbb{E}[Y(1) \mid A=1] - \mathbb{E}[Y(0) \mid A=1] \\
    &= \mathbb{E}[Y \mid A=1] - \E_{X\sim P_X^0}\mathbb{E}[Y \mid T(X), A=0]\\
    &:= \mathbb{E}_{P_Y^1}[Y] - \E_{P_X^0}[h_0 \circ T(X)],
\end{aligned}
$$
where we recall the notation $h_a(x) = \E[Y \mid X = x, A=a]$ for $a\in\{0,1\}$ and let $P^a_Y = P(Y \mid A=a)$. 

We call the resulting procedure the \emph{Stable Transport Estimator} (STE) because it replaces the density-ratio reweighting of IPW, which can be unstable when the propensity score is near zero, with the Monge map, whose plug-in estimator converges stably in $L_2$. We define two versions. The \emph{oracle} STE $\tau_{ATT}^{\ast}$, in which both $h_0$ and $T$ are known and only the statistical functionals are estimated from the sample, is
$$
\tau_{ATT}^{\ast} = \mathbb{E}_{P_{Y,n}^1}[Y] - \E_{P_{X,n}^0}[h_0 \circ T(X)],
$$
where we denote the empirical measure of covariates as $P_{X,n}^a:= \frac{1}{|\Omega_a|}\sum_{i\in\Omega_a}\delta_{X_i}$, and the empirical measure of observed outcome as $P_{Y,n}^a:= \frac{1}{|\Omega_a|}\sum_{i\in\Omega_a}\delta_{Y_i}$. 
The \emph{plug-in} STE $\hat{\tau}_{ATT}$, which substitutes estimated nuisance parameters $\hat{h}_0$ and $\hat{T}$, is
$$
\hat{\tau}_{ATT} = \mathbb{E}_{P_{Y,n}^1}[Y] - \E_{P_{X,n}^0}[\hat{h}_0\circ \hat{T}(X)].
$$
The plug-in STE accepts any consistent estimator $\hat T$ of the transport map. However, consistency and inference place different requirements on $\hat T$, which we discuss in \Cref{rem:plug-in-T}.

In simple settings, both nuisance parameters admit closed-form estimators. For instance, the following example assumes that the distributions are normal,
$P_X^1 = N(\mu_1,\Sigma_1)$ and $P_X^0= N(\mu_0, \Sigma_0)$. Between normal distributions, the optimal transport map can be characterized by a linear transformation.

\begin{example}[Gaussian Covariates and Linear Regression]
Consider the setting where $P_X^0 = N(\mu_0, \Sigma_0)$ and $P_X^1 = N(\mu_1, \Sigma_1)$. The transport map satisfying $T_{\sharp} P_X^0 = P_X^1$ has the closed form
$$
T(x) = \mu_1 + \Sigma_0^{-1 / 2}\left(\Sigma_0^{1 / 2} \Sigma_1 \Sigma_0^{1 / 2}\right)^{1 / 2} \Sigma_0^{-1 / 2}(x-\mu_0),
$$
which can be estimated via the plug-in estimators $\mu_a$ and $\Sigma_a$, i.e. $\hat{\mu}_a = \frac{1}{n_a}\sum_{i\in\Omega_a}X_i$, $\hat{\Sigma}_a = \frac{1}{n_a-1}\sum_{i\in\Omega_a}(X_i-\hat{\mu}_a)(X_i-\hat{\mu}_a)^T$:
$$
\hat{T}(x) = \hat{\mu}_1 + \hat{\Sigma}_0^{-1/2}\left(\hat\Sigma_0^{1 / 2} \hat\Sigma_1 \hat\Sigma_0^{1 / 2}\right)^{1 / 2}\hat{\Sigma}_{0}^{-1/2}(x-\hat{\mu}_0).
$$

When the outcome regression is linear in the covariates,
$$
h_0(x) = \mathbb{E}[Y \mid X=x, A=0] = \beta_0 + \beta_1^Tx:= \tilde{x}^T\beta,
$$
where $\tilde{x} = (1, x^T)^T \in \R^{d+1}$, and the observed outcome has additive independent Gaussian noise
$$
\quad Y_i = h_0(X_i) + \epsilon_i,\quad \epsilon_i \sim N(0,\sigma^2).
$$
Then $h_0$ admits the unbiased OLS estimator
$$
\hat{h}_0(x) = \tilde{x}^T \hat{\beta}, \qquad \hat{\beta} = (\tilde{X}_0^T \tilde{X}_0)^{-1} \tilde{X}_0^T Y_0,
$$
where $X_0\in\R^{n_0\times (d+1)}$ is the design matrix for the control group that stacks $\tilde{X}_i\in \Omega_0$ and $Y_0 \in \R^{n_0}$ the corresponding outcomes.
\end{example}
We revisit this example in \Cref{sec:role-of-map}. The Gaussian plug-in map satisfies the asymptotic-linearity conditions of \Cref{sec:theory}, and \Cref{cor:exact-map-cancellation} gives an exact identity for the contribution of the estimated map to the asymptotic variance.

\subsection{Cross-Fitting}\label{subsec:crossfit}

The STE requires first estimating the outcome regression $\hat{h}_a$ and the optimal map $\hat{T}$ before evaluating the statistical functional. Cross-fitting allows us to sidestep the bias that arises when the same data are used for both steps, by splitting the data into $K$ folds.

Partition the sample into $K$ equally sized folds $\Omega = \bigcup_{k=1}^K \Omega^{(k)}$ with $|\Omega^{(k)}| = n/K$ for all $k \in [K]$. For each fold $k$, the nuisance parameters $h_a$ and $T$ are estimated on the complement $\Omega^{(-k)} = \bigcup_{j \neq k} \Omega^{(j)}$, and the functional is evaluated on the held-out fold $\Omega^{(k)}$. As $K$ is fixed and finite with respect to the sample size, the cross-fitting procedure does not influence the asymptotics of the estimators.
 In the simplest case $K=2$, half of the sample is used for function estimation and the other half for functional evaluation.
 
The $k$-th fold ATT estimator is
$$
\hat{\tau}_{ATT}^{(k)} = \E_{P_{Y,n}^{1,(k)}}[Y] - \E_{P_{X,n}^{0,(k)}}\bigl[\hat{h}_0^{(-k)} \circ \hat{T}^{(-k)}\bigr],
$$
and the final estimator averages over all folds:
$$
\hat{\tau}_{ATT} = \frac{1}{K}\sum_{k=1}^K \hat{\tau}_{ATT}^{(k)},
$$
which reduces the variance relative to a single estimator computed on a smaller subsample. From this point on, $\hat{\tau}_{ATT}$ always denotes this cross-fitted estimator.

Cross-fitting serves a second purpose that is specific to the transport map. If the map or coupling is estimated on the same sample on which the functional is evaluated, the transport step has no effect on the estimator. Any feasible plan integrates the fitted regression to its treated-arm average, so the estimator reduces to the outcome-regression estimator on the treated arm. \Cref{prop:empirical-plan-cancellation} in \Cref{app:estimated-map-variance} makes this exact. The map must therefore be trained on held-out data for transport to contribute at all.

\section{Theoretical Properties of STE}
\label{sec:theory}

The identification formula of \Cref{sec:monge_OR} readily applies to any transport map: any $T$ with $T_{\sharp}P_X^0 = P_X^1$ yields
$\E_{P_X^0}[h_0(T(X))] = \E_{P_X^1}[h_0(X)]$. The estimator, however, is built from the plug-in composition $\hat{h}_0\circ \hat{T}$, and its bias and variance are governed by the scalar error of this composition integrated over fresh control covariates. In this section, we quantify how the estimated transport map enters that error, in four steps. \Cref{prop:honest-split-map-variance} gives an exact bias and variance decomposition under sample splitting. \Cref{thm:balance-error-regimes} derives an asymptotic expansion for the cross-fitted estimator and classifies its limiting regimes, from which consistency follows as \Cref{thm:att_consistency}. \Cref{sec:clt-conditions} states the conditions under which the map contribution is asymptotically linear and establishes the central limit theorem. \Cref{sec:role-of-map} examines the transport-map estimators considered in this paper and discusses which of these conditions each one satisfies. We focus on the ATT throughout; the properties of the ATC and ATE estimators follow by the same arguments under analogous assumptions on $h_1$ and $h_0$.

Throughout, the nuisance estimators $\hat h_0^{(-k)}$ and $\hat T^{(-k)}$ are trained on the held-out folds $\Omega^{(-k)}$. For each fold $k$, define
$$
g_k := \hat h_0^{(-k)} \circ \hat T^{(-k)},
\qquad
Q_k := \hat T^{(-k)}_{\sharp} P_X^0,
\qquad
B_k := P_X^0 g_k - P_X^1 h_0,
$$
where $P_X^a f := \E_{P_X^a}[f(X)]$ denotes the population functional. The conditional bias of the fold-$k$ estimator is $-B_k$, see \eqref{eq:conditional-bias}, so $B_k$ is the population-level error of the plug-in composition $g_k$. It admits the two exact decompositions
\begin{equation}
B_k
= \underbrace{(Q_k - P_X^1)h_0}_{=:B_{T,k}}
 + \underbrace{Q_k\bigl(\hat h_0^{(-k)} - h_0\bigr)}_{=:B_{h,k}}
= (Q_k - P_X^1)\hat h_0^{(-k)} + P_X^1\bigl(\hat h_0^{(-k)} - h_0\bigr),
\label{eq:fold-drift}
\end{equation}
which separates the contribution of the plug-in map, $B_{T,k}$, from that of the fitted regression, $B_{h,k}$. Write $\overline B_n := K^{-1}\sum_{k=1}^K B_k$.
Finally, for a candidate limit $g_{\star} \in L_2(P_X^0)$ of the compositions $g_k$, define the influence function with the nuisances held fixed at $g_{\star}$, evaluated at an observation $p = (a,x,y)$, as
\begin{equation}
\psi_{g_{\star}}(p)
:= \frac{a}{r}\bigl(y - \E_{P_Y^1}[Y]\bigr)
 - \frac{1-a}{1-r}\bigl(g_{\star}(x) - P_X^0 g_{\star}\bigr).
\label{eq:frozen-nuisance-if}
\end{equation}
When $g_{\star} = h_0 \circ T$ for an exact map $T$, this is precisely the influence function $\psi_{Y,h_0,T}$ derived in \Cref{app:proofs}, since $P_X^0(h_0\circ T) = \E_{P_X^0}[h_0(T(X))]$.

\subsection{Exact Decomposition under Sample Splitting}
\label{sec:estimated-map-variance}

We first quantify the impact of the plug-in map estimator on bias and variance in an idealized setting, where the nuisances are trained on an independent auxiliary sample.

\begin{prop}[Exact bias and variance under sample splitting]
\label{prop:honest-split-map-variance}
Let $\Omega_n$ be a training sample independent of an evaluation sample $\Omega^e_m$ drawn i.i.d.\ from $P$, and write $\calg_n := \sigma(\Omega_n)$ for the sigma algebra generated by the training sample. Let
$\hat h_0, \hat T$ be $\calg_n$-measurable with
$g_n := \hat h_0 \circ \hat T$, $Q_n := \hat T_{\sharp}P_X^0$, and
$B_n := P_X^0 g_n - P_X^1 h_0$. Consider
$$
\hat\tau_{ATT}
= \E_{P_{Y,m}^{1,e}}[Y] - \E_{P_{X,m}^{0,e}}[g_n],
$$
with $P^{1,e}_{Y,m}, P^{0,e}_{X,m}$ the arm-wise empirical measures on
$\Omega^e_m$.
Condition on $\calg_n$ and on the treatment indicators
$\{A_i\}_{i\in\Omega^e_m}$, with realized arm sizes $m_1, m_0 \ge 1$. Then
\begin{align}
\E\bigl(\hat\tau_{ATT} - \tau_{ATT} \mid \calg_n, \{A_i\}\bigr)
&= -B_n,
\label{eq:conditional-bias}\\
\operatorname{Var}\bigl(\hat\tau_{ATT} \mid \calg_n, \{A_i\}\bigr)
&= \frac{\operatorname{Var}(Y \mid A=1)}{m_1}
 + \frac{\operatorname{Var}_{P_X^0}\{g_n(X)\}}{m_0}\\
&= \frac{\operatorname{Var}(Y \mid A=1)}{m_1}
 + \frac{\operatorname{Var}_{Q_n}\{\hat h_0(X)\}}{m_0}.
\label{eq:conditional-variance}
\end{align}
Consequently, with expectations over both samples given $\{A_i\}$,
\begin{align}
\operatorname{Var}\bigl(\hat\tau_{ATT}\mid\{A_i\}\bigr)
&= \frac{\operatorname{Var}(Y \mid A=1)}{m_1}
 + \frac{\E\bigl[\operatorname{Var}_{Q_n}\{\hat h_0(X)\}\bigr]}{m_0}
 + \operatorname{Var}(B_n),
\label{eq:unconditional-variance}\\
\E\bigl[(\hat\tau_{ATT} - \tau_{ATT})^2\mid\{A_i\}\bigr]
&= \frac{\operatorname{Var}(Y \mid A=1)}{m_1}
 + \frac{\E\bigl[\operatorname{Var}_{Q_n}\{\hat h_0(X)\}\bigr]}{m_0}
 + \E(B_n^2).
\label{eq:unconditional-mse}
\end{align}
\end{prop}

The proposition separates the three effects of the estimated map. The conditional evaluation variance is driven by the random pushforward $Q_n$. Plug-in estimation shifts the conditional mean of the estimator by $-B_n$. The randomness of $B_n$ adds $\operatorname{Var}(B_n)$.
At an exact population endpoint, where $\hat h_0 = h_0$ and $\hat T = T$ with
$T_{\sharp}P_X^0 = P_X^1$, we have $B_n = 0$, and every exact map gives
$\operatorname{Var}_{P_X^0}\{h_0(T(X))\} = \operatorname{Var}_{P_X^1}\{h_0(X)\}$. The fixed-map oracle variance is therefore invariant to the selected transport.

\subsection{Asymptotic Regimes for the Cross-Fitted Estimator}
\label{sec:regimes}

Under cross-fitting the training folds overlap, and the exact identities above
acquire a remainder; \Cref{prop:crossfit-variance-identity} in
\Cref{app:estimated-map-variance} makes this exact. At first order, the nuisance estimators affect the limit only through $\overline B_n$.

\begin{theorem}[Asymptotic regimes of the cross-fitted STE]
\label{thm:balance-error-regimes}
Let \Cref{ass:SUTVA,ass:confoundedness,ass:supp_overlap} hold, let
$K \ge 2$ be fixed, and let $\hat\tau_{ATT}$ be the cross-fitted STE. Suppose
\begin{enumerate}
    \item there exists a deterministic $g_{\star}\in L_2(P_X^0)$ with
    $\max_{1\le k\le K}\|g_k - g_{\star}\|_{L_2(P_X^0)} = o_p(1)$;
    \item $\E_{P_Y^1}[Y^2] < +\infty$.
\end{enumerate}
Then the expansion
\begin{equation}
\sqrt n\,(\hat\tau_{ATT} - \tau_{ATT})
= \frac{1}{\sqrt n}\sum_{i=1}^n \psi_{g_{\star}}(p_i)
 - \sqrt n\,\overline B_n + o_p(1)
\label{eq:master-expansion}
\end{equation}
holds. The limit distribution is determined by the behaviour of $\overline B_n$, in the following four cases.
\begin{enumerate}
\item[(R1)] If $\overline B_n = o_p(n^{-1/2})$, then
$\sqrt n(\hat\tau_{ATT} - \tau_{ATT}) \Rightarrow N(0, V_{\star})$ with
$V_{\star} = \E[\psi_{g_{\star}}(p)^2]$. If moreover $g_{\star} = h_0\circ T$
with $T_{\sharp}P_X^0 = P_X^1$, then
\begin{equation}
V_{\star} = V_{\mathrm{or}}
:= \frac{\operatorname{Var}(Y\mid A=1)}{r}
 + \frac{\operatorname{Var}\{h_0(X)\mid A=1\}}{1-r},
\label{eq:oracle-variance}
\end{equation}
invariant to the selected exact map.
\item[(R2)] If $\sqrt n\,\overline B_n \to_p b$ for a constant $b$, then
$\sqrt n(\hat\tau_{ATT} - \tau_{ATT}) \Rightarrow N(-b, V_{\star})$: the
fitted nuisances can create a local bias without changing the variance
$V_{\star}$.
\item[(R3)] If $\overline B_n$ is asymptotically linear,
\begin{equation}
\overline B_n
= \frac{1}{n}\sum_{i=1}^n \psi_B(p_i) + o_p(n^{-1/2}),
\qquad \E[\psi_B(p)] = 0,
\label{eq:drift-al}
\end{equation}
then $\sqrt n(\hat\tau_{ATT} - \tau_{ATT}) \Rightarrow N(0, V)$ with
influence function $\psi_{g_{\star}} - \psi_B$ and
\begin{equation}
V = V_{\star} + \operatorname{Var}\{\psi_B(p)\}
 - 2\operatorname{Cov}\{\psi_{g_{\star}}(p), \psi_B(p)\}.
\label{eq:total-variance}
\end{equation}
A bias term of exact order $n^{-1/2}$ can therefore increase or
decrease the variance, depending on its covariance with
$\psi_{g_{\star}}$.
\item[(R4)] If $n^{\alpha}\overline B_n \to_p b \neq 0$ for some
$0 < \alpha < 1/2$, then
$n^{\alpha}(\hat\tau_{ATT} - \tau_{ATT}) \to_p -b$: the estimator remains
consistent but at a rate slower than root $n$. If
$\overline B_n \not\to_p 0$, the estimator is inconsistent for $\tau_{ATT}$
and, when it converges, targets the pseudo-true parameter
$\tau_{ATT} - \operatorname{plim}\overline B_n$.
\end{enumerate}
\end{theorem}

The theorem is deliberately stated in terms of the scalar quantity $\overline B_n$. A global
$L_2(P_X^0)$ rate on $\hat T^{(-k)} - T$ is sufficient but not necessary: if
$h_0$ is Lipschitz with constant $L_h$ and $T_{\sharp}P_X^0 = P_X^1$, then
\begin{equation}
\begin{aligned}
\bigl|(Q_k - P_X^1)h_0\bigr|
&= \bigl|P_X^0\{h_0(\hat T^{(-k)}(X)) - h_0(T(X))\}\bigr|\\
&\le L_h\|\hat T^{(-k)} - T\|_{L_2(P_X^0)},
\end{aligned}
\label{eq:l2-sufficient-bound}
\end{equation}
and the signed integral on the left can converge faster than the norm on the
right because of cancellation.
The regimes yield consistency of the cross-fitted STE under weak conditions on the nuisance estimators.

\begin{prop}[Consistency of the cross-fitted STE]
\label{thm:att_consistency}
The STE for the ATT under cross-fitting with fixed fold $K>0$ is consistent, namely, $\hat{\tau}_{ATT}\rightarrow_p \tau_{ATT}$ as $n \rightarrow \infty$, provided the following conditions are satisfied:
    \begin{enumerate}
        \item $L_2$ consistency of nuisance parameters: 
        \begin{enumerate}
            \item $\|\hat{h}^{(-k)}_0 - h_0\|_{L_2(P_X^1)} \to_p 0$ as $n\to\infty$;
            \item $\|\hat{T}^{(-k)} - T\|_{L_2(P_X^0)} \to_p 0$ as $n\to\infty$;
        \end{enumerate}
        \item Moment conditions:
        \begin{enumerate}
            \item $\|h_0\|_{L_2(P_X^1)}<+\infty$ and $\E_{P_Y^1}[Y^2]<+\infty$;
            \item  $\sup_k\E\|\hat{T}^{(-k)}-T\|_{L_2(P_X^0)}<+\infty$, $\sup_k\E\|\hat{h}_0^{(-k)} - h_0\|_{L_2(P_X^1)}<+\infty$;
        \end{enumerate}
        \item Lipschitz condition on the outcome regression estimator: $|\hat{h}^{(-k)}_0(x) - \hat{h}^{(-k)}_0(y)|\leq L^{(-k)}\|x-y\|$ for all $x,y\in\calx$, with $L^{(-k)} = O_p(1)$.
    \end{enumerate}
\end{prop}

Owing to the cross-fitting procedure, only $L_2$ consistency of the nuisance parameters is required, a condition satisfied by many standard estimators for regression and the Monge map. Using an independent sample for functional evaluation enables the application of the weak law of large numbers; without cross-fitting, uniform consistency over the entire domain $\calx$ would be needed. \Cref{lem:consistency-bridge} in the appendix verifies that conditions 1 to 3 imply hypothesis 1 of \Cref{thm:balance-error-regimes} with $g_{\star} = h_0\circ T$; then $\overline B_n \to_p 0$, so the proposition follows from the expansion \eqref{eq:master-expansion}. A direct proof under the same conditions is given in \Cref{app:proofs}.

\begin{remark}[Requirements of the plug-in map]
\label{rem:plug-in-T}
In the STE the transport map enters only as a plug-in nuisance parameter: $\hat\tau_{ATT}$ is well defined for any measurable $\hat T$, and consistency and inference impose different requirements on it.
Consistency (\Cref{thm:att_consistency}) requires only $L_2(P_X^0)$ consistency, $\|\hat T^{(-k)}-T\|_{L_2(P_X^0)}\xrightarrow{p}0$. This is attained by all standard constructions under their stated regularity conditions, including the barycentric projection and its $1$-nearest-neighbour extension \citep{manole2021plugin}, parametric plug-ins such as the Gaussian map, and entropic maps. For point estimation the choice of $\hat T$ is thus essentially unrestricted.

Statistical inference is more demanding. Nonparametric Brenier-map plug-ins attain $L_2(P_X^0)$ errors slower than $n^{-1/2}$, of order $n^{-1/4}$ up to logarithmic factors for $d\le 4$ and of order $n^{-1/d}$ for $d\ge 5$ \citep{manole2021plugin}, as discussed in \Cref{app:map-conditions}. These maps are also not known to admit the linear expansion that inference requires. \Cref{sec:clt-conditions} states the required asymptotic-linearity condition on $\hat T^{(-k)}$, which we call (A1), establishes a central limit theorem under it, and matches common map estimators to the conditions they satisfy. The empirical entropic map at fixed regularization satisfies the expansion around its population counterpart $T_\epsilon$. This expansion follows by composing the Hadamard differentiability of the entropic potentials \citep{goldfeld2024limit} with the functional delta method, and \Cref{app:1b-prime} carries out the argument.
\end{remark}

\subsection{Conditions and the Central Limit Theorem}
\label{sec:clt-conditions}

We now state the conditions for asymptotic normality. The conditions make explicit that the map estimator is asymptotically linear and centered at an exact map, so that it contributes to the variance but not to the first-order bias.

\begin{enumerate}
\item[(A1)] \emph{The transport-map estimator is asymptotically linear at a limit map $T^{\dagger}$.} There exists a measurable kernel
$\phi_T:\calx\times(\calx\times\{0,1\})\to\R^d$ such that for each fold $k$
and each $x\in\calx$,
\begin{equation}
\hat T^{(-k)}(x) - T^{\dagger}(x)
= \frac{1}{|\Omega^{(-k)}|}\sum_{j\in\Omega^{(-k)}}\phi_T(x; X_j, A_j)
 + r_n^{T,(k)}(x),
\label{eq:map-al}
\end{equation}
with remainder $\|r_n^{T,(k)}\|_{L_2(P_X^0)} = o_p(n^{-1/2})$. The kernel is mean zero and square integrable in the following sense. For $P_X^0$-almost every $x$, $\E[\phi_T(x; X, A)] = 0$, where the expectation is over an observation $(X,A)$ drawn from the population data law, and
$$
\int_{\calx}\E\|\phi_T(x; X, A)\|^2\, dP_X^0(x) < \infty.
$$
\item[(A2)] \emph{The limit map is exact:} $T^{\dagger} = T$ with $T_{\sharp}P_X^0 = P_X^1$. The map estimator is consistent for an exact transport map, so the map contributes no first-order bias.
\item[(A3)] \emph{The regression error is negligible at the root-$n$ scale:}
$\|\hat h_0^{(-k)} - h_0\|_{L_2(P_X^1)} = o_p(n^{-1/2})$.
\item[(A4)] \emph{Smoothness.} $h_0 \in C^{1,1}(\calx)$ with
$\sup_x\|\nabla h_0(x)\| \le L_h$ and
$\|\nabla h_0(x) - \nabla h_0(y)\| \le L^{\ast}_{\nabla h_0}\|x-y\|$.
Almost surely, $\hat h_0^{(-k)} \in C^{1,1}(\calx)$ with Lipschitz constant
$L^{(-k)} = O_p(1)$ and gradient Lipschitz constant $M^{(-k)} = O_p(1)$, and
$\|\nabla\hat h_0^{(-k)} - \nabla h_0\|_{L_2(P_X^1)} = o_p(1)$.
\item[(A5)] \emph{Moments:} $\E_{P_Y^1}[Y^2] < \infty$ and
$\E_{P_X^1}[h_0^2] < \infty$.
\end{enumerate}

Condition (A1) is the substantive restriction. It requires the map estimator to admit a root-$n$ linear expansion with a well-behaved kernel. For each of the transport-map estimators used in this paper, \Cref{app:map-conditions} verifies conditions (A1) and (A2) or explains why they fail. The integrated second moment is deliberately weaker than a supremum over $x$. The supremum form fails already for the Gaussian plug-in map, whose kernel is affine in $x$, while the integrated form holds whenever $P_X^0$ has finite second moments. Under (A1) and the $h_0$ part of (A4), the term $B_{T,k}$ admits a linear expansion.

\begin{corollary}[Asymptotic expansion of $B_{T,k}$]
\label{cor:al-transport}
Assume (A1) with limit map $T^{\dagger}$ and the $h_0$ smoothness of (A4). Define
$$
\psi_T(x,a) := \E_{P_X^0}\bigl[\nabla h_0(T^{\dagger}(X'))^{\top}
\phi_T(X'; x, a)\bigr],
\qquad
b^{\dagger} := P_X^0\{h_0(T^{\dagger}(X))\} - P_X^1 h_0,
$$
where $X' \sim P_X^0$ is an independent copy serving as the inner integration variable. Then
\begin{equation}
\frac{1}{K}\sum_{k=1}^K B_{T,k}
= b^{\dagger} + \frac{1}{n}\sum_{i=1}^n\psi_T(X_i, A_i) + o_p(n^{-1/2}),
\label{eq:map-drift-al}
\end{equation}
with $\psi_T$ mean zero and $\E[\psi_T^2]<\infty$. Under (A2) the limit map is exact and $b^{\dagger} = 0$. If
$T^{\dagger}$ is not exact, the expansion is centered at the pseudo-true parameter
$\tau^{\dagger} = \tau_{ATT} - b^{\dagger}$.
\end{corollary}

\begin{theorem}[CLT for the cross-fitted STE]
\label{thm:att_normality}
Let \Cref{ass:SUTVA,ass:confoundedness,ass:supp_overlap} hold, and let
$K \ge 2$ be fixed. Assume (A1), (A2), (A3), (A4) and (A5). Then
$$
\sqrt n\,(\hat\tau_{ATT} - \tau_{ATT}) \Rightarrow N(0, V),
$$
with influence function $\psi_{Y,h_0,T} - \psi_T$ and
\begin{equation}
V = \frac{\operatorname{Var}(Y\mid A=1)}{r}
 + \frac{\operatorname{Var}\{h_0(X)\mid A=1\}}{1-r}
 + \operatorname{Var}(\psi_T)
 - 2\operatorname{Cov}\bigl(\psi_{Y,h_0,T},\, \psi_T\bigr).
\label{eq:combined-variance}
\end{equation}
\end{theorem}

Under (A2) the estimated map contributes no first-order bias. Its estimation error contributes the variance terms
$\operatorname{Var}(\psi_T) - 2\operatorname{Cov}(\psi_{Y,h_0,T},\psi_T)$,
which can be of either sign. Assuming instead an $o_p(n^{-1/2})$ rate for
$\hat T^{(-k)}$ in $L_2(P_X^0)$ would erase $\psi_T$ from the expansion. An
$O_p(n^{-1/2})$ rate without the linear structure would keep
$\sqrt n(\hat\tau_{ATT} - \tau_{ATT})$ tight but would not identify a limit
distribution.

\begin{remark}[Root-$n$ outcome regressions]
\label{rem:root-n-regression}
Condition (A3) excludes regression estimators that converge at exactly
the root-$n$ rate, such as correctly specified OLS. Such estimators
contribute a second influence term $\psi_h$, symmetric to $\psi_T$, and the
CLT continues to hold with influence function
$\psi_{Y,h_0,T} - \psi_T - \psi_h$. We state the corresponding
asymptotic-linearity condition on the outcome regression and verify it for
OLS in \Cref{app:regression-al}. The main text stays focused on the
role of the transport map.
\end{remark}

\subsection{Discussion: Transport-Map Estimators}
\label{sec:role-of-map}

Conditions (A1) and (A2) are properties of the chosen transport-map estimator. This subsection discusses the estimators considered in this paper. For each estimator, we describe which conditions it satisfies and what this implies for the causal estimator.

\textbf{Direct in-sample estimation.}
The simplest choice estimates the transport on the same sample on which the functional is evaluated. In that case the estimator coincides with the outcome-regression estimator on the treated arm, whether or not the plan is optimal. Any feasible plan has the prescribed second marginal, so integrating $\hat h_0$ through the plan returns its treated-arm average, and the STE reduces to the treated-arm average of the residuals $Y_j - \hat h_0(X_j)$. \Cref{prop:empirical-plan-cancellation} in \Cref{app:estimated-map-variance} makes this exact, and the same marginal argument extends the conclusion to the ATC and the ATE. Converting the plan into a barycentric projection does not change this. The discrepancy is a Jensen gap, which vanishes for affine $\hat h_0$ and is otherwise bounded by the spread of the plan, see \eqref{eq:barycentric-gap-bound}. For transport to contribute, the map must therefore be trained on held-out data, as the cross-fitted STE of \Cref{subsec:crossfit} does.

\textbf{Out-of-sample nonparametric maps.}
Barycentric projections and their nearest-neighbour extensions can be trained on the held-out folds and evaluated at fresh controls. These estimators are $L_2(P_X^0)$ consistent, with rates of order $n^{-1/4}$ up to logarithmic factors for $d \le 4$ and of order $n^{-1/d}$ for $d \ge 5$ \citep{manole2021plugin}. Consistency of the causal estimator then follows from \Cref{thm:att_consistency}. No asymptotic-linearity expansion of the form (A1) is known for these maps, so \Cref{thm:att_normality} does not apply to them. Their behaviour is instead governed by the regimes of \Cref{thm:balance-error-regimes}. If the induced bias term is $o_p(n^{-1/2})$, regime (R1) gives root-$n$ normality without a $\psi_T$ term. Otherwise regime (R4) gives a slower rate. The bias term is a signed integral, so it can converge faster than the $L_2$ rate of the map itself because of cancellation, see \eqref{eq:l2-sufficient-bound}.

\textbf{Parametric plug-in maps.}
When the covariate laws are modeled parametrically, the transport map is a smooth functional of the model parameters, and the plug-in map inherits asymptotic linearity from the parameter estimates. Under correct specification, the Gaussian plug-in map satisfies (A1) and (A2), as verified in \Cref{app:prop:gaussian-a1}, so \Cref{thm:att_normality} applies and the map contributes the influence term $\psi_T$. For the plug-in of an exact-map functional with efficiently estimated parameters, the pushforward constraint determines $\psi_T$, and the asymptotic variance admits an exact identity relative to the oracle variance \eqref{eq:oracle-variance}, stated as \Cref{cor:exact-map-cancellation} in \Cref{app:estimated-map-variance}. Under a constant conditional treatment effect the identity gives a strict variance reduction, paralleling the classical gain from weighting by an estimated rather than the true propensity score \citep{hirano2003efficient, henmi2004paradox}. In general the sign of the net effect depends on the treatment effect function. The efficiency bound for the ATT is unchanged by the choice of map.

\textbf{Entropic maps at fixed regularization.}
The empirical entropic map at fixed $\epsilon > 0$ is asymptotically
linear around its population counterpart $T_{\epsilon}$. We verify this in
\Cref{app:1b-prime}, so condition (A1) holds with
$T^{\dagger} = T_{\epsilon}$. The limit map is not exact, because
$(T_{\epsilon})_{\sharp}P_X^0 \ne P_X^1$ in general. Condition (A2)
therefore fails, the offset
$b^{\dagger} = P_X^0\{h_0(T_{\epsilon}(X))\} - P_X^1 h_0$ is nonzero in
general, and the central limit theorem is centered at the pseudo-true parameter
$\tau_{ATT} - b^{\dagger}$. Inference centered at $\tau_{ATT}$ requires
$b^{\dagger} = o(n^{-1/2})$, which holds when $h_0$ is affine. A vanishing
schedule $\epsilon = \epsilon_n \to 0$
would require a triangular-array analysis, which we leave to future work.
The balanced entropic plan behaves differently. It has
the prescribed second marginal for every fixed $\epsilon$ by
\Cref{prop:entropic-ipw}, so the in-sample identity of
\Cref{prop:empirical-plan-cancellation} applies to it exactly and the
offset vanishes.

\textbf{What quadratic optimality contributes.}
\label{rem:quadratic-optimality-variance}
Quadratic optimality itself provides no first-order efficiency gain.
Every exact map
yields the same identification, the same oracle variance
\eqref{eq:oracle-variance}, and the same efficiency bound. The
contribution of optimality enters instead through the bias term. If $h_0$
is Lipschitz with constant $L_h$, the term $(Q_k - P_X^1)h_0$ obeys
\begin{equation}
\bigl|(Q_k - P_X^1)h_0\bigr|
\le L_h\, W_1(Q_k, P_X^1) \le L_h\, W_2(Q_k, P_X^1),
\label{eq:wasserstein-balance-bound}
\end{equation}
so the out-of-sample Wasserstein error of the fitted map bounds the bias
channel directly. If $\hat T^{(-k)}$ generalizes in Wasserstein distance
at a root-$n$ rate, the estimator falls in regime (R1). A slower
Wasserstein rate still bounds this term and yields regime (R4) at that
rate. Optimality on the training sample alone does not supply this
out-of-sample generalization. The optimal map also has the smallest
displacement $\|\mathrm{Id} - T\|_{L_2(P_X^0)}$ among exact maps. We use
this property in the extrapolation analysis of \Cref{sec:interpolation},
where the transport is interpolated with the identity to control
extrapolation. Optimality does not in general minimize the actual bias,
variance, or mean squared error of the causal estimator, which depend on
$h_0$ and on the selected coupling.

\section{Extrapolation-Controlled Stable Transport Estimator for Overlap Violation}
\label{sec:interpolation}

The stable transport estimator (STE) of \Cref{sec:monge_OR} evaluates the outcome regression $\hat{h}_0$ at transported samples $\hat{T}(X)$, which follow the target distribution $P_X^1$. When the overlap assumption holds, these transported samples remain within the training domain of $\hat{h}_0$, and the estimator is consistent as shown in \Cref{thm:att_consistency}. In practice, however, the overlap assumption is frequently violated: high-dimensional covariate spaces \citep{damour2021overlap} and small sample sizes can cause $\supp(P_X^1)$ to extend well beyond $\supp(P_X^0)$. In such settings, evaluating $\hat{h}_0$ at transported samples amounts to extrapolation outside the training domain, and the resulting bias can be substantial even when the estimator avoids the variance explosion typical of IPW-based methods \citep{crump2009dealing} or the sample loss incurred by caliper-based matching \citep{austin2011optimal, kormos2023asymptotics}.

To make this precise, we decompose the bias of the transport-adjusted ATT estimator,
\begin{equation}\label{eq:bias_ATT}
\begin{aligned}
    \mathbb{E}[\hat{\tau}_{ATT}] - \tau_{ATT} &=
    -(1-r)^n\mathbb{E}[Y|A=1] +r^n\mathbb{E}[h_0(T(X))|A=0]\\
    &\quad  + (1-r^n)\bigl[\mathbb{E}[h_0(T(X))|A=0] - \mathbb{E}[{h}_0(\hat{T}(X))|A=0]\\ 
    &\qquad + 
    \mathbb{E}[{h}_0(\hat{T}(X))|A=0]
    - \mathbb{E}[\hat{h}_0(\hat{T}(X))|A=0]\bigr]\\
    &\approx \underbrace{\mathbb{E}[h_0(T(X))-{h}_0(\hat{T}(X))|A=0]}_{\text{(a) Bias from $\hat{T}$}}
    +\underbrace{\mathbb{E}[{h}_0(\hat{T}(X))-\hat{h}_0(\hat{T}(X))|A=0]}_{\text{(b) Bias from $\hat{h}_0$}},
\end{aligned}
\end{equation}
where $h_a(x) := E[Y |X= x,A= a].$
The exact bias consists of two components. Term (a) measures the error from the estimated transport map $\hat{T}$, and term (b) measures the error from the estimated outcome regression $\hat{h}_0$. Terms (a) and (b) are the negatives of the map term $B_{T,k}$ and the regression term $B_{h,k}$ in \eqref{eq:fold-drift} of \Cref{sec:theory}, so the bias equals $-B_k$ up to remainders that vanish geometrically in $n$. Crucially, both terms are evaluated at the transported population $T(X) \sim P_X^1$ rather than at the training distribution $P_X^0$. When $\hat{h}_0$ is $L_2$-consistent over $P_X^1$, this poses no asymptotic difficulty, see \Cref{thm:att_consistency}. With finite samples and a misspecified model class, however, $\hat{h}_0$ may extrapolate poorly to $P_X^1$, and term (b) can dominate the estimation error.

Because $\hat{h}_0$ is trained on the control covariates $X \sim P_X^0$ while the full transport $T$ targets $P_X^1=T_{\sharp} P_X^0$, evaluating $\hat{h}_0 \circ T$ extrapolates the regression onto $\supp \left(P_X^1\right)$. We index the interpolation by the weight $\eta \in[0,1]$ on the identity, 
$$
T_\eta=\eta\,\mathrm{Id}+(1- \eta) T,
$$ 
and write $P_X^{\eta} := (T_{\eta})_{\sharp} P_X^0$ for the induced interpolated distribution. At $\eta=0$ the transport is full, $P_X^{\eta}=P_X^1$, and extrapolation is maximal. At $\eta=1$ evaluation stays on the training support, $P_X^{\eta}=P_X^0$, and there is no extrapolation. A larger $\eta$ thus trades imbalance correction for reduced extrapolation.

When an exact map exists, the interpolation admits an explicit population bias bound. If $T_{\sharp}P_X^0 = P_X^1$ and $h_0$ is Lipschitz with constant $L_h$, then
\begin{equation}
\bigl|P_X^0\{h_0(T_{\eta}(X))\} - P_X^1 h_0\bigr|
\le L_h\,\eta\,\bigl\{P_X^0\|X - T(X)\|^2\bigr\}^{1/2}
= L_h\,\eta\, W_2(P_X^0, P_X^1),
\label{eq:interpolation-displacement-bound}
\end{equation}
and the quadratic optimal map minimizes this bound among exact maps, since $W_2(P_X^0,P_X^1) = \min\{\|\mathrm{Id} - T\|_{L_2(P_X^0)} : T_{\sharp}P_X^0 = P_X^1\}$.
The idea of this interpolated adjustment is to transport the control covariates only partway toward the treated distribution. The evaluation points then stay close enough to the training domain to avoid severe extrapolation, while still correcting part of the covariate imbalance. For the region $D\subset \supp(P_X^1)\setminus\supp(P_X^0)$, which the control data never reach, the interpolation replaces the unattainable evaluation points with approximations that remain within reach of the training support. 

This formulation stands in contrast to existing approaches that assume the outcome regression is invariant under covariate shift \citep{oberst2020characterization, johansson2019support}. Such an assumption amounts to requiring that the outcome model is correctly specified everywhere, so that extrapolation error is zero by construction. In our setting, we make no such assumption; instead, we explicitly control the degree of extrapolation through the weight $\eta$ and characterize the resulting bias trade-off.

\subsection{Bias Decomposition and Trade-off}
\label{subsec:interpolation_prop}

We now formalize the bias trade-off governing the choice of $\eta$. We call the resulting estimator the Extrapolation-Controlled Stable Transport Estimator (EC-STE). Its cross-fitted form replaces the map $\hat{T}^{(-k)}$ in the STE of \Cref{subsec:crossfit} by the interpolated map $\hat{T}^{(-k)}_{\eta}=\eta\,\mathrm{Id} + (1-\eta) \hat{T}^{(-k)}$:
$$
\hat{\tau}_{ATT}^{\eta,(k)} = \E_{P_{Y,n}^{1,(k)}}[Y] - \E_{P_{X,n}^{0,(k)}}\bigl[\hat{h}_0^{(-k)} \circ \hat{T}^{(-k)}_{\eta}\bigr],
\qquad
\hat{\tau}_{ATT}^{\eta} = \frac{1}{K}\sum_{k=1}^K \hat{\tau}_{ATT}^{\eta,(k)}.
$$
Up to remainders that vanish geometrically in $n$, the bias satisfies, suppressing the fold index,
$$
\begin{aligned}
 \E[\hat{\tau}_{ATT}^{\eta}] - \tau_{ATT}
\approx& - \underbrace{\E[(\hat{h}_0-h_0)(\hat{T}_{\eta}(X))|A=0]}_{\text{(a) Extrapolation bias}}
- \underbrace{\E[h_0(\hat{T}_{\eta}(X)) - h_0(T_{\eta}(X))|A=0]}_{\text{(b) Estimation bias of $T$.}}\\
&- \underbrace{\E[h_0(T_{\eta}(X))-h_0(T(X))|A=0]}_{\text{(c) Approximation bias}}.\\
\end{aligned}
$$
The suppressed remainders involve the empty-arm probabilities $r^n$ and $(1-r)^n$, which vanish geometrically in $n$.

The decomposition highlights a fundamental trade-off among three sources of bias. 
Term (a), the extrapolation bias, captures how accurately $\hat{h}_0$ approximates $h_0$ at the evaluation points $\hat{T}_{\eta}(X)$; it increases with $\eta$ as these points move farther from the training distribution $P_X^0$, particularly under misspecification. 
Term (b), the transport estimation bias, reflects the error in approximating $T$ by $\hat{T}$ and is present for any $\eta > 0$. 
Term (c), the approximation bias, arises from evaluating $h_0$ under $P_X^{\eta}$ rather than $P_X^1$; it vanishes at $\eta = 0$ and increases as $\eta$ increases.

The choice of $\eta$ is thus driven by the tension between terms (a) and (c). 
When $\hat{h}_0$ is well specified or extrapolates reliably, extrapolation bias remains small and the approximation bias dominates, favoring $\eta = 0$ and the full transport estimator. 
When $\hat{h}_0$ is poorly specified, extrapolation bias becomes large for $\eta$ near $0$, and larger values of $\eta$ are preferred to restrict evaluation to regions supported by $P_X^0$. 
In the limiting case $\eta \to 1$, the estimator resembles a difference-in-means approach, eliminating extrapolation error at the cost of potentially large approximation bias due to uncorrected covariate imbalance. 
Intermediate values of $\eta$ balance these two effects, trading off robustness to misspecification against accuracy in adjusting for distributional shift. Term (c) is exactly the quantity controlled by the displacement bound \eqref{eq:interpolation-displacement-bound}. When $h_0$ is Lipschitz with constant $L_h$, the approximation bias is at most $L_h\,\eta\, W_2(P_X^0, P_X^1)$.

The EC-STE remains consistent under exactly the conditions of the uninterpolated estimator, provided the interpolation weight vanishes so that the evaluation distribution approaches the target distribution.

\begin{corollary}[Consistency of the EC-STE]\label{cor:att_consistency}
Let conditions 1 to 3 of \Cref{thm:att_consistency} hold, and let
$\eta = \eta_n \to 0$ as $n\to\infty$. Then the cross-fitted EC-STE is
consistent, $\hat{\tau}^{\eta}_{ATT} \to_p \tau_{ATT}$.
\end{corollary}

The corollary requires only $L_2$ consistency of the map estimator, so it
inherits the full scope of \Cref{thm:att_consistency}. In particular it
covers the barycentric projection and its nearest-neighbour extension, for
which no asymptotic-linearity expansion is known and the central limit
theory of \Cref{sec:clt-conditions} is unavailable. Consistency of the
EC-STE therefore costs nothing beyond the conditions of the uninterpolated
estimator. The only genuinely new requirement is that the interpolation
weight vanishes. If $\eta$ is held fixed, the approximation bias, term (c)
above, persists, and consistency for $\tau_{ATT}$ fails. In practice, any consistent map estimator therefore yields a consistent
EC-STE under near-violation of overlap. The weight $\eta$ controls
extrapolation in finite samples, and \Cref{cor:interpolated-ste} below
strengthens the guarantee to a central limit theorem when the map
satisfies condition (A1).

Under the asymptotic-linearity conditions of \Cref{sec:clt-conditions}, the decomposition \eqref{eq:fold-drift} of \Cref{sec:theory} quantifies the effect of the interpolation exactly. The interpolated estimator integrates the regression error against $P_X^{\eta}$ instead of $P_X^1$, and adds the deterministic imputation bias
$b_{\eta} := (P_X^{\eta} - P_X^1)h_0$, which equals term (c) above and is bounded by \eqref{eq:interpolation-displacement-bound}. \Cref{cor:interpolated-ste} shows that $b_{\eta}$ is linear in $\eta$ to first order, with an estimable coefficient, so the interpolation can be tuned or debiased without changing the central limit theorem or its asymptotic variance. Define the first-order bias coefficient
\begin{equation}
c_1 := \E_{P_X^0}\bigl[\nabla h_0(T(X))^{\top}\bigl(X - T(X)\bigr)\bigr],
\label{eq:interp-bias-coefficient}
\end{equation}
which is the derivative of the imputed mean along the displacement field toward the identity, and its cross-fitted plug-in
\begin{equation}
\hat c_1 := \frac{1}{K}\sum_{k=1}^K
P_{X,n}^{0,(k)}\Bigl[\nabla\hat h_0^{(-k)}\bigl(\hat T^{(-k)}(X)\bigr)^{\top}
\bigl(X - \hat T^{(-k)}(X)\bigr)\Bigr].
\label{eq:interp-bias-plugin}
\end{equation}

\begin{corollary}[CLT for the interpolated STE]
\label{cor:interpolated-ste}
Let the conditions of \Cref{thm:att_normality} hold, namely (A1) with the exact limit map of (A2), together with (A3), (A4) and (A5). Let
$W_2 := W_2(P_X^0, P_X^1) < \infty$, and let $\eta_n \in [0,1]$ be a
deterministic sequence with $\eta_n \to 0$. Let
$\hat\tau_{ATT}^{\eta}$ denote the cross-fitted EC-STE built from
$g_{k,\eta} := \hat h_0^{(-k)} \circ \hat T^{(-k)}_{\eta_n}$, and let $V$ be
the asymptotic variance \eqref{eq:combined-variance}.
\begin{enumerate}
\item[(i)] \emph{Oracle regime.} If $\eta_n = o(n^{-1/2})$, then
$$
\sqrt n\,\bigl(\hat\tau_{ATT}^{\eta} - \tau_{ATT}\bigr) \Rightarrow N(0, V).
$$
The interpolation has no first-order effect, and the limit coincides with
that of the uninterpolated estimator.
\item[(ii)] \emph{Local-bias regime.} If
$\sqrt n\,\eta_n \to \kappa \in [0, \infty)$, then
$$
\sqrt n\,\bigl(\hat\tau_{ATT}^{\eta} - \tau_{ATT}\bigr)
\Rightarrow N(-\kappa\, c_1,\; V).
$$
Moreover $\hat c_1 \to_p c_1$, so the recentered estimator
$\hat\tau_{ATT}^{\eta} + \eta_n \hat c_1$ satisfies
$\sqrt n(\hat\tau_{ATT}^{\eta} + \eta_n\hat c_1 - \tau_{ATT}) \Rightarrow
N(0, V)$. Confidence intervals in this regime require only consistency of
$\hat c_1$.
\item[(iii)] \emph{Debiased interpolation.} Define
$\hat\tau^{\,dc}_{ATT,\eta} := \hat\tau_{ATT}^{\eta} + \eta_n\,\hat c_1$. If
$\eta_n = o(n^{-1/4})$ and the product condition
\begin{equation}
\eta_n \cdot \max_{1\le k\le K}
\bigl\|\nabla\hat h_0^{(-k)} - \nabla h_0\bigr\|_{L_2(P_X^1)}
= o_p(n^{-1/2})
\label{eq:interp-product-condition}
\end{equation}
holds, then
$$
\sqrt n\,\bigl(\hat\tau^{\,dc}_{ATT,\eta} - \tau_{ATT}\bigr)
\Rightarrow N(0, V).
$$
\end{enumerate}
In all three regimes the asymptotic variance is the unchanged $V$ of
\Cref{thm:att_normality}. The map contributes its variance through
$(1-\eta_n)\psi_T \to \psi_T$.
\end{corollary}

\begin{remark}[Interpretation]
\label{rem:interpolated-ste-reading}
(a) \emph{Effect of the interpolation.} For the interpolated estimator,
the quantity $B_{k,\eta}$ decomposes exactly as
$$
B_{k,\eta} = b_{\eta} + (Q_{k,\eta} - P_X^{\eta})h_0
 + Q_{k,\eta}\bigl(\hat h_0^{(-k)} - h_0\bigr),
\qquad
b_{\eta} = \eta\,c_1 + R_{\eta},
\quad |R_{\eta}| \le \tfrac{L^{\ast}_{\nabla h_0}}{2}\,\eta^2 W_2^2,
$$
where $Q_{k,\eta} := (\hat T^{(-k)}_{\eta})_{\sharp}P_X^0$. The
interpolation thus converts part of the statistical estimation error
into a deterministic bias that is bounded by
\eqref{eq:interpolation-displacement-bound}, linear to first order, and
estimable by \eqref{eq:interp-bias-plugin}. Under (A3) the regression term stays asymptotically negligible at every
$\eta$, while the measure it is integrated against moves from $P_X^1$
toward $P_X^0$, which is how the interpolation reduces extrapolation. For regressions
converging at exactly the root-$n$ rate, the same expansion holds with the
additional influence term $\psi_h$ of \Cref{app:regression-al}.
(b) \emph{Choice of $\eta_n$.} In regime (ii), $\kappa$ is an explicit
trade-off parameter with known local bias $\kappa c_1$. Regime (iii) allows a
larger interpolation weight, $\eta_n = o(n^{-1/4})$ instead of
$o(n^{-1/2})$, in exchange for the product condition
\eqref{eq:interp-product-condition} on the regression gradient rate. For
instance, a gradient rate of $o_p(n^{-1/4})$ suffices when
$\eta_n \asymp n^{-1/4}$.
(c) \emph{Affine case.} If $h_0(x) = \beta_0 + \beta_1^{\top}x$, then
$c_1 = \beta_1^{\top}(\mu_0 - \mu_1)$ with
$\mu_a = \E[X\mid A=a]$, $R_{\eta} \equiv 0$, and
$b_{\eta} = \eta\,c_1$ exactly for every $\eta$.
(d) \emph{Relation to the consistency corollary.} A fixed $\eta$ places the
estimator in regime (R4) of \Cref{thm:balance-error-regimes} with
pseudo-true parameter $\tau_{ATT} - b_{\eta}$. \Cref{cor:att_consistency}
establishes consistency under only $L_2$ consistency of the map estimator;
under the stronger condition (A1), the present corollary strengthens
the conclusion to a central limit theorem with explicit and removable
bias.
\end{remark}

\subsection{Selecting the Interpolation Parameter} 

The bias decomposition above shows that the optimal $\eta$ depends on the unknown extrapolation error $\E[(\hat{h}_0 - h_0)(\hat{T}_\eta(X)) | A=0]$. Since $h_0$ is unobserved, this quantity cannot be computed directly. Instead, we adopt a distribution-based, data-driven strategy. We select $\eta$ by controlling the distributional discrepancy between the training domain $P_X^0$ and the evaluation domain $P_X^{\eta}$. The rationale is that if $P_X^{\eta}$ remains close to $P_X^0$, the outcome regression $\hat{h}_0$ is evaluated near its training data, limiting extrapolation regardless of model misspecification.
 
Formally, under the assumption that $P_X^0, P_X^1, P_X^{\eta} \in \calp_{p,ac}(\calx)$, we select $\eta$ by solving a one-dimensional program
\begin{equation}\label{eq:eta_condition}
\begin{aligned}
    &\max_{\eta} u(\eta)\\
    &\text{s.t.} \eta\in[0,1],  d(P_X^0, P_X^{\eta}) \leq c
\end{aligned}
\end{equation}
for a chosen discrepancy measure $d$, a tolerance $c > 0$, and a utility function $u(\eta)$. The natural choice $u(\eta) = -\eta$ selects the smallest feasible $\eta$, which brings the evaluation domain as close to the target $P_X^1$ as the constraint allows.

The choice of discrepancy measure $d$ determines the type of distributional deviation being controlled. We consider four classes, ordered from weakest to strongest geometric assumptions, and for each we establish how the bound on $|g(0) - g(\eta)|$ relates to the constraint on $d(P_X^0, P_X^{\eta})$. Here, $f(x) = |\hat{h}_0(x) - h_0(x)|$ denotes the pointwise estimation error, $g(\eta) = \E_{X\sim P_X^{\eta}}[f(X)]$ is the expected estimation error under the interpolated distribution, and $g(0) = \E_{X\sim P_X^0}[f(X)]$ is the in-distribution estimation error. The quantity $|g(0) - g(\eta)|$ measures how much additional bias is incurred by evaluating at $P_X^{\eta}$ rather than $P_X^0$.

\textbf{Integral probability metrics (IPMs).} 
IPMs impose the weakest assumptions, requiring only that the estimation error $f$ lie in a test-function class $\calf$. For a symmetric class $\calf$, the associated IPM is
$$
d_\calf(P_X^0,P_X^\eta):=\sup_{f\in\calf}\Bigl|\int f(x)d(P_X^0-P_X^\eta)(x)\Bigr|.
$$
\begin{prop}\label{prop:ipm}
Assume $f\in\calf$ for a symmetric test-function class $\calf$. Then the bias incurred by evaluating at $P_X^\eta$ rather than $P_X^0$ is controlled by the IPM,
$$
|g(0)-g(\eta)|\le d_\calf(P_X^0,P_X^\eta),
$$
so the $\hat\eta$ obtained from \eqref{eq:eta_condition} with $d=d_\calf$ satisfies $|g(0)-g(\hat\eta)|\le c$.
\end{prop}
\begin{proof}
Since $f\in\calf$,
$$
|g(0)-g(\eta)|=\Bigl|\int fd(P_X^0-P_X^\eta)\Bigr|
\le \sup_{f\in\calf}\Bigl|\int fd(P_X^0-P_X^\eta)\Bigr|
= d_\calf(P_X^0,P_X^\eta),
$$
and feasibility of $\hat\eta$ in \eqref{eq:eta_condition} gives $d_\calf(P_X^0,P_X^{\hat\eta})\le c$.
\end{proof}

For a generic class $\calf$, the map $\eta\mapsto d_\calf(P_X^0,P_X^\eta)$ need not be linear or even monotone in $\eta$, since the optimal witness function may itself vary with $\eta$. The feasible set $\{\eta\in[0,1]:d_\calf(P_X^0,P_X^\eta)\le c\}$ is therefore not necessarily an interval, and the selection is carried out numerically, as discussed under implementation below.

The Lipschitz class is the leading special case. Take $\calf=\{f:\mathrm{Lip}(f)\le 1\}$, so that $d_\calf=W_1$ by Kantorovich--Rubinstein duality, and let $T$ be the quadratic-cost optimal map fixed in \Cref{sec:monge_OR}. Because $T$ is optimal, McCann's displacement interpolation makes $T_\eta$ the quadratic-cost optimal map between $P_X^0$ and $P_X^\eta$, and the curve $\eta\mapsto P_X^\eta$ is a constant-speed geodesic in $(\calp_2(\calx),W_2)$. This gives the exact scaling
$$
W_2(P_X^0,P_X^\eta)=(1-\eta)W_2(P_X^0,P_X^1),
\qquad
W_1(P_X^0,P_X^\eta)\le W_2(P_X^0,P_X^\eta).
$$
Imposing the computable constraint $W_2(P_X^0,P_X^\eta)\le c$ therefore enforces the IPM bound of \Cref{prop:ipm} and admits the closed-form maximizer
$$
\eta^\star=1-\min\Bigl(1,\ \frac{c}{W_2(P_X^0,P_X^1)}\Bigr),
$$
which is nontrivial with $\eta^\star<1$ precisely when $c<W_2(P_X^0,P_X^1)$.

\textbf{$W_p$ distance. } When the transport map $T$ is defined under the $W_p$ geometry, the constant-speed geodesic property provides $W_p(P_X^0, P_X^{\eta}) = (1-\eta) W_p(P_X^0, P_X^1)$, giving a structure analogous to the IPM case but under the stronger assumption that $f$ is Lipschitz. In addition, the Jensen inequality step loosens the bound from $W_1$ to $W_p$.

In particular, the Wasserstein distance induces a weak geometry on $\mathcal{P}_p(\calx)$, requiring only finite $p$-th moments. When $T$ is constructed under this geometry, the interpolated distribution $P_X^{\eta}$ evolves linearly in $W_p$, so its distance from $P_X^0$ is fully characterized by the interpolation parameter $\eta$ alone.
\begin{prop}
Assume $f$ is $L$-Lipschitz, $P_X^0,P_X^1\in\calp_p(\calx)$, and $T_{\sharp} P_X^0=P_X^1$ with $T$ the optimal map under the cost $\|\cdot\|^p$. With $W_p$ the Wasserstein distance of order $p$,
$$
|g(0)-g(\eta)|\le L W_p(P_X^0,P_X^\eta)=L(1-\eta)W_p(P_X^0,P_X^1).
$$
\end{prop}
\begin{proof}
By the Lipschitz property of $f$ followed by Jensen's inequality, which yields $L^1\le L^p$ under any probability measure,
$$
|g(0)-g(\eta)|\le \int|f(x)-f(T_\eta(x))|dP_X^0
\le L\int\|x-T_\eta(x)\|dP_X^0
\le L\Bigl(\int\|x-T_\eta(x)\|^p dP_X^0\Bigr)^{1/p}.
$$
Since $T$ is optimal, McCann's displacement interpolation makes $T_\eta= \eta\,\mathrm{Id}+(1-\eta) T$ the optimal map between $P_X^0$ and $P_X^\eta$, so the last quantity equals $L W_p(P_X^0,P_X^\eta)$. The same constant-speed geodesic property gives $W_p(P_X^0,P_X^\eta)=(1-\eta) W_p(P_X^0,P_X^1)$, and \eqref{eq:eta_condition} yields $\eta=1-\min\bigl(1,c/W_p(P_X^0,P_X^1)\bigr)$.
\end{proof}

\textbf{$\chi^2$ distance.} 
Moving to stronger geometric assumptions, the $\chi^2$ divergence requires the existence of the density ratio $dP_X^0/dP_X^{\eta}$. It provides a tighter bound when overlap is nearly violated. In that case the density ratio exists but is large.
\begin{prop}\label{prop:chisq}
Assume $P_X^0\ll P_X^\eta$ and $C:=\sup_{\eta\in[0,1]}\|f\|_{L_2(P_X^\eta)}<\infty$. Then for any $\eta$,
$$
|g(0)-g(\eta)|\le \|f\|_{L_2(P_X^\eta)}\cdot \chi^2(P_X^0\|P_X^\eta)^{1/2}
\le C\cdot \chi^2(P_X^0\|P_X^\eta)^{1/2},
$$
and the selected $\hat\eta$ from \eqref{eq:eta_condition} (with $d=\chi^2$, tolerance $c$) satisfies $|g(0)-g(\hat\eta)|\le Cc^{1/2}$.
\end{prop}
\begin{proof}
By Cauchy--Schwarz in $L_2(P_X^\eta)$,
$$
\begin{aligned}
|g(0)-g(\eta)|&=\Bigl|\int f\tfrac{d(P_X^0-P_X^\eta)}{dP_X^\eta}dP_X^\eta\Bigr|
\le \Bigl(\int f^2dP_X^\eta\Bigr)^{1/2}\Bigl(\int\tfrac{(dP_X^0-dP_X^\eta)^2}{dP_X^\eta}\Bigr)^{1/2}\\
&=\|f\|_{L_2(P_X^\eta)}\chi^2(P_X^0\|P_X^\eta)^{1/2}.
\end{aligned}
$$
Bounding $\|f\|_{L_2(P_X^\eta)}\le C$ and $\chi^2(P_X^0\|P_X^{\hat\eta})\le c$ gives $|g(0)-g(\hat\eta)|\le Cc^{1/2}$.
\end{proof}

\textbf{Hellinger distance.} 
When overlap is violated, the density ratio does not exist and the $\chi^2$ bound is not applicable. The Hellinger distance provides an alternative that requires only the existence of densities, not their ratio, making it better suited to settings with complete support mismatch.
\begin{prop}
Assume densities exist for $P_X^0,P_X^\eta$ and $\sup_x|f(x)|\le C$. With the Hellinger distance
$H(P_X^0,P_X^\eta):=\bigl(\tfrac12\int(\sqrt{dP_X^0}-\sqrt{dP_X^\eta})^2\bigr)^{1/2}$,
$$
|g(0)-g(\eta)|\le 2\sqrt2C\cdot H(P_X^0,P_X^\eta).
$$
\end{prop}
\begin{proof}
With $\mathrm{TV}(P_X^0,P_X^\eta)=\tfrac12\int|dP_X^0-dP_X^\eta|$,
$$
|g(0)-g(\eta)|\le \sup_x|f(x)|\int|dP_X^0-dP_X^\eta|
= 2C\mathrm{TV}(P_X^0,P_X^\eta).
$$
Because the Hellinger distance upper-bounds total variation, $\mathrm{TV}\le\sqrt2H$,
$$
|g(0)-g(\eta)|\le 2C\sqrt2H(P_X^0,P_X^\eta)=2\sqrt2CH(P_X^0,P_X^\eta).
$$
Under $H(P_X^0,P_X^{\hat\eta})\le c$ this gives $|g(0)-g(\hat\eta)|\le 2\sqrt2Cc$.
\end{proof}

\textbf{Summary.} Across all examples of discrepancy measures, bounding $d(P_X^0, P_X^{\eta}) \leq c$ ensures that the additional bias from evaluating at $P_X^{\eta}$ rather than $P_X^0$ is controlled. The measures differ in their assumptions, namely test-function class membership, Lipschitz regularity, existence of the density ratio, or existence of densities, and in the tightness of the resulting bound. In settings where the severity of overlap violation is unknown, the IPM and $W_p$ bounds offer the safest choices, as they impose the weakest conditions. When the density ratio is known to exist but may be large, the $\chi^2$ bound provides a sharper constraint. When support mismatch is complete, the Hellinger bound applies without requiring a density ratio.

\textbf{Implementation.}
Solving \eqref{eq:eta_condition} in practice requires evaluating the chosen discrepancy along the interpolation path. Under the Wasserstein constraint, the closed form $\eta^\star$ above applies directly. For a generic discrepancy the feasible set need not be an interval, and we maximize the utility by grid search over $\eta$. The $\chi^2$ constraint requires an estimate of the density ratio. When the covariate distributions are Gaussian with a common covariance, a closed-form solution for $\eta$ is available. In low dimensions, $d\leq 4$, a plug-in kernel density estimator estimates the density ratio well. In higher dimensions we recommend uLSIF, unconstrained least-squares importance fitting, which approximates the density ratio directly by minimizing a squared error; we implement uLSIF with an RKHS function class and a Gaussian kernel. The guarantees of \Cref{cor:interpolated-ste} are stated for deterministic sequences $\eta_n$. Extending them to the data-driven $\hat\eta$ requires uniformity in $\eta$, which we do not pursue.

\textbf{Oracle estimator.}
As a theoretical benchmark, we consider the case where the true outcome regression $h_0$ can be used for selecting $\eta$. This allows direct optimization of the bias trade-off rather than the indirect control through distributional discrepancy. Assuming that $h_0$ and $\hat{h}_0$ are first-order differentiable and that the estimation bias of $\hat{T}$ is negligible, a Taylor expansion yields the first-order approximation
$$
\begin{aligned}
    \mathrm{Bias}(\hat{\tau}_{ATT}^{\eta}) 
    &\approx  - \E[(\hat{h}_0-h_0)(\hat{T}_{\eta}(X))\mid A=0] - \E[h_0(\hat{T}_{\eta}(X))-h_0(\hat{T}(X))\mid A=0]\\
    &= - \E[(\hat{h}_0 - h_0)(\hat{T}(X)) + \eta(X-\hat{T}(X))^{\top}(\nabla\hat{h}_0 - \nabla h_0)(\hat{T}(X)) \mid A=0]\\
    &\quad -\E[\eta(X-\hat{T}(X))^{\top}\nabla h_0(\hat{T}(X))\mid A=0]
    + o\bigl(\eta\,\E[\|X-\hat{T}(X)\|\mid A=0]\bigr)\\
    &= -\E[(\hat{h}_0 - h_0)(\hat{T}(X)) + \eta(X-\hat{T}(X))^{\top}\nabla\hat{h}_0(\hat{T}(X))\mid A=0]\\
    & \quad + o\bigl(\eta\,\E[\|X-\hat{T}(X)\|\mid A=0]\bigr).
\end{aligned}
$$
Therefore, we could correct for the zeroth and first order bias with the oracle $\eta^{\ast}$,
$$
\eta^{\ast} =
\frac{\E[(h_0-\hat{h}_0)(\hat{T}(X))\mid A=0]}{\E[(X-\hat{T}(X))^{\top}\nabla\hat{h}_0(\hat{T}(X))\mid A=0]}.
$$
This quantity is not guaranteed to lie in $[0,1]$ and is not available in practice, but it serves as a useful baseline in simulations for evaluating how close the data-driven selection of $\eta$ comes to the optimal trade-off. When overlap is sufficient, both the numerator and the displacement $\|X - \hat{T}(X)\|$ become small, and $\eta^*$ exhibits high variability.

\section{Empirical Studies}
\label{sec:sim}

We evaluate the Stable Transport Estimator (STE) and its Extrapolation-Controlled variant (EC-STE). We compare them with standard causal inference methods: inverse propensity weighting (IPW), stabilized IPW (SIPW), the doubly robust augmented IPW (AIPW), difference in means (DiM), and outcome regression (OR). 

We consider four settings of increasing complexity: (1) simulated examples with Normal and Beta covariates, (2) the Lalonde--PSID job training dataset, (3) a semi-synthetic experiment based on the Right Heart Catheterization (RHC) study, and (4) the Medical Expenditure Panel Survey (MEPS). For EC-STE, we report variants using data-driven selection of $\hat{\eta}$ based on $W_1$, $W_2$, and $\chi^2$ divergence, alongside the grid-search oracle benchmark ($\eta_{gs}^{\ast}$); the simulation figures additionally show the first-order oracle $\eta^{\ast}$ of \Cref{sec:interpolation}. Throughout, we employ the sliced Wasserstein distance \citep{bonneel2015sliced} for covariate spaces of dimension $d > 3$, estimate transport maps via the exact linear solver, and use the nearest-neighbor estimator for map evaluation.\footnote{Code reproducing all experiments is available at \url{https://github.com/edenx/CausalSTE}.}

Under misspecification and overlap violation, appropriately selected interpolation reduces bias relative to outcome regression while maintaining substantially lower variance than IPW-based alternatives. When the outcome model is well specified, as in the Lalonde study, the interpolation neither helps nor hurts.

\subsection{Simulations: Unbounded Density Ratio}

To illustrate how the density ratio between treatment groups affects estimator performance, we consider two simulation designs: one with Beta-distributed covariates on compact support, and one with Gaussian covariates on unbounded support. In both cases, we vary distributional parameters to transition the density ratio from bounded to unbounded, exposing the failure modes of existing methods.

Recall from \Cref{sec:interpolation} the interpolated transport map $T_{\eta}(X) = \eta X + (1-\eta)T(X)$ with $P_X^{\eta} = (T_{\eta})_{\sharp} P_X^0$. At $\eta = 0$ it recovers the full transport map $T$ with $T_{\sharp} P_X^0=P_X^1$, and at $\eta = 1$ no transport is applied. When the covariate density is highly skewed with mass concentrated away from the origin, the intercept of the outcome regression can be poorly estimated. In such cases the interpolated map corrects for extrapolation bias in the ATT estimation by keeping the transported covariates within the support of the regression training data.

\begin{figure}[t]
    \centering
    
    \begin{subfigure}{0.48\linewidth}
        \centering
        \includegraphics[width=\linewidth]{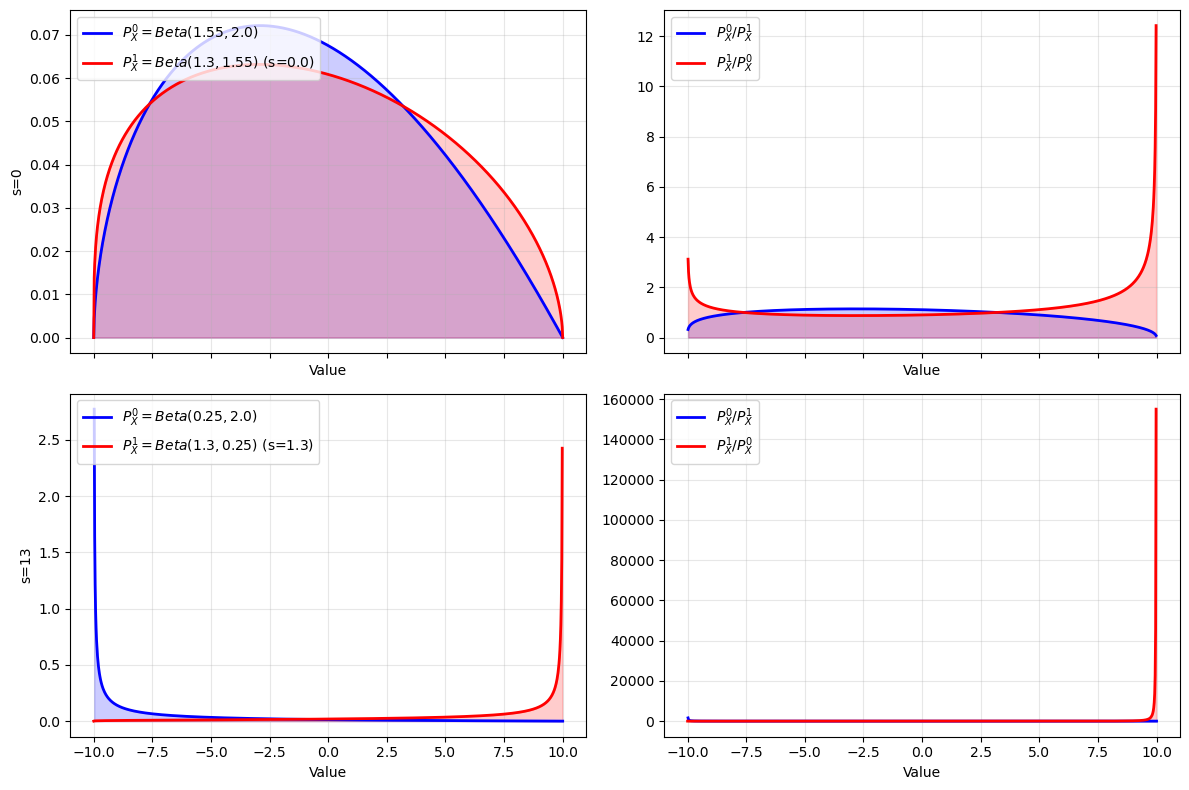}
        \caption{Beta example.}
        \label{fig:beta_scale}
    \end{subfigure}
    \hfill  
    \begin{subfigure}{0.48\linewidth}
        \centering
        \includegraphics[width=\linewidth]{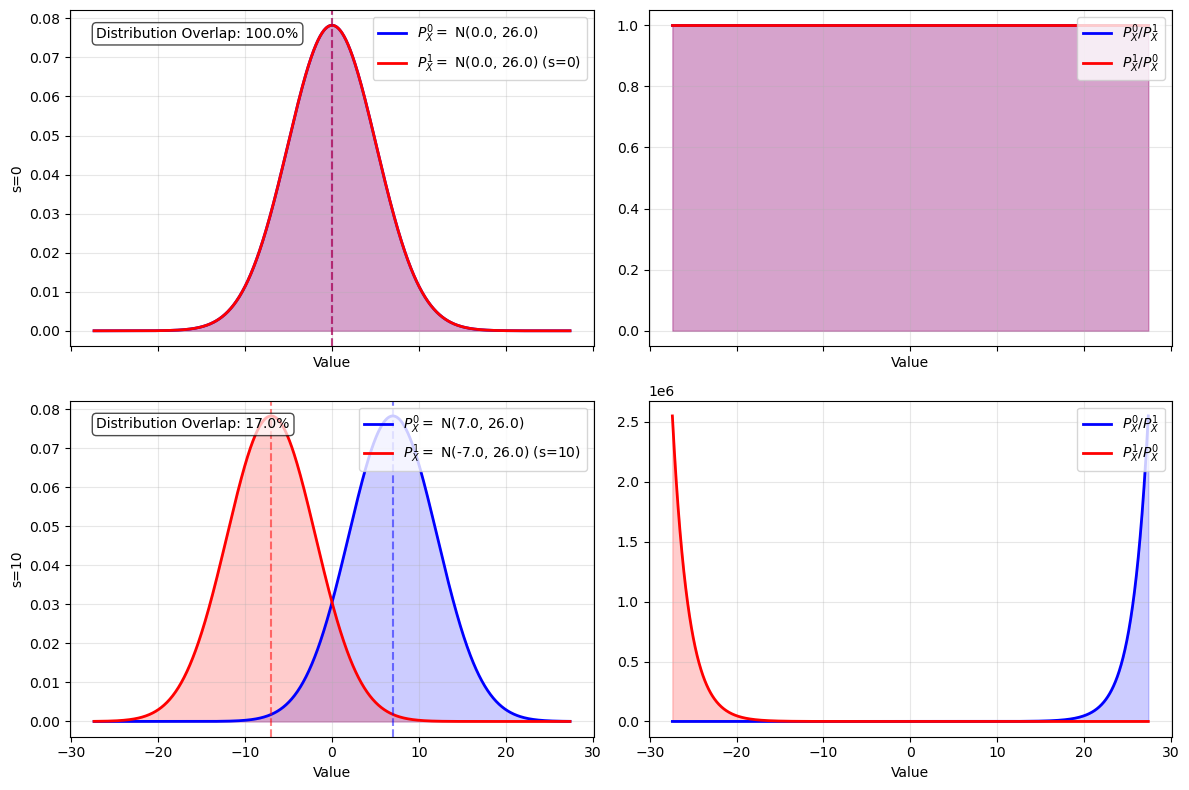}
        \caption{Gaussian example.}
        \label{fig:gaussian_shift}
    \end{subfigure}
    
    \caption{Covariate distributions under increasing distributional divergence. (a) Generalized Beta. (b) Gaussian.}
    \label{fig:covariate_distributions}
\end{figure}

\textbf{Generalized Beta.} Recall that if $X \sim \text{Beta}(\alpha, \beta)$, then $Y = a + (b-a)X \sim \text{Beta}(\alpha, \beta; a, b)$ follows a generalized Beta distribution. Here the two distribution covers the origin $0$, which ensures the intercept can be estimated reliably.
$$
\begin{aligned}
X=(X_1, \dots, X_d)\mid A=1 &\sim Beta(1.3,\, 1.55-0.1\times s;\, 50, -25)^d,\\
X=(X_1, \dots, X_d)\mid A=0 &\sim Beta(1.55-0.1\times s,\, 2;\, 50, -25)^d,
\end{aligned}
$$
with i.i.d.\ coordinates, dimension $d\in\{5, 20\}$, $P(A=1)=0.6$, and $s\in[0,13]$.
Denote $X^2 = (X_1^2, \dots, X_d^2)$, and let the generating process follow a quadratic regression function, where the treatment effect is in the intercept:
$$
    \begin{aligned}
        Y(A)\mid X, A=0 &\sim N(\beta_0^0 + \bfbeta_1^{0\top} X + \bfbeta_1^{0\top} X^2,\; \sigma^2),\\
        Y(A)\mid X, A=1 &\sim N(\beta_0^1 + \bfbeta_1^{1\top} X + \bfbeta_1^{1\top} X^2,\; \sigma^2),
    \end{aligned}
$$
with a common slope on every coordinate: $\beta_0^0 = 1.2$ and $\bfbeta_1^0 = -1.25\cdot\mathbf{1}_d$ under control, $\beta_0^1 = -1.1$ and $\bfbeta_1^1 = 2.5\cdot\mathbf{1}_d$ under treatment, and $\sigma^2=0.1$. The same slope vector multiplies the linear and the quadratic terms. For the objective \eqref{eq:eta_condition} we choose $c=1.1$. 

As shown in \Cref{fig:beta_5d_combined} and \Cref{fig:beta_20d_combined}, our estimators achieve smaller bias and variance than IPW-based methods as $s$ increases from $0$ to $13$. Because the outcome regression is deliberately misspecified, a linear model fit to quadratic data, OR exhibits exploding bias and growing variance as the treated and control covariate distributions diverge. By contrast, the EC-STE with appropriately chosen $\eta$ maintains stable performance throughout.

\begin{figure}[t]
    \centering
    
    \includegraphics[width=\linewidth]{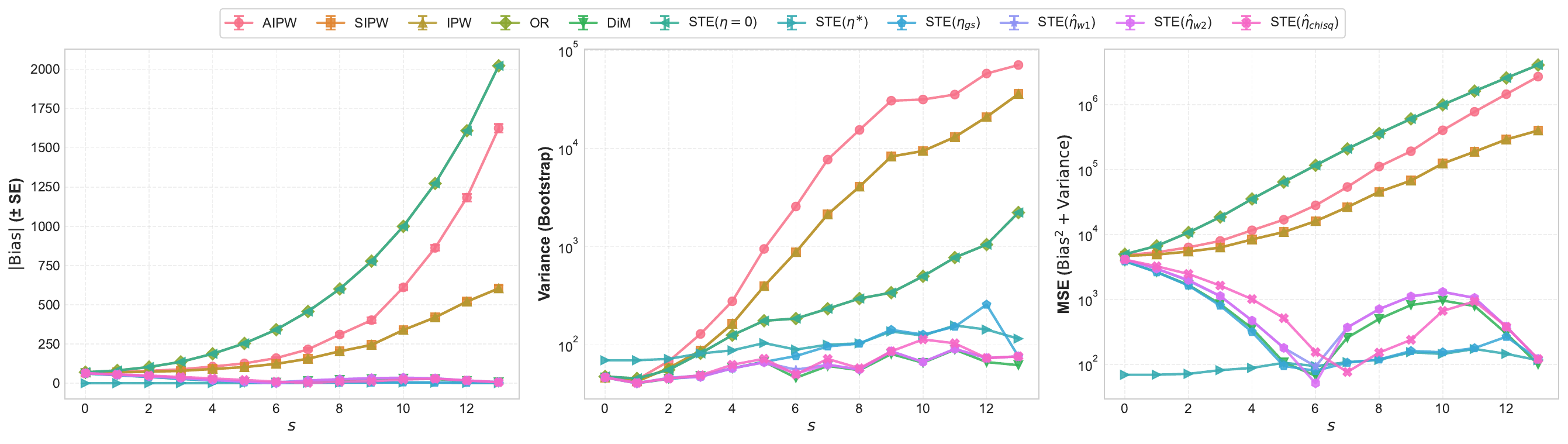}
    
    \vspace{0.3cm}
    
    \includegraphics[width=\linewidth]{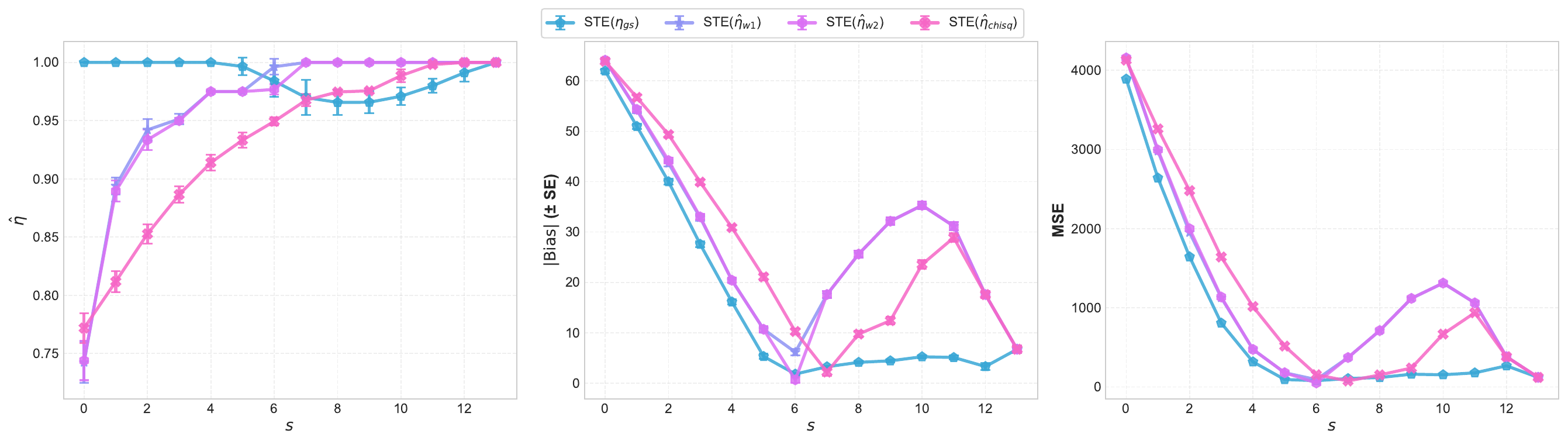}
    
    \caption{Beta 5D example: bias--variance--MSE decomposition. 
    \textbf{Top:} Comparison among AIPW, IPW, SIPW, OR, DiM and our method with different $\eta$ estimation strategies. 
    \textbf{Bottom:} Comparison among different $\eta$ selection strategies with our method.}
    \label{fig:beta_5d_combined}
\end{figure}

\begin{figure}[t]
    \centering
    
    \includegraphics[width=\linewidth]{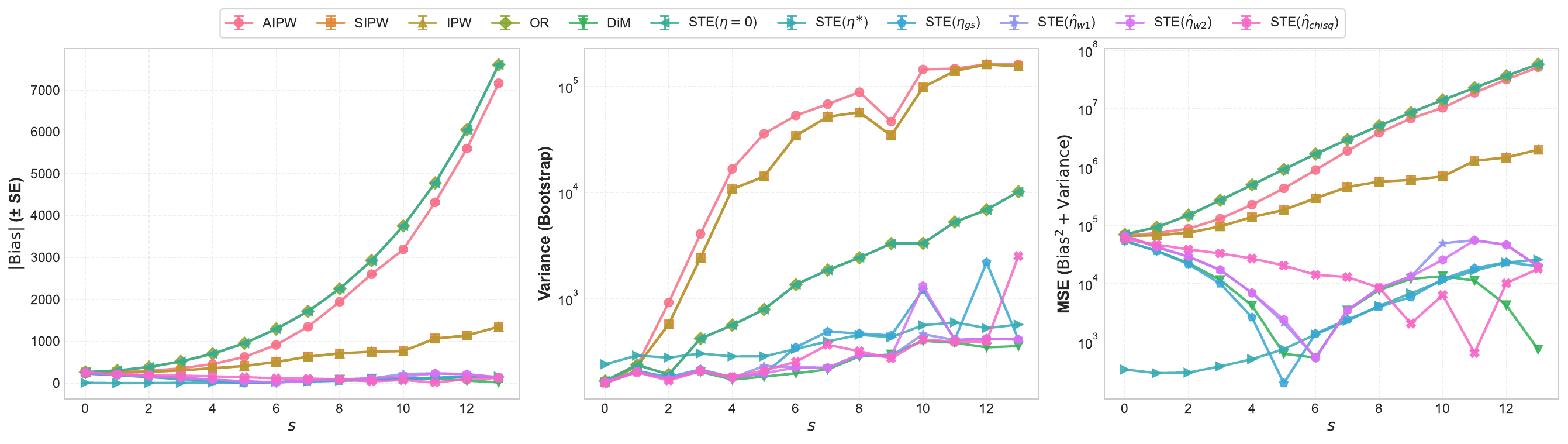}
    
    \vspace{0.3cm}
    
    \includegraphics[width=\linewidth]{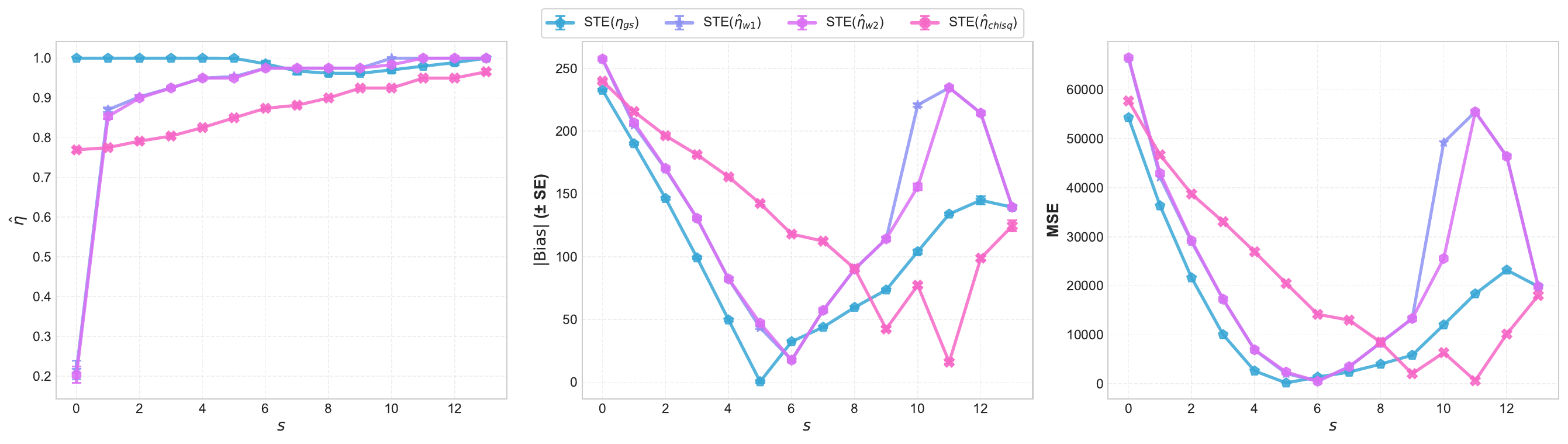}
    
    \caption{Beta 20D example: bias--variance--MSE decomposition. 
    \textbf{Top:} Comparison among AIPW, IPW, SIPW, OR, DiM and our method with different $\eta$ estimation strategies. 
    \textbf{Bottom:} Comparison among different $\eta$ selection strategies with our method.}
    \label{fig:beta_20d_combined}
\end{figure}

\textbf{Gaussian.}
We next consider i.i.d.\ samples from:
$$
\begin{aligned}
X=(X_1, \dots, X_d)\mid A=1 &\sim N(-s,\,5^2)^d,\\
X=(X_1, \dots, X_d)\mid A=0 &\sim N(s,\,5^2)^d,
\end{aligned}
$$
with i.i.d.\ coordinates, dimension $d\in\{5, 20\}$, $P(A=1)=0.6$, and $s\in[0,13]$.
Denote $X^2 = (X_1^2, \dots, X_d^2)$, and let the generating process follow the same quadratic regression function as in the previous example.

In this setting, the Bayes-optimal classifier for the IPW propensity model is linear, representing a best-case scenario for propensity-based methods. Nevertheless, IPW exhibits growing bias and exploding variance as $s$ increases. Our estimators, by contrast, maintain consistently low bias and variance even as $s \to 13$, demonstrating robustness to distributional separation regardless of the support structure (\Cref{fig:norm_5d_combined} and \Cref{fig:norm_20d_combined}).

\begin{figure}[t]
    \centering
    
    \includegraphics[width=\linewidth]{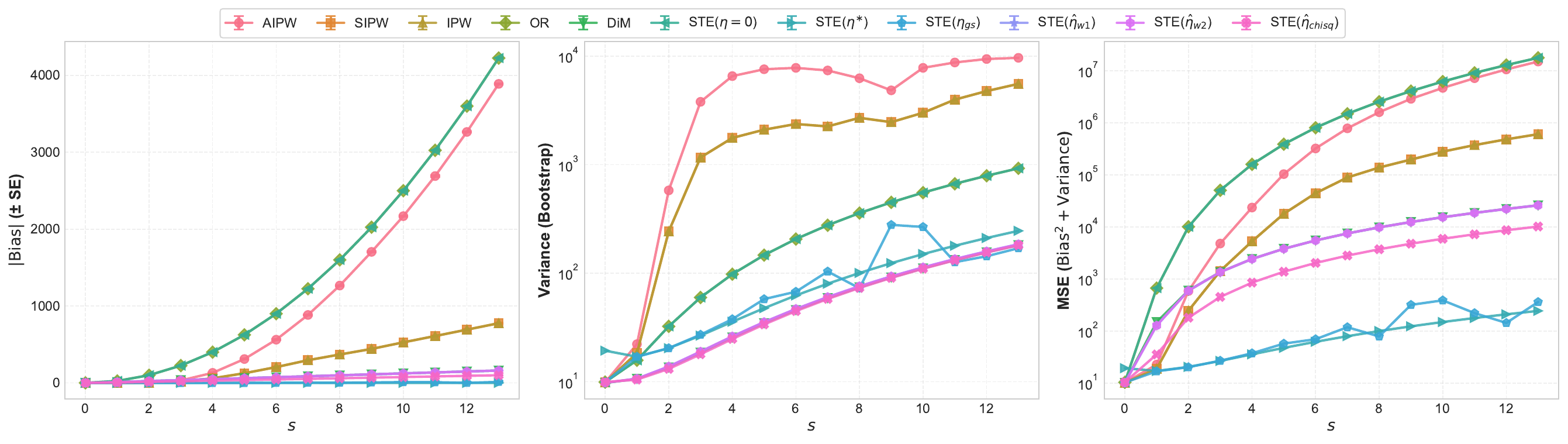}
    
    \vspace{0.3cm}
    
    \includegraphics[width=\linewidth]{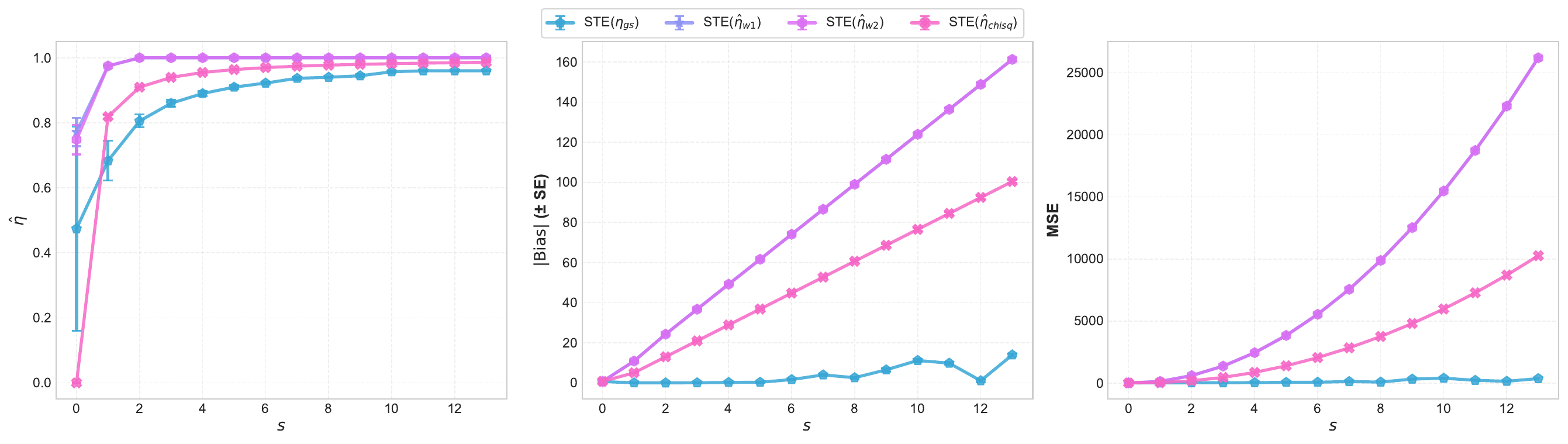}
    
    \caption{Normal 5D example: bias--variance--MSE decomposition.  
    \textbf{Top:} Comparison among AIPW, IPW, SIPW, OR, DiM and our method with different $\eta$ estimation strategies. 
    \textbf{Bottom:} Comparison among different $\eta$ selection strategies with our method.}
    \label{fig:norm_5d_combined}
\end{figure}

\begin{figure}[t]
    \centering
    
    \includegraphics[width=\linewidth]{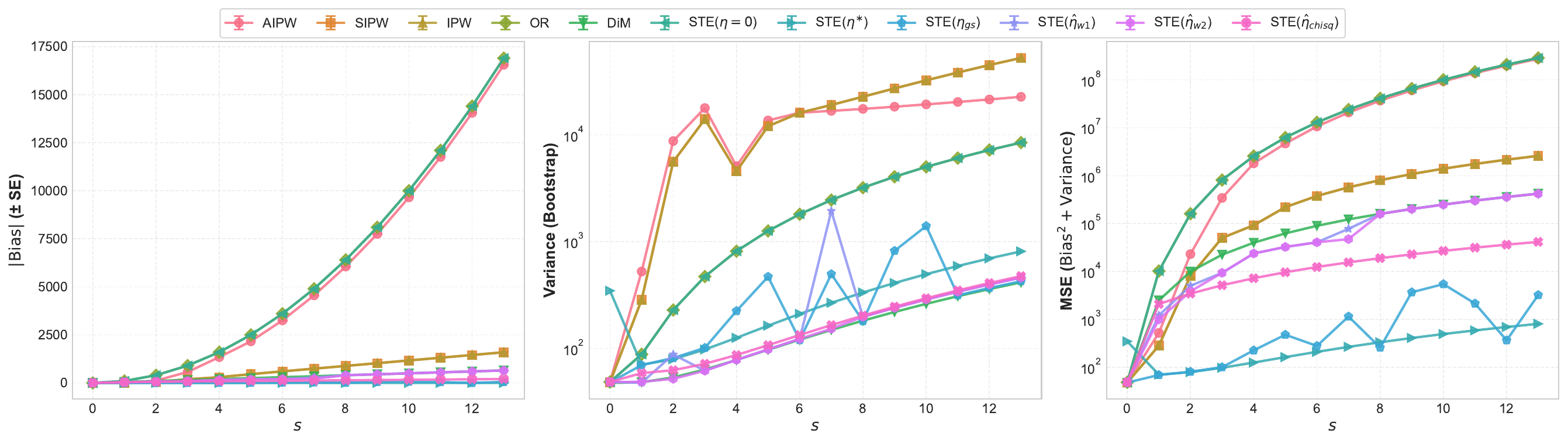}
    
    \vspace{0.3cm}
    
    \includegraphics[width=\linewidth]{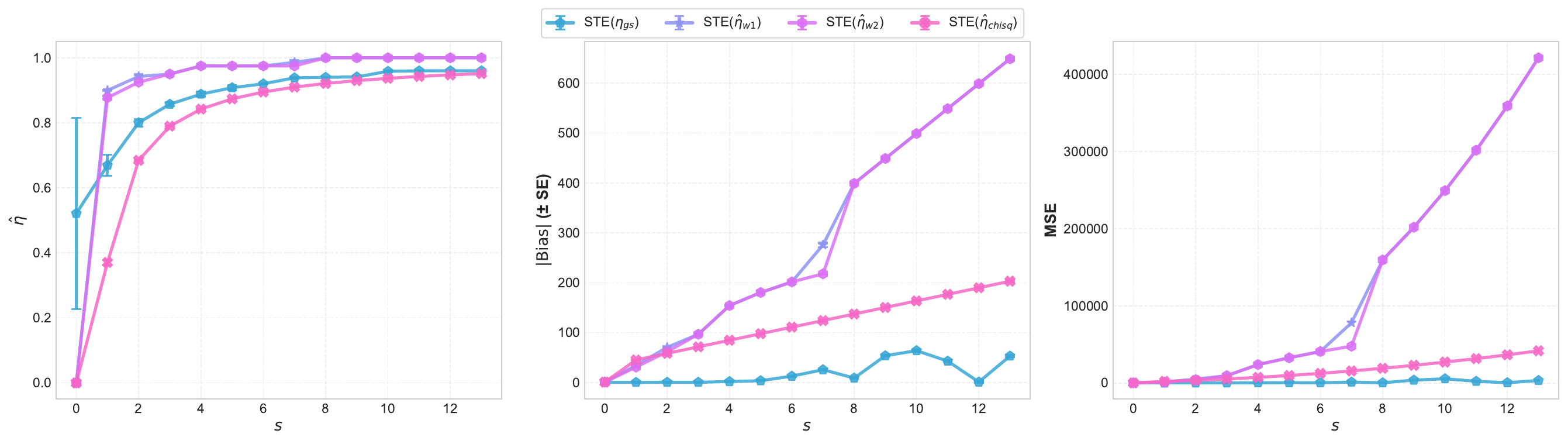}
    
    \caption{Normal 20D example: bias--variance--MSE decomposition.  
    \textbf{Top:} Comparison among AIPW, IPW, SIPW, OR, DiM and our method with different $\eta$ estimation strategies. 
    \textbf{Bottom:} Comparison among different $\eta$ selection strategies with our method.}
    \label{fig:norm_20d_combined}
\end{figure}

\subsection{Lalonde Job Training Dataset}

The Lalonde dataset \citep{lalonde1986evaluating, dehejia1999causal} combines experimental data from the National Supported Work (NSW) Demonstration job training program with an observational comparison group drawn from the Panel Study of Income Dynamics (PSID). This combination is a canonical benchmark for causal methods because the PSID comparison group exhibits severe covariate imbalance: treated individuals are predominantly disadvantaged workers with low baseline earnings, while PSID controls include many stably employed individuals with substantially higher incomes. The experimental benchmark provides a ground-truth ATT of approximately \$1,794 in annual earnings gains, but naive observational comparisons with PSID controls yield severely biased estimates due to extreme selection. For consistency with our methodological framework, we restrict attention to continuous covariates (age, education, and earnings histories from 1974 and 1975), excluding binary indicators. We model the outcome with degree-3 polynomial regression with ridge regularization and the propensity score with linear logistic regression.

In this setting, the outcome regression model is sufficiently well specified that the STE at $\eta = 0$, which coincides with OR by \Cref{prop:empirical-plan-cancellation}, already achieves strong performance. Introducing interpolation with $\eta > 0$ does not improve bias or variance, as the data-driven $\eta$-selection criteria yield estimates comparable to or slightly worse than the baseline. This result is consistent with our theoretical analysis: when the outcome model is well-specified and overlap is adequate for regression, transport-based correction offers no additional benefit. The value of EC-STE emerges precisely in settings where outcome regression extrapolates beyond its training support, a condition we examine in the remaining experiments.

\begin{table}[htbp]
        \centering
        \caption{Lalonde experiment: bias, variance, and MSE decomposition across estimation methods and tuning parameter values}
        \label{tab:lalonde}
        \resizebox{\textwidth}{!}{%
        \begin{tabular}{lrrrrrrrrrr}
            \toprule
            \textbf{Metric} & AIPW & SIPW & IPW & OR & DiM & STE($\eta=0$) & STE($\eta_{gs}^{\ast}$) & STE($\hat{\eta}_{\text{chisq}}$) & STE($\hat{\eta}_{w1}$) & STE($\hat{\eta}_{w2}$) \\
            \midrule
            \textbf{Bias $(\times 10^{3})$} & $-1.268$ & $-1.791$ & $-1.791$ & $-1.879$ & $-17.555$ & $-1.879$ & $-1.879$ & $-2.237$ & $-2.483$ & $-2.447$ \\
    \textbf{Variance $(\times 10^{5})$} & 2.054 & 1.528 & 1.528 & 2.693 & 2.128 & 2.693 & 2.693 & 4.954 & 4.665 & 5.887 \\
    \textbf{MSE $(\times 10^{6})$} & 1.813 & 3.360 & 3.360 & 3.799 & 308.388 & 3.799 & 3.799 & 5.500 & 6.629 & 6.576 \\
    \bottomrule
        \end{tabular}%
        }
        \end{table}

\subsection{Right Heart Catheterization Semi-synthetic Study}

We use the Right Heart Catheterization (RHC) dataset \citep{connors1996effectiveness}, which contains observational data on critically ill ICU patients, to evaluate estimator performance under controlled overlap violations. Rather than using the original treatment assignments, we adopt a semi-synthetic design: we retain the real covariates and outcomes but generate new treatment assignments through a known propensity score model. This allows us to create deliberate regions of minimal covariate overlap while preserving the realistic correlation structure and clinical relationships in the original data. Specifically, we construct propensity scores that strongly tie treatment assignment to disease severity markers, including APACHE III scores, vital signs, and laboratory values, so that the sickest patients have near-deterministic probability of receiving RHC ($e(X) \approx 1$) while the healthiest are almost certainly untreated ($e(X) \approx 0$). The resulting covariate distributions between treated and control groups exhibit substantial non-overlap, providing a controlled testbed for evaluating ATT estimators when overlap is severely violated. We use linear logistic regression for the propensity model and polynomial outcome regression with ridge regularization.

\Cref{tab:rhc_results} reveals a clear separation in performance. The STE at $\eta = 0$ coincides with OR by \Cref{prop:empirical-plan-cancellation} and inherits its large bias of $45.8$ from extrapolation beyond the regression training support. IPW-based methods achieve lower bias but at the cost of high variance, an order of magnitude larger than our best estimators. The EC-STE with data-driven $\eta$-selection dramatically improves on both: the $W_2$-selected variant achieves bias of $-0.814$ and MSE of $6.64$, representing a $99.7\%$ MSE reduction relative to OR and a $71\%$ reduction relative to IPW. The oracle grid-search estimator ($\eta_{gs}^{\ast}$) achieves the lowest MSE of $3.02$, confirming that the data-driven selectors approximate the optimal interpolation well.

\begin{table}[htbp]
        \centering
        \caption{RHC semi-synthetic study: bias, variance, and MSE decomposition across estimation methods and tuning parameter values}
        \label{tab:rhc_results}
        \resizebox{\textwidth}{!}{%
        \begin{tabular}{lrrrrrrrrrr}
            \toprule
            \textbf{Metric} & AIPW & SIPW & IPW & OR & DiM & STE($\eta=0$) & STE($\eta_{gs}^{\ast}$) & STE($\hat{\eta}_{\text{chisq}}$) & STE($\hat{\eta}_{w1}$) & STE($\hat{\eta}_{w2}$) \\
            \midrule
            \textbf{Bias $(\times 10^{-1})$} & 116.952 & 10.982 & 10.982 & 457.969 & $-412.328$ & 457.969 & 16.432 & $-11.468$ & $-8.251$ & $-8.139$ \\
    \textbf{Variance $(\times 10^{-1})$} & 9527.541 & 217.692 & 217.692 & 458.617 & 5.236 & 458.617 & 3.171 & 90.859 & 60.978 & 59.776 \\
    \textbf{MSE} & 1089.531 & 22.975 & 22.975 & 2143.215 & 1700.669 & 2143.215 & 3.017 & 10.401 & 6.779 & 6.640 \\
    \bottomrule
        \end{tabular}%
        }
        \end{table}

\subsection{Medical Expenditure Panel Survey (MEPS) Study}

We analyze data from the Medical Expenditure Panel Survey (MEPS) 2022 Full Year Consolidated files \citep{meps}, a nationally representative survey of the U.S.\ civilian noninstitutionalized population. The goal is to estimate the ATT of Medicare coverage on total healthcare expenditure (TOTEXP) among individuals aged 58--72. 

The treatment is a binary indicator for Medicare coverage, and our observed confounders include continuous variables capturing age, socioeconomic status (income relative to the poverty line), health status (PCS-42 and MCS-42 physical and mental component summary scores, PHQ-2 depression scale scores), and prior healthcare utilization.

To ensure a clean identification strategy with a single, well-defined treatment assignment mechanism, we restrict our sample to individuals who obtained Medicare through age-based eligibility only. Specifically, we define the treated group as individuals aged 65 or older with Medicare coverage (identified via the MEPS variable MCREV=1), and the control group as individuals under age 65 without Medicare coverage (MCREV=2). Critically, we exclude two groups from our analysis: (1) individuals under age 65 with Medicare coverage, who obtained coverage through disability-based pathways; and (2) individuals aged 65 or older without Medicare coverage (non-compliers), as our estimand targets the ATT rather than the intent-to-treat effect. 

This sample restriction yields approximately 40{,}000 to 50{,}000 observations across the four years, with roughly 60\% in the control group and 40\% in the treated group. The resulting design exhibits an extreme violation of overlap by construction: the sharp discontinuity in Medicare eligibility at age 65 creates an approximately 95\% point jump in coverage probability, leading to near-complete separation in the age distribution and substantial imbalances across the continuous health status and socioeconomic measures.

\Cref{tab:meps_results} demonstrates the practical value of transport-based estimation in this challenging setting. Standard methods struggle: OR and AIPW exhibit large bias (around $-1{,}850$) and high variance, while IPW and SIPW trade bias in the opposite direction ($+2{,}830$) for lower variance but worse overall MSE. DiM performs worst due to the strong confounding structure. The EC-STE with grid-search oracle ($\eta_{gs}^{\ast}$) achieves an order-of-magnitude improvement, with MSE of $2{,}388$ compared to $3{,}884 \times 10^3$ for OR. Among data-driven selectors, the $W_1$ and $W_2$ variants achieve MSEs of $193 \times 10^3$ and $190 \times 10^3$ respectively, roughly a 95\% reduction relative to OR, demonstrating that principled interpolation can substantially improve estimation even in real-world settings with extreme overlap violations.

\begin{table}[htbp]
        \centering
        \caption{MEPS study: bias, variance, and MSE decomposition across estimation methods and tuning parameter values}
        \label{tab:meps_results}
        \resizebox{\textwidth}{!}{%
        \begin{tabular}{lrrrrrrrrrr}
            \toprule
            \textbf{Metric} & AIPW & SIPW & IPW & OR & DiM & STE($\eta=0$) & STE($\eta_{gs}^{\ast}$) & STE($\hat{\eta}_{\text{chisq}}$) & STE($\hat{\eta}_{w1}$) & STE($\hat{\eta}_{w2}$) \\
            \midrule
            \textbf{Bias $(\times 10^{1})$} & $-185.557$ & 282.949 & 282.949 & $-184.758$ & 498.157 & $-184.758$ & $-2.733$ & $-94.266$ & $-3.429$ & 2.229 \\
    \textbf{Variance $(\times 10^{3})$} & 472.934 & 25.788 & 25.788 & 470.748 & 7.592 & 470.748 & 1.642 & 354.138 & 191.990 & 189.088 \\
    \textbf{MSE $(\times 10^{3})$} & 3916.086 & 8031.807 & 8031.807 & 3884.285 & 24823.632 & 3884.285 & 2.388 & 1242.754 & 193.165 & 189.585 \\
    \bottomrule
        \end{tabular}%
    }
\end{table}

\section{Discussion}
\label{sec:conclusion}
Limited overlap between treatment groups poses a fundamental challenge in observational causal inference, particularly when treatments are rare, selectively assigned, or measured alongside high-dimensional covariates. Limited overlap manifests theoretically through unbounded density ratios that violate standard regularity conditions, and practically through degenerate propensity score estimates that destabilize inverse probability weighting. Existing approaches either restrict inference to subpopulations with adequate overlap, thereby compromising external validity, or employ trimming procedures that discard data and change the estimand. The need for principled methods that maintain population-level inference while providing stability under overlap violations has remained a critical gap in the causal inference literature.

This paper addresses these challenges through estimators grounded in optimal transport that explicitly use the geometry of the covariate distributions. We introduce two complementary estimators. The Stable Transport Estimator (STE) performs covariate adjustment via optimal transport maps, achieving variance comparable to outcome regression even under near-violations of overlap. When overlap fails entirely and standard methods do not apply, we further propose the Extrapolation-Controlled Stable Transport Estimator (EC-STE), which interpolates along the displacement geodesic between the covariate distributions to control extrapolation. For both estimators we establish consistency and central limit theorems: the STE under explicit conditions that characterize the contribution of the estimated transport map to the asymptotic variance, and the EC-STE with an interpolation bias that is explicit to first order, estimable, and removable, at unchanged asymptotic variance.

These theoretical advances translate to substantial empirical improvements. Through extensive empirical studies, we demonstrate that STE and EC-STE achieve meaningful variance reduction compared to IPW-based methods and MSE reduction compared to outcome regression, particularly in high-dimensional settings with limited overlap. Semi-synthetic experiments on the Right Heart Catheterization dataset validate the $\eta$-selection procedures under realistic model misspecification, while the application to the Medical Expenditure Panel Survey illustrates the practical utility of the approach for policy-relevant causal questions where overlap is naturally limited by design.

Several directions follow naturally from this work. The current central limit theory covers transport maps that admit an asymptotically linear expansion, such as parametric plug-in maps, and extending it to nonparametric map estimators is a natural next step. Moreover, the choice among exact transport maps leaves the first-order variance unchanged, and whether the transport structure can be exploited for semiparametric efficiency remains open. Finally, the identification relies on unconfoundedness, and relaxing this assumption within the transport framework is an avenue we leave for future exploration.

\bigskip

\section*{Acknowledgements} This work was supported in part by funding from the Office of Naval Research under grant N00014-23-1-2590, the Department of Energy under grant DE-SC0026902, the National Science Foundation under grant DMS-2310831, DMS-2428059, IIS-2435696, IIS-2440954, Two Sigma Investments LP, and LG Management Development Institute AI Research. 

\bibliographystyle{agsm}
\bibliography{biblio}

\newpage
\appendix

\section{Notation Glossary}
\label{app:notation}

\Cref{tab:notation,tab:notation2} collect the notation used throughout the paper, in order of first appearance.

\begin{table}[H]
\centering
\begingroup
\spacingset{1.1}\small
\begin{tabular}{@{}l p{0.56\textwidth}@{}}
\toprule
Symbol & Meaning \\
\midrule
\multicolumn{2}{@{}l}{Population and causal quantities} \\
$(Y(1), Y(0), A, X) \sim \mathcal P$ & potential outcomes, treatment $A\in\{0,1\}$, covariates $X\in\calx\subseteq\R^d$ \\
$Y = AY(1) + (1-A)Y(0)$ & observed outcome (SUTVA) \\
$p = (a, x, y)$ & generic realized observation \\
$r = \mathcal P(A=1)$ & marginal treatment probability \\
$P_X^a$, $P_Y^a$ & law of $X$, resp.\ $Y$, given $A=a$ \\
$p_X^a$ & Lebesgue density of $P_X^a$ \\
$e(x) = \mathcal P(A=1 \mid X=x)$ & propensity score \\
$h_a(x) = \E[Y \mid X=x, A=a]$ & outcome regression under arm $a$ \\
$\tau_{ATT}$, $\tau_{ATC}$, $\tau_{ATE}$ & average treatment effects (on treated, control, population) \\
\midrule
\multicolumn{2}{@{}l}{Optimal transport} \\
$\calp_p(\calx)$, $\calp_{\mathrm{ac}}(\calx)$, $\calp_{p,\mathrm{ac}}(\calx)$ & measures with finite $p$-th moment; absolutely continuous; both \\
$T_{\sharp}P$ & pushforward of $P$ under $T$ \\
$T$ & Monge map, $T_{\sharp}P_X^0 = P_X^1$ \\
$S$ & Monge map in the reverse direction, $S_{\sharp}P_X^1 = P_X^0$ \\
$W_p(\cdot,\cdot)$ & Wasserstein distance of order $p$ \\
$T_{\eta} = \eta\,\mathrm{Id} + (1-\eta)T$ & interpolated map, $\eta\in[0,1]$; $\eta\to 0$ recovers $T$ \\
$P_X^{\eta} = (T_{\eta})_{\sharp}P_X^0$ & interpolated covariate distribution \\
$c_1$, $\hat c_1$ & first-order interpolation bias coefficient and its plug-in \\
$b_{\eta} = (P_X^{\eta} - P_X^1)h_0$ & interpolation imputation bias, equal to term (c) of Section 5 \\
\bottomrule
\end{tabular}
\endgroup
\caption{Notation: population, causal, and transport quantities.}
\label{tab:notation}
\end{table}

\begin{table}[H]
\centering
\begingroup
\spacingset{1.1}\small
\begin{tabular}{@{}l p{0.56\textwidth}@{}}
\toprule
Symbol & Meaning \\
\midrule
\multicolumn{2}{@{}l}{Samples and cross-fitting} \\
$\Omega$, $n$ & full sample and its size; $\Omega_a$, $n_a$: arm-$a$ subsample \\
$\Omega^{(k)}$, $n^{(k)} = n/K$ & evaluation fold $k$ of $K$; $\Omega_a^{(k)}$: its arm-$a$ part \\
$\Omega^{(-k)}$ & training folds, $|\Omega^{(-k)}| = (K-1)n/K$ \\
$P_{X,n}^{a}$, $P_{Y,n}^{a}$ & empirical arm measures; fold versions $P_{X,n}^{a,(k)}$, $P_{X,n}^{a,(-k)}$ \\
$P_n^{(k)}[\,\cdot\,]$ & empirical average over fold $k$ \\
$\hat h_0^{(-k)}$, $\hat T^{(-k)}$ & nuisance estimators trained on $\Omega^{(-k)}$ \\
$\hat\tau_{ATT}$, $\hat\tau_{ATT}^{(k)}$ & cross-fitted STE and its fold-$k$ component \\
\midrule
\multicolumn{2}{@{}l}{Estimation and variance theory} \\
$P_X^a f := \E_{P_X^a}[f(X)]$ & population functional (operator) notation \\
$P^a f$, $P_n^{a}$, $P_n^{a,(k)}$ & $P^a f := \E[f(X,Y)\mid A=a]$; arm and arm-fold empirical counterparts \\
$X'$ & independent copy of $X$ (inner integration variable) \\
$g_k := \hat h_0^{(-k)} \circ \hat T^{(-k)}$ & composed plug-in regression; $g_{\star}$: its $L_2(P_X^0)$ limit \\
$Q_k := \hat T^{(-k)}_{\sharp} P_X^0$ & pushforward of the control law under the fitted map \\
$B_k := P_X^0 g_k - P_X^1 h_0$ & population error of the fold-$k$ plug-in; the conditional bias is $-B_k$; $\overline B_n = K^{-1}\sum_k B_k$ \\
$B_{T,k}$, $B_{h,k}$ & map and regression components of $B_k$ \\
$\psi_{Y,h_0,T}$, $\psi_{g_{\star}}$ & influence function at exact $(h_0, T)$; version with nuisances held fixed at $g_{\star}$ \\
$\phi_T$, $\phi_h$ & asymptotic-linearity kernels of $\hat T^{(-k)}$, $\hat h_0^{(-k)}$ \\
$\psi_T$, $\psi_h$ & influence terms induced by the nuisance estimation \\
$r_n^{T,(k)}$, $r_n^{h,(k)}$ & asymptotic-linearity remainders \\
$T^{\dagger}$, $Q^{\dagger} = T^{\dagger}_{\sharp}P_X^0$, $b^{\dagger}$ & pseudo-limit map, its pushforward, its bias offset \\
$V_{\mathrm{or}}$, $V_{\star}$ & fixed-map oracle variance; $V_{\star} = \E[\psi_{g_{\star}}^2]$ \\
$L_h$, $L^{\ast}_{\nabla h_0}$ & bound and Lipschitz constant of $\nabla h_0$ \\
$L^{(-k)}$, $M^{(-k)}$ & Lipschitz constants of $\hat h_0^{(-k)}$ and of $\nabla \hat h_0^{(-k)}$ \\
\bottomrule
\end{tabular}
\endgroup
\caption{Notation: samples, cross-fitting, and estimation theory.}
\label{tab:notation2}
\end{table}

\section{Brief Recapitulation to Optimal Transport}
\label{sec:brief_recapitulation_to_optimal_transport}

In the following, we briefly recall the key concepts in Optimal Transport.
\textbf{Monge Problem} 
Theorem 5.30 of \citet{villani2009optimal} establishes the existence and uniqueness of the Monge map, also called the deterministic coupling. 
When the source measure $\mu$ is absolutely continuous with respect to Lebesgue measure and both $\mu, \nu \in \mathcal{P}_p\left(\mathbb{R}^d\right)$, the Monge formulation of optimal transport admits a unique optimal map $T$ pushing $\mu$ forward to $\nu$. For the quadratic cost $c(x, y)=\|x-y\|^2$, Brenier's theorem \citep{brenier1991polar} shows that $T= \nabla \psi$ for a convex function $\psi$, uniquely determined $\mu$-a.e. McCann \citep{mccann1995existence} extended this result by relaxing several regularity assumptions on the measures. For general costs $c(x, y)=h(x-y)$ with $h$ strictly convex - including $c(x, y)=\|x-y\|^p$ for any $p>1$ -
\citet{gangbo1996geometry} established existence and uniqueness of the optimal map, which is characterized through $c$-concave Kantorovich potentials rather than gradients of convex functions. The case $p=1$ is qualitatively different: the cost $c(x, y)=\|x-y\|$ is not strictly convex, so uniqueness fails in general; existence of an optimal map was proved by \citet{evans1999differential} under compact support assumptions. Comprehensive treatments of these results can be found in \citet[Theorems 2.12 and 2.44]{villani2003topics} and \citet[Part I, Chapters 9--10]{villani2009optimal}.

We present here the most well-known case established in the Euclidean space with quadratic cost.
\begin{theorem}[Brenier Theorem \citep{villani2009optimal}]
    Let $\calx,\caly\subset \R^d$ carry the quadratic cost $c(x,y)=\|x-y\|^2_2$.
    Let $\mu\in\calp_{2}(\calx)$ and $\nu\in\calp_2(\caly)$ be probability measures with finite second moments, and suppose that the source measure $\mu$ is absolutely continuous with respect to the Lebesgue measure $\lambda$ on $\mathbb{R}^d$.
    
    Then there exists a unique coupling $\pi$ that solves the Kantorovich problem. Under $\pi$, $y$ is uniquely determined for $\mu$-a.e.\ $x$. In addition, there exists a lower semi-continuous function $\varphi:\calx\to\R$ such that $y\in\partial_x\varphi(x)$ for $\mu$-a.e.\ $x$.

    Since $\partial_x\varphi(x)$ is a singleton for $\mu$-a.e. $x$, the Monge map $T$ exists which is $y=T(x)=\Tilde{\nabla}\varphi(x)$, and the optimal coupling $\pi$ is characterized by 
    $$
    \begin{aligned}
    \nu & =T_{\sharp} \mu \\
    \pi & =(I d, T)_{\sharp} \mu .
    \end{aligned}
    $$
\end{theorem}

The Monge problem above is equivalent to an optimal matching problem when $P_X^1$ and $P_X^0$ are discrete with the same number of atoms \citep{peyre2019computational}. To allow mass splitting, one solves the Kantorovich problem instead, which finds an optimal coupling between distributions that may be discrete or continuous. From an estimated coupling, the Brenier map can be estimated by barycentric projection \citep{reich2013nonparametric}, even for finite and unbalanced samples. A parametric plug-in \citep{perrot2016mapping, deb2021rates, zhu2021functional} extends the map to out-of-sample points. 

\textbf{Kantorovich Problem.} \citet{kantorovich1942translocation} proposed a relaxation of the Monge problem that seeks a coupling rather than a map. The coupling allows mass splitting from $P$ to $Q$, which enlarges the set of solvable problems. The problem is to find a coupling $\pi$ of $P$ and $Q$ that minimizes the expected cost of transportation,
$$
\inf_{\pi\in\Pi(P,Q)} \int c(x,y) d\pi(x,y) 
$$
The Kantorovich problem also induces a metric on the space of probability measures.
\textbf{Wasserstein metric.} Let $(\mathcal{S}, d)$ be a Polish metric space and let $q \in[1, \infty)$. For any two probability measures $\mu, \nu$ on $\mathcal{S}$, the Wasserstein metric of order $q$ is given by
$$
W_q(\mu, \nu) =\left(\inf _{\pi \in \Pi(\mu, \nu)} \int_{\mathcal{S}\times\mathcal{S}} d(x, y)^q \, d \pi(x, y)\right)^{1 / q},
$$
which is a metric on the space of probability measures on $\mathcal{S}$ with finite $q$-th moments.

\section{Connection to Inverse Propensity Weighting}
\label{app:ipw-connection}

Throughout this appendix we work at the population level and write $p_X^1, p_X^0$ for the Lebesgue densities of $P_X^1, P_X^0$. Let $r := \Pr(A=1)$ denote the marginal treatment probability, so that $\Pr(A=0) = 1-r$, and let $e(x) = \Pr(A=1 \mid X=x)$ denote the propensity score. Recall $h_0(x) = \E[Y \mid A=0, X=x]$ and that the Monge map $T$ satisfies $T_{\sharp} P_X^0 = P_X^1$. The control counterfactual component of
$\tau_{ATT}$ is
$$
  \theta := \E[Y(0) \mid A=1] = \int h_0  dP_X^1.
$$
By Bayes' rule the ATT density ratio is
$$
  \varrho(x) := \frac{dP_X^1}{dP_X^0}(x) = \frac{p_X^1(x)}{p_X^0(x)}
            = \frac{1-r}{r}\frac{e(x)}{1-e(x)},
$$
so $\varrho$ and the inverse-propensity odds $e/(1-e)$ coincide up to the
constant $(1-r)/r$, which equals one under equal group sizes ($r=\tfrac12$)
and cancels in any self-normalized (H\'ajek) form. The two adjusted forms of
the estimand are
$$
  \theta_{\mathrm{OT}} := \E_{P_X^0}\!\left[h_0 \circ T\right],
  \qquad
  \theta_{\mathrm{IPW}} := \E_{P_X^0}\!\left[\varrho \cdot h_0\right].
$$

\begin{ass}[Bounded density ratio]\label{ass:overlap}
$p_X^0(x) > 0$ wherever $p_X^1(x) > 0$, and $\varrho$ is bounded on
the support of $P_X^1$.
\end{ass}
This is a one-sided form of strict overlap (\Cref{ass:strong_positivity}). It bounds the density ratio only in the direction used in this appendix.

\begin{ass}[Regularity]\label{ass:reg}
$p_X^0, p_X^1$ have common convex support on which they are bounded away from $0$ and $\infty$, and $h_0 \in L^1(P_X^1)$. Under
\Cref{ass:reg} with quadratic cost, Brenier's theorem gives a unique map $T = \nabla \varphi$ with $\varphi$ convex, $T_\sharp P_X^0 = P_X^1$, and $T$ a diffeomorphism on the interior of the support, with inverse $S := T^{-1}$ satisfying $S_\sharp P_X^1 = P_X^0$.
\end{ass}

\subsection{Monge map}
We first connect the exact Monge map to the IPW weights.

\begin{prop}
\label{prop:monge-ipw}
Under \Cref{ass:overlap} and \Cref{ass:reg},
$$
  \theta_{\mathrm{OT}} = \theta_{\mathrm{IPW}} = \theta,
$$
and the Monge-Amp\`ere identity holds pointwise:
$$
  p_X^0(x) = p_X^1\!\big(T(x)\big)\big|\!\det \nabla T(x)\big|
  \qquad\Longleftrightarrow\qquad
  \big|\!\det \nabla S(y)\big| = \frac{p_X^1(y)}{p_X^0\!\big(S(y)\big)}.
$$
Consequently the weight that transport assigns to a point $y \in \supp P_X^1$,
$$
  w_{\mathrm{OT}}(y) = \big|\!\det \nabla S(y)\big|
                    = \frac{p_X^1(y)}{p_X^0\!\big(S(y)\big)},
  \qquad\text{versus}\qquad
  w_{\mathrm{IPW}}(y) = \varrho(y) = \frac{p_X^1(y)}{p_X^0(y)},
$$
satisfies $w_{\mathrm{OT}} \equiv w_{\mathrm{IPW}}$ on $\supp P_X^1$ if and only if $S = \mathrm{id}$, equivalently $T = \mathrm{id}$, which under a balanced map forces $P_X^0 = P_X^1$.
\end{prop}

\begin{proof}
Equality of the functionals is the push-forward identity $\int (\phi \circ T)dP_X^0 = \int \phi dP_X^1$ applied to $\phi = h_0$,
which simultaneously yields the IPW form on writing $dP_X^1 = \varrho dP_X^0$. The displayed identity is the change of variables for the diffeomorphic push-forward $S_\sharp P_X^1 = P_X^0$. The weight comparison reads $p_X^1$ against $p_X^0$ at the same point ($w_{\mathrm{IPW}}$) versus against $p_X^0$ at the preimage $S(y)$ ($w_{\mathrm{OT}}$); these agree for all $y$ iff $S(y)=y$.
\end{proof}

\subsection{Entropic OT}
For the entropic OT map, we give the following characterization.

\begin{prop}
\label{prop:entropic-ipw}
For $\epsilon > 0$ let
$$
  \pi^\epsilon = \argmin_{\pi \in \Pi(P_X^0, P_X^1)}
    \int c  d\pi + \epsilon \mathrm{KL}\!\big(\pi \| P_X^0 \otimes P_X^1\big),
$$
with Schr\"odinger potentials $(f,g)$, barycentric map
$T^\epsilon(x) = \int y  d\pi^\epsilon(y \mid x)$, and entropic estimand
$\theta_{\mathrm{OT}}^\epsilon := \E_{P_X^0}[h_0 \circ T^\epsilon]$. Then:
\begin{enumerate}
  \item\label{it:plan} \emph{(Plan exactness.)} The second marginal of
  $\pi^\epsilon$ is exactly $P_X^1$, so pushing $h_0$ through the plan gives
  $\iint h_0(y)d\pi^\epsilon(x,y) = \theta$ for every $\epsilon$.
  \item\label{it:limits} \emph{(Limits.)} As $\epsilon \to 0$,
  $\pi^\epsilon \to \pi^{\mathrm{Monge}}$ and $T^\epsilon \to T$, so
  $\theta_{\mathrm{OT}}^\epsilon \to \theta_{\mathrm{OT}} = \theta$ and
  Proposition~\ref{prop:monge-ipw} applies. As $\epsilon \to \infty$,
  $\pi^\epsilon \to P_X^0 \otimes P_X^1$ and $T^\epsilon(x) \to \E_{P_X^1}[X]$,
  collapsing all covariate adjustment.
  \item\label{it:recover} \emph{(IPW recovery.)} Replace the marginal constraints by the penalties
  $$\lambda_0\,\mathrm{KL}(\pi_0 \| p_X^0) + \lambda_1\,\mathrm{KL}(\pi_1 \| p_X^1),$$
  and take the localizing cost $c_\epsilon(x,y) = \tfrac1\epsilon\|x-y\|^2$.
  The optimal plan factorizes as
  $d\pi(x,y) = u(x)v(y)e^{-c_\epsilon(x,y)/\epsilon}dP_X^0(x)dP_X^1(y)$.
  As $\epsilon \to 0$ with the treated mass held hard ($\lambda_1 \to \infty$),
  the kernel concentrates on the diagonal, the control-side scaling satisfies
  $u(x)p_X^0(x) \to p_X^1(x)$, hence
  $$
    u(x) \rightarrow \frac{p_X^1(x)}{p_X^0(x)} = \varrho(x),
  $$
  and the induced control weights coincide with the IPW weights.
\end{enumerate}
\end{prop}

\begin{proof}
Part~\ref{it:plan} is the marginal constraint of the balanced entropic
problem. Part~\ref{it:limits} is the standard $\Gamma$-convergence of
entropic OT to the unregularized problem as $\epsilon \to 0$ and to the
independent coupling as $\epsilon \to \infty$, together with continuity of
the barycentric projection. For part~\ref{it:recover}, the Gibbs form of the
unbalanced optimum gives the stated factorization; under the localizing cost
the kernel $e^{-c_\epsilon/\epsilon}$ concentrates on $\{x=y\}$, so matching
the hard treated mass forces $up_X^0 \to p_X^1$.
\end{proof}

\section{Proof Details}\label{app:proofs}

\begin{proof}[Direct proof of \Cref{thm:att_consistency}]
    The estimation error can be decomposed into
    $$
    \begin{aligned}
        |\hat{\tau}_{ATT}^{(k)} - \tau_{ATT}| &\leq |E_{P_{Y,n}^{1,(k)}}[Y] - E_{P_{Y}^{1}}[Y]| +\Bigl|\Bigl( P_{X,n}^{0,(k)}\Bigl(\hat{h}_0^{(-k)}\circ \hat{T}^{(-k)}\Bigr) - P_X^0\bigl(h_0\circ T\bigr)\Bigr)\Bigr|\\
        &= (I) + (II)
    \end{aligned}
    $$
    The term $(I)$ is the sample mean estimate of the population mean. As $|\Omega_{1}^{(k)}|/|\Omega^{(k)}|\to_p r$ and $\E[A_iY_i] = r\E[Y_i|A_i=1]=rE_{P_{Y}^{1}}[Y]<\infty$, the WLLN implies $(I) = E_{P_{Y,n}^{1,(k)}}[Y] - E_{P_{Y}^{1}}[Y] = o_p(1).$
    The second term $(II)$ can be further decomposed into three terms
    $$
    \begin{aligned}
        (II) &\leq \Bigl|P_{X,n}^{0,(k)}\Bigl(\hat{h}_0^{(-k)}\circ \hat{T}^{(-k)}\Bigr)- P_{X,n}^{0,(k)}\Bigl(\hat{h}_0^{(-k)}\circ T\Bigr)\Bigr|\\
        &\quad + \Bigl|P_{X,n}^{0,(k)}\Bigl(\hat{h}_0^{(-k)}\circ T\Bigr)
        - P_{X,n}^{0,(k)}\Bigl({h}_0\circ T\Bigr)\Bigr|\\
        &\quad+ \Bigl|P_{X,n}^{0,(k)}\Bigl({h}_0\circ T\Bigr)
        - P_X^0\bigl(h_0\circ T\bigr)\Bigr|\\
        &= (II.a)+(II.b)+(II.c)
    \end{aligned}
    $$
    (II.c) is the difference between the empirical average and the expectation for the function $h_0\circ T$ in the control group. Therefore, under assumption 2.(a), the LLN applies and gives $(II.c)=o_p(1)$.
    Now we turn to the terms that involve nuisance parameters $T$ and $h_0$.
    By the Lipschitz condition on the outcome regression $\hat{h}_0^{(-k)}$, 
    $$
    \begin{aligned}
        (II.a) &\leq \frac{1}{\bigl|\Omega_0^{(k)}\bigr|}\sum_{i\in \Omega_0^{(k)}} \Bigl| 
        \hat{h}_0^{(-k)}\circ \hat{T}^{(-k)}(X_i)- \hat{h}_0^{(-k)}\circ T(X_i)
        \Bigr|\\
        &\leq \frac{L^{(-k)}}{\bigl|\Omega_0^{(k)}\bigr|}\sum_{i\in \Omega_0^{(k)}}\|\hat{T}^{(-k)}(X_i) - T(X_i)\|\\
        &\leq L^{(-k)} \Bigl(\frac{1}{\bigl|\Omega_0^{(k)}\bigr|}\sum_{i\in \Omega_0^{(k)}}\|\hat{T}^{(-k)}(X_i) - T(X_i)\|^2\Bigr)^{1/2}
    \end{aligned}
    $$
    where the third line is by the Cauchy-Schwarz inequality. The constant $L^{(-k)}$ is measurable with respect to $\Omega^{(-k)}$, so the conditional argument below is unaffected. Under moment condition 2.(b), cross fitting (independent samples when holding $\Omega^{(-k)}$ and therefore $\hat{T}^{(-k)}$ fixed) and the law of large numbers,
    $$
    \begin{aligned}
        (II.a)^2 &\lesssim (L^{(-k)})^2 \|\hat{T}^{(-k)} - T\|^2_{L_2(P_X^0)}+o_p(1)
        = o_p(1).
    \end{aligned}
    $$
    The last equality follows from the $L_2$ consistency of $\hat{T}^{(-k)}$, since $(L^{(-k)})^2 = O_p(1)$ and the product $O_p(1)\, o_p(1)$ is $o_p(1)$ by Slutsky. Therefore $(II.a)=o_p(1)$.

    Finally, we consider $(II.b)$, which is the validation error on the outcome regression estimator. 
    Notably, if $X\sim P_X^0$ and $T(X)\sim P_X^1$. Therefore, $\frac{1}{|\Omega_0^{(k)}|}\sum_{i\in\Omega_0^{(k)}}\delta_{T(X_i)}\Rightarrow P_X^1$ as sample size goes to infinity. Applying similar condition for $(II.a)$, we can show
    $$
    \begin{aligned}
        (II.b) &\leq \frac{1}{\bigl|\Omega_0^{(k)}\bigr|}\sum_{i\in \Omega_0^{(k)}} \Bigl| 
        \hat{h}_0^{(-k)}\circ {T}(X_i)- h_0\circ T(X_i)
        \Bigr|\\
        &\leq \Bigl(\frac{1}{\bigl|\Omega_0^{(k)}\bigr|}\sum_{i\in \Omega_0^{(k)}}\|\hat{h}_0^{(-k)}\circ {T}(X_i) - h_0\circ T(X_i)\|^2\Bigr)^{1/2}.\\
    \end{aligned}
    $$
    Under moment condition 2.(b), LLN, and asssumption 1, we can show $(II.b)=o_p(1)$
    $$
    (II.b)^2\lesssim\|\hat{h}_0^{(-k)} - h_0\|^2_{L_2(P_X^1)} + o_p(1)= o_p(1)
    $$
    Therefore, bring all terms to consideration, the consistency can be shown by first principle: for any $\epsilon>0$, as $n\to +\infty$
    $$
    P(|\hat{\tau}_{ATT}^{(k)} - \tau_{ATT}|>\epsilon) \leq P(Z_n>\epsilon) \to 0 \quad \text{for } Z_n=o_p(1).
    $$
    The overall estimator $\hat{\tau}_{ATT} - \tau_{ATT} = \frac{1}{K}\sum_{k=1}^K \hat{\tau}_{ATT}^{(k)} - \tau_{ATT} = o_p(1)$.
\end{proof}

\textbf{Notation.} 
Throughout, observations are i.i.d.\ tuples $(Y, A, X)$ with $Y\in\mathbb R$ the outcome, $A\in\{0,1\}$ the treatment indicator,
and $X\in\mathcal X\subseteq\mathbb R^d$ the covariate. Let $P_X^0,P_X^1$ denote the conditional laws of $X$ given $A=0,1$,
respectively, and $r := \Pr(A=1)\in(0,1)$. 

The outcome regression is $h_0(x) := \E[Y\mid X=x, A=0]$, and $T:\mathcal X\to\mathcal X$ is the optimal transport map pushing $P_X^0$ to $P_X^1$ under the chosen cost (we work with the quadratic cost in the verification of condition (A1) below).

Cross-fitting splits the sample into $K$ folds $\Omega^{(1)},\dots,\Omega^{(K)}$ of equal size $n^{(k)} = n/K$, with 
$\Omega^{(-k)} := \bigcup_{k'\neq k}\Omega^{(k')}$ of size $|\Omega^{(-k)}| = (K-1)n/K$. For each fold $k$, $\hat h_0^{(-k)}$ and $\hat T^{(-k)}$ denote regression and transport-map estimators trained on $\Omega^{(-k)}$. Empirical averages over fold $k$ are written $P_n^{(k)}[\cdot] := n^{(k),-1}\sum_{i\in\Omega^{(k)}}(\cdot)$.

For functions $f:\mathcal X\to\mathbb R^q$, $\nabla f$ denotes the Jacobian / gradient and $\|f\|_{L_2(\nu)}^2 := \int\|f\|^2 d\nu$ for any probability measure $\nu$ on $\mathcal X$. We use $a^\top b$ for the standard Euclidean inner product.

\subsection{Influence Function and Auxiliary Lemmas}
\label{app:if-lemmas}

\textbf{Influence function and oracle variance.}
With the nuisances held fixed at their population values, the influence function of the target functional is derived as follows. We first derive the influence function of 
$$
\tau_{ATT} = 
\E[Y|A=1] - \E_{P_X^0}[h_0(T(X))]
$$
from which we can define the parameter functional 
$$
\Psi(P, h, T, r) = \tau_{ATT}
= \int y P(dy|A=1) - \int h_0(T(x)) P(dx|A=0).
$$
Now, we design a class of parametric submodels in the direction of observation $p=(a,x,y)$
$$
P_{\epsilon} = (1-\epsilon)P + \epsilon\delta_p
$$
following von Mises expansion, the Gateaux derivative (influence function) is
$$
\left.\frac{d}{d \epsilon} \Psi\left(P_\epsilon, h_0, T, r\right)\right|_{\epsilon=0}=\psi_{Y,h_0,T}(p)
.$$
Now, under the perturbation, the first term has the form
$$
P_{\epsilon, Y \mid A=1}(y)=\frac{(1-\epsilon) P_{Y, A=1}(y)+\epsilon a \delta_y(y)}{(1-\epsilon) P(A=1)+\epsilon a}
$$
taking derivatives gives
$$
\begin{aligned}
&\frac{d}{d \epsilon}\left[\frac{\int y d\left((1-\epsilon) P_{Y, A=1}(y)+\epsilon a \delta_y(y)\right)}{(1-\epsilon) r+\epsilon a}\right]_{\epsilon=0} \\
=&\frac{1}{r}[a y-E[Y \mid A=1]]-\frac{E[Y \mid A=1]}{r}(a-r) \\
=&\frac{a}{r}(y-E[Y \mid A=1]).
\end{aligned}
$$
For the second term, we similarly derive the distribution under perturbation 
$$
P_{\epsilon, X \mid A=0}(x)=\frac{(1-\epsilon) P_{X, A=0}(x)+\epsilon(1-a) \delta_x(x)}{(1-\epsilon)(1-r)+\epsilon(1-a)}
$$
and the corresponding derivative
$$
\begin{aligned}
&\frac{d}{d \epsilon}\left[\frac{\int h_0\left(T(x)\right) d\left((1-\epsilon) P_{X, A=0}(x)+\epsilon(1-a) \delta_x(x)\right)}{(1-\epsilon)(1-r)+\epsilon(1-a)}\right]_{\epsilon=0} \\
=&\frac{1}{1-r}\bigl[(1-a) h_0(T(x))-\E[h_0(T(X)) \mid A=0]\bigr]\\
&\quad -\frac{\E[h_0(T(X)) \mid A=0]}{1-r}\bigl((1-a)-(1-r)\bigr) \\
=&\frac{1-a}{1-r}\bigl(h_0(T(x))-\E_{P_X^0}[h_0(T(X))]\bigr).
\end{aligned}
$$
Combining the two terms gives rise to 
$$
\psi_{Y,h_0,T}(p)=\frac{a}{r}\left(y-\E_{P_Y^1}[Y]\right)-\frac{1-a}{1-r}\left(h_0\left(T(x)\right)-\E_{P_X^0}[h_0(T(X)]\right).
$$
In addition, we can easily verify that $\psi_{Y,h_0,T}(p)$ has $0$ mean:
$$\E[\psi_{Y,h_0,T}(p)] 
= \E[Y-\E_{P_Y^1}[Y]|A=1] - \E[h_0\left(T(x)\right)-\E_{P_X^0}[h_0(T(X)]|A=0]=0.
$$
From the influence function we derive the asymptotic variance
$$
\begin{aligned}
V_{\mathrm{or}} &= \operatorname{Var}[\psi_{Y,h_0,T}(p)]\\
&= \E[\psi_{Y,h_0,T}(p)^2]\\
&= \frac{1}{r^2} \E\left[A\left(Y-\E_{P_Y^1}[Y]\right)^2\right]+\frac{1}{(1-r)^2} \E\left[(1-A)\left(h_0\left(T(X)\right)-\E_{P_X^0}[h_0(T(X)]\right)^2\right]\\
&= \frac{1}{r} \operatorname{Var}(Y \mid A=1)+\frac{1}{1-r} \operatorname{Var}\left(h_0\left(T(X)\right) \mid A=0\right)
\end{aligned}
$$
where the cross term $A(1-A)$ cancels out.

\textbf{Delta-method expansion.}
Several proofs below expand H\'ajek ratios of the form
$$
P_n^{(k)}[f_n]\big/P_n^{(k)}[g],
$$
where $g$ is an arm indicator and $f_n$ is a fixed function conditional on $\Omega^{(-k)}$. We state the expansion once for both choices of $g$.
\smallskip
\begin{lemma}[delta-method expansion]
\label{lem:delta-method}
Let $f_n$ be measurable with $\E[f_n^2 \mid \Omega^{(-k)}] < \infty$,
and let $g \in \{A, 1-A\}$ with $\gamma := \E[g] > 0$. Define
$$
m_n := \E[f_n\mid\Omega^{(-k)}],\qquad
a_n := (P_n^{(k)} - P)[f_n],\qquad
b_n := P_n^{(k)}[g] - \gamma.
$$
Then
$$
\frac{P_n^{(k)}[f_n]}{P_n^{(k)}[g]}
= \underbrace{\frac{m_n}{\gamma}}_{\text{(I) zeroth order}}
+ \underbrace{\frac{a_n}{\gamma}}_{\text{(II) emp.\ process}}
- \underbrace{\frac{m_n b_n}{\gamma^2}}_{\text{(III) denom.\ corr.}}
+ R_n,
$$
with remainder satisfying
$$
|R_n| \le \frac{|a_n||b_n|}{\gamma^2}
            + \frac{(|m_n|+|a_n|)b_n^2}{\gamma^3}\cdot \frac{1}{|1+b_n/\gamma|}.
$$
\end{lemma}
\begin{proof} Algebraic:
$\frac{m_n + a_n}{\gamma+b_n}
= \frac{m_n+a_n}{\gamma}\cdot\frac{1}{1+b_n/\gamma}$,
expand the geometric factor as $1 - b_n/\gamma + b_n^2/(\gamma(\gamma+b_n))$
and collect.
\end{proof}
\smallskip
We note some standard inputs used repeatedly below: $\|\hat T^{(-k)} - T\|_{L_2(P_X^0)} = O_p(n^{-1/2})$, by condition (A1), asymptotic linearity at rate $n^{-1/2}$. This immediately gives $\|\hat T^{(-k)} - T\|_{L_2(P_X^0)}^2 = O_p(n^{-1})$.

\begin{lemma}[Fold aggregation]
\label{lem:fold-aggregation}
With equal fold sizes $n^{(k)} = n/K$ and $|\Omega^{(-k)}| = (K-1)n/K$, for
any function $\psi$ of the observation,
$$
\frac{1}{K}\sum_{k=1}^{K} P_n^{(k)}[\psi]
= \frac{1}{n}\sum_{i=1}^n \psi(p_i),
\qquad
\frac{1}{K}\sum_{k=1}^{K}
\frac{1}{|\Omega^{(-k)}|}\sum_{j\in\Omega^{(-k)}}\psi(p_j)
= \frac{1}{n}\sum_{j=1}^{n}\psi(p_j).
$$
\end{lemma}
\begin{proof}
The first identity is the definition of the average over equal folds. For
the second, each observation belongs to exactly $K-1$ of the sets
$\Omega^{(-k)}$, so its total coefficient equals
$(K-1)/(K\,|\Omega^{(-k)}|)$, which is $1/n$.
\end{proof}

\subsection{Proofs for Estimated-Map Bias and Variance}
\label{app:estimated-map-variance}

Throughout this appendix we extend the population functional notation to joint arguments. For a measurable $f$ of the pair $(X,Y)$, write $P^a f := \E[f(X,Y)\mid A=a]$. Its arm empirical counterpart is $P_n^{a} f := n_a^{-1}\sum_{i\in\Omega_a} f(X_i,Y_i)$. Its arm-fold counterpart is $P_n^{a,(k)} f := (n_a/K)^{-1}\sum_{i\in\Omega_a^{(k)}} f(X_i,Y_i)$, assuming equal arm-fold sizes.

\begin{lemma}[Consistency bridge]
\label{lem:consistency-bridge}
Conditions 1--3 of \Cref{thm:att_consistency} imply that
$$\max_k\|g_k - h_0\circ T\|_{L_2(P_X^0)} = o_p(1).$$
\end{lemma}

\begin{proof}
By the Lipschitz condition 3 of \Cref{thm:att_consistency} applied pointwise and the pushforward
identity $T_{\sharp}P_X^0 = P_X^1$,
$$
\begin{aligned}
\|g_k - h_0\circ T\|_{L_2(P_X^0)}
&\le \|\hat h_0^{(-k)}\circ\hat T^{(-k)} - \hat h_0^{(-k)}\circ T\|_{L_2(P_X^0)}
 + \|(\hat h_0^{(-k)} - h_0)\circ T\|_{L_2(P_X^0)}\\
&\le L^{(-k)}\|\hat T^{(-k)} - T\|_{L_2(P_X^0)}
 + \|\hat h_0^{(-k)} - h_0\|_{L_2(P_X^1)},
\end{aligned}
$$
and both terms are $o_p(1)$ by condition 1 together with
$L^{(-k)} = O_p(1)$. The maximum over finitely many folds is $o_p(1)$.
\end{proof}

\begin{lemma}[Cross term]
\label{lem:cross-term}
Under condition (A1) and condition (A4), with $L_2(P_X^1)$ replaced by
$L_2(Q^{\dagger})$,
$$
\bigl|(Q_k - Q^{\dagger})\bigl(\hat h_0^{(-k)} - h_0\bigr)\bigr|
= \bigl|\E_{P_X^0}\bigl[(\hat h_0^{(-k)} - h_0)(\hat T^{(-k)}(X))
 - (\hat h_0^{(-k)} - h_0)(T^{\dagger}(X))\bigr]\bigr|
= o_p(n^{-1/2}).
$$
\end{lemma}

\begin{proof}
Write $D(x) := (\hat h_0^{(-k)} - h_0)(\hat T^{(-k)}(x))
 - (\hat h_0^{(-k)} - h_0)(T^{\dagger}(x))$. By the mean value theorem there
is $s(x)\in[0,1]$ with
$D(x) = \nabla(\hat h_0^{(-k)} - h_0)(\xi_{s}(x))^{\top}
(\hat T^{(-k)}(x) - T^{\dagger}(x))$, where
$\xi_{s}(x) = T^{\dagger}(x) + s(x)(\hat T^{(-k)}(x) -
T^{\dagger}(x))$. By the gradient Lipschitz constants in (A4),
$$
\|\nabla(\hat h_0^{(-k)} - h_0)(\xi_{s}(x))\|
\le \|\nabla(\hat h_0^{(-k)} - h_0)(T^{\dagger}(x))\|
 + \bigl(M^{(-k)} + L^{\ast}_{\nabla h_0}\bigr)
   \|\hat T^{(-k)}(x) - T^{\dagger}(x)\|.
$$
Taking $\E_{P_X^0}$ and applying Cauchy--Schwarz to the first product, together with the pushforward identity
$$
\|\nabla(\hat h_0^{(-k)} - h_0)\circ T^{\dagger}\|_{L_2(P_X^0)}
= \|\nabla(\hat h_0^{(-k)} - h_0)\|_{L_2(Q^{\dagger})},
$$
$$
|\E_{P_X^0}[D(X)]|
\le \|\nabla(\hat h_0^{(-k)} - h_0)\|_{L_2(Q^{\dagger})}\,
\|\hat T^{(-k)} - T^{\dagger}\|_{L_2(P_X^0)}
 + \bigl(M^{(-k)} + L^{\ast}_{\nabla h_0}\bigr)
   \|\hat T^{(-k)} - T^{\dagger}\|_{L_2(P_X^0)}^2.
$$
The gradient consistency in (A4) gives $o_p(1)$ for the first factor.
Condition (A1) gives
$\|\hat T^{(-k)} - T^{\dagger}\|_{L_2(P_X^0)} = O_p(n^{-1/2})$, as noted
after the delta-method lemma,
and hence the squared norm is $O_p(n^{-1})$, with $M^{(-k)} = O_p(1)$ by
(A4). The total is $o_p(1)\,O_p(n^{-1/2}) + O_p(n^{-1}) = o_p(n^{-1/2})$.
\end{proof}

\begin{proof}[Proof of \Cref{prop:honest-split-map-variance}]
Conditional on $\calg_n$ and $\{A_i\}$, $g_n$ is a fixed function and the two
arm-wise sample means are independent, with $m_1, m_0$ nonrandom. Hence
$\E(\hat\tau_{ATT}\mid\calg_n,\{A_i\}) = \E_{P_Y^1}[Y] - P_X^0 g_n
= \tau_{ATT} - B_n$, giving \eqref{eq:conditional-bias}, and independence
gives \eqref{eq:conditional-variance}; the second form follows from
$\hat T(X)\sim Q_n$ when $X\sim P_X^0$. The law of total variance (given
$\{A_i\}$) yields \eqref{eq:unconditional-variance}, and the conditional
MSE decomposition
$\E[(\hat\tau - \tau)^2\mid\{A_i\}]
= \E[\operatorname{Var}(\hat\tau\mid\calg_n,\{A_i\})] + \E[B_n^2]$
yields \eqref{eq:unconditional-mse}. The decompositions
\eqref{eq:fold-drift} are algebraic:
$Q_n\hat h_0 - P_X^1 h_0 = (Q_n - P_X^1)h_0 + Q_n(\hat h_0 - h_0)
= (Q_n - P_X^1)\hat h_0 + P_X^1(\hat h_0 - h_0)$.
\end{proof}

\begin{proof}[Proof of \Cref{thm:balance-error-regimes}]
Fix a fold $k$ and recall
$$
\hat\tau_{ATT}^{(k)}
= \frac{P_n^{(k)}[A\,Y]}{P_n^{(k)}[A]}
 - \frac{P_n^{(k)}[(1-A)\,g_k]}{P_n^{(k)}[1-A]}.
$$
For the treated term, the expansion
$$
\frac{P_n^{(k)}[A Y]}{P_n^{(k)}[A]} - \E_{P_Y^1}[Y]
= \frac{1}{r}\,P_n^{(k)}\bigl[A\bigl(Y - \E_{P_Y^1}[Y]\bigr)\bigr]
 + o_p(n^{-1/2})
$$
follows from the delta-method expansion lemma of this appendix applied with
$f_n = AY$ and $g = A$.
For the control term, apply the delta-method expansion lemma of this appendix
with $f_n = (1-A)g_k$, conditional on $\Omega^{(-k)}$
(so $m_n = (1-r)P_X^0 g_k$):
$$
\frac{P_n^{(k)}[(1-A)g_k]}{P_n^{(k)}[1-A]}
= P_X^0 g_k
 + \frac{1}{1-r}\,(P_n^{(k)} - P)\bigl[(1-A)\bigl(g_k - P_X^0 g_k\bigr)\bigr]
 + o_p(n^{-1/2}),
$$
where the remainder control uses $\E[g_k^2\mid\Omega^{(-k)}] = O_p(1)$, which
follows from hypothesis 1 and $g_{\star}\in L_2(P_X^0)$. Next, replace $g_k$
by $g_{\star}$ inside the empirical process: conditional on $\Omega^{(-k)}$
the summands of
$(P_n^{(k)} - P)[(1-A)\{(g_k - P_X^0 g_k) - (g_{\star} - P_X^0 g_{\star})\}]$
are i.i.d.\ and centred with conditional variance at most
$(1-r)\|g_k - g_{\star}\|^2_{L_2(P_X^0)}/1 = o_p(1)$, so by the conditional
Chebyshev argument used in the proof of \Cref{thm:att_consistency}
this term is $o_p(n^{-1/2})$. Collecting terms and subtracting
$\tau_{ATT} = \E_{P_Y^1}[Y] - P_X^1 h_0$,
$$
\hat\tau_{ATT}^{(k)} - \tau_{ATT}
= P_n^{(k)}\bigl[\psi_{g_{\star}}\bigr] - B_k + o_p(n^{-1/2}),
$$
with $\psi_{g_{\star}}$ as in \eqref{eq:frozen-nuisance-if}. Averaging over
folds with \Cref{lem:fold-aggregation}
gives \eqref{eq:master-expansion}. Since $\psi_{g_{\star}}(p_i)$ are i.i.d.,
mean zero, with finite variance (hypothesis 2 and
$g_{\star}\in L_2(P_X^0)$), the i.i.d.\ CLT applies to the first term.
(R1) and (R2) follow by Slutsky; the tightness claim in (R3) from
$O_p(1) + O_p(1)$; the AL case of (R3) by absorbing
\eqref{eq:drift-al} into the i.i.d.\ sum, yielding influence function
$\psi_{g_{\star}} - \psi_B$ and \eqref{eq:total-variance}; and (R4) by
multiplying \eqref{eq:master-expansion} through by $n^{\alpha - 1/2}$.
When $g_{\star} = h_0\circ T$ with $T$ exact,
$P_X^0\{(h_0\circ T)^j\} = P_X^1\{h_0^j\}$ for $j = 1, 2$, which gives
\eqref{eq:oracle-variance} together with its invariance to the selected exact transport map.
\end{proof}

\begin{proof}[Proof of \Cref{cor:al-transport}]
Decompose
$B_{T,k} = \bigl[P_X^0\{h_0(\hat T^{(-k)}(X))\}
 - P_X^0\{h_0(T^{\dagger}(X))\}\bigr] + b^{\dagger}$.
By Taylor's theorem with the Lipschitz gradient of $h_0$ from (A4),
$$
\begin{aligned}
&\Bigl|P_X^0\{h_0(\hat T^{(-k)}(X)) - h_0(T^{\dagger}(X))\}
 - P_X^0\bigl[\nabla h_0(T^{\dagger}(X))^{\top}
   (\hat T^{(-k)}(X) - T^{\dagger}(X))\bigr]\Bigr|\\
&\qquad\le \frac{L^{\ast}_{\nabla h_0}}{2}
 \|\hat T^{(-k)} - T^{\dagger}\|^2_{L_2(P_X^0)}
= O_p(n^{-1}),
\end{aligned}
$$
using $\|\hat T^{(-k)} - T^{\dagger}\|_{L_2(P_X^0)} = O_p(n^{-1/2})$ from
(A1).
Substituting \eqref{eq:map-al} into the linear term and exchanging
integration and summation by Fubini's theorem, justified by the integrated
second moment in (A1), gives
$$
B_{T,k}
= b^{\dagger}
 + \frac{1}{|\Omega^{(-k)}|}\sum_{j\in\Omega^{(-k)}}\psi_T(X_j, A_j)
 + o_p(n^{-1/2}),
$$
where the remainder collects the Taylor error and
$|P_X^0[\nabla h_0(T^{\dagger})^{\top} r_n^{T,(k)}]|
\le L_h\|r_n^{T,(k)}\|_{L_2(P_X^0)} = o_p(n^{-1/2})$. The kernel
$\psi_T$ is mean zero by the kernel mean-zero property in (A1) and Fubini, with
$\E[\psi_T^2] \le L_h^2\int_{\calx}\E\|\phi_T(x;X,A)\|^2\,dP_X^0(x) < \infty$
by Jensen's inequality applied to the inner $P_X^0$-average
($\psi_T^2 \le L_h^2\,\E_{X'\sim P_X^0}\|\phi_T(X';x,a)\|^2$)
followed by Tonelli and the integrated second moment; the same two steps give the integrability
needed for the Fubini exchange and, applied to the leading term of
\eqref{eq:map-al}, the input
$\|\hat T^{(-k)} - T^{\dagger}\|_{L_2(P_X^0)}^2 = O_p(n^{-1})$.
Averaging over folds with \Cref{lem:fold-aggregation}
yields \eqref{eq:map-drift-al}.
\end{proof}

\begin{proof}[Proof of \Cref{thm:att_normality}]
The hypotheses of \Cref{thm:balance-error-regimes} hold with
$g_{\star} = h_0\circ T$: the $L_2$ consistency of both nuisances follows
from (A1) and (A3), and \Cref{lem:consistency-bridge} applies under
the Lipschitz condition in (A4), so hypothesis 1 holds; hypothesis 2 is part
of (A5). For $B_{T,k}$, \Cref{cor:al-transport} gives
$K^{-1}\sum_{k} B_{T,k} = b^{\dagger} + n^{-1}\sum_{i} \psi_T(X_i, A_i)
+ o_p(n^{-1/2})$, and (A2) gives $b^{\dagger} = 0$. For $B_{h,k}$, decompose
$B_{h,k} = P_X^1\bigl(\hat h_0^{(-k)} - h_0\bigr)
+ (Q_k - P_X^1)\bigl(\hat h_0^{(-k)} - h_0\bigr)$.
The first term is bounded in absolute value by
$\|\hat h_0^{(-k)} - h_0\|_{L_2(P_X^1)}$, which is $o_p(n^{-1/2})$ by
(A3). The second term is $o_p(n^{-1/2})$ by \Cref{lem:cross-term} with
$T^{\dagger} = T$ and $Q^{\dagger} = P_X^1$. Hence \eqref{eq:drift-al} holds
with $\psi_B = \psi_T$, regime (R3) of \Cref{thm:balance-error-regimes}
applies with $\psi_{g_{\star}} = \psi_{Y,h_0,T}$, and
\eqref{eq:total-variance} evaluates to \eqref{eq:combined-variance} using
$\E[\psi_{Y,h_0,T}^2] = V_{\mathrm{or}}$ from the influence function
derivation in this appendix.
\end{proof}

\begin{proof}[Proof of \Cref{cor:interpolated-ste}]
Throughout write $\eta = \eta_n$, $T_{\eta} = \eta\,\mathrm{Id} + (1-\eta)T$,
$P_X^{\eta} = (T_{\eta})_{\sharp}P_X^0$,
$Q_{k,\eta} = (\hat T^{(-k)}_{\eta})_{\sharp}P_X^0$, and
$\rho_n := \max_k\|\nabla\hat h_0^{(-k)} - \nabla h_0\|_{L_2(P_X^1)}$, which
is $o_p(1)$ by (A4).

\emph{Step 0: the expansion \eqref{eq:master-expansion} applies.} By the triangle inequality, the
Lipschitz condition of (A4), and
$\hat T^{(-k)}_{\eta} - T = \eta(\mathrm{Id} - T)
+ (1-\eta)(\hat T^{(-k)} - T)$,
$$
\|g_{k,\eta} - h_0\circ T\|_{L_2(P_X^0)}
\le L^{(-k)}\bigl[\eta\,W_2
 + (1-\eta)\|\hat T^{(-k)} - T\|_{L_2(P_X^0)}\bigr]
 + \|\hat h_0^{(-k)} - h_0\|_{L_2(P_X^1)} = o_p(1),
$$
so \Cref{thm:balance-error-regimes} yields
\eqref{eq:master-expansion} with $g_{\star} = h_0\circ T$ and bias term
$\overline B_{n,\eta} = K^{-1}\sum_k B_{k,\eta}$.

\emph{Step 1: deterministic bias.} By Taylor's theorem with the Lipschitz
gradient of $h_0$ and $T_{\eta}(x) - T(x) = \eta(x - T(x))$,
$$
b_{\eta} = P_X^0\bigl[h_0(T_{\eta}(X)) - h_0(T(X))\bigr]
= \eta\,c_1 + R_{\eta},
$$
with
$$
|R_{\eta}| \le \frac{L^{\ast}_{\nabla h_0}}{2}\,\eta^2\,
\E_{P_X^0}\|X - T(X)\|^2
= \frac{L^{\ast}_{\nabla h_0}}{2}\,\eta^2 W_2^2.
$$

\emph{Step 2: map noise.} With
$\hat T^{(-k)}_{\eta} - T_{\eta} = (1-\eta)(\hat T^{(-k)} - T)$, the Taylor
argument of the proof of \Cref{cor:al-transport}, expanded around
$T_{\eta}$, gives
$$
(Q_{k,\eta} - P_X^{\eta})h_0
= (1-\eta)\,P_X^0\bigl[\nabla h_0(T_{\eta}(X))^{\top}
(\hat T^{(-k)} - T)(X)\bigr]
+ O_p\bigl((1-\eta)^2\|\hat T^{(-k)} - T\|_{L_2(P_X^0)}^2\bigr).
$$
Replacing $\nabla h_0(T_{\eta}(X))$ by $\nabla h_0(T(X))$ costs at most
$L^{\ast}_{\nabla h_0}\eta\,W_2\cdot
\|\hat T^{(-k)} - T\|_{L_2(P_X^0)}$, which is $O_p(\eta\, n^{-1/2})$ by
Cauchy--Schwarz, and replacing the factor $(1-\eta)$ by $1$ costs
$\eta\cdot O_p(n^{-1/2})$ after substituting the linear expansion. Hence,
substituting \eqref{eq:map-al} and aggregating folds with
\Cref{lem:fold-aggregation} as in the proof of \Cref{cor:al-transport},
$$
\frac{1}{K}\sum_k (Q_{k,\eta} - P_X^{\eta})h_0
= \frac{1}{n}\sum_{i=1}^n\psi_T(X_i, A_i)
+ O_p(\eta\,n^{-1/2}) + o_p(n^{-1/2}).
$$

\emph{Step 3: the regression term.} Decompose
$Q_{k,\eta}(\hat h_0^{(-k)} - h_0)
= P_X^1(\hat h_0^{(-k)} - h_0)
+ (P_X^{\eta} - P_X^1)(\hat h_0^{(-k)} - h_0)
+ (Q_{k,\eta} - P_X^{\eta})(\hat h_0^{(-k)} - h_0)$.
The first term is bounded in absolute value by
$\|\hat h_0^{(-k)} - h_0\|_{L_2(P_X^1)} = o_p(n^{-1/2})$ by (A3).
For the other two, the gradient bound of \Cref{lem:cross-term}
applied along the segments from $T(x)$ to $T_{\eta}(x)$ and from
$T_{\eta}(x)$ to $\hat T^{(-k)}_{\eta}(x)$ gives, for any point $z$ on
either segment,
$$
\begin{aligned}
\|\nabla(\hat h_0^{(-k)} - h_0)(z)\|
&\le \|\nabla(\hat h_0^{(-k)} - h_0)(T(x))\|\\
&\quad + \bigl(M^{(-k)} + L^{\ast}_{\nabla h_0}\bigr)
\bigl[\eta\|x - T(x)\| + \|\hat T^{(-k)}(x) - T(x)\|\bigr],
\end{aligned}
$$
and therefore, by the mean value theorem and Cauchy--Schwarz,
\begin{align*}
\bigl|(P_X^{\eta} - P_X^1)(\hat h_0^{(-k)} - h_0)\bigr|
&\le \bigl[\rho_n + O_p(\eta)\bigr]\cdot \eta\,W_2
= O_p(\eta\,\rho_n) + O_p(\eta^2),\\
\bigl|(Q_{k,\eta} - P_X^{\eta})(\hat h_0^{(-k)} - h_0)\bigr|
&\le (1-\eta)\bigl[\rho_n + O_p(\eta)
+ O_p(n^{-1/2})\bigr]\cdot
\|\hat T^{(-k)} - T\|_{L_2(P_X^0)}\\
&= o_p(n^{-1/2}).
\end{align*}

\emph{Step 4: collect.} Combining Steps 1 to 3,
\begin{equation}
\overline B_{n,\eta}
= \eta\,c_1
+ \frac{1}{n}\sum_{i=1}^n\psi_T(X_i,A_i)
+ O_p(\eta^2) + O_p(\eta\,\rho_n) + O_p(\eta\,n^{-1/2}) + o_p(n^{-1/2}).
\label{eq:interp-drift-expansion}
\end{equation}
\emph{Part (i).} If $\eta = o(n^{-1/2})$, every $\eta$ term above is
$o_p(n^{-1/2})$, using $\rho_n = o_p(1)$, so
\eqref{eq:drift-al} holds with $\psi_B = \psi_T$ and regime (R3)
of \Cref{thm:balance-error-regimes} gives the conclusion with $V$ as in
\eqref{eq:combined-variance}.
\emph{Part (ii).} With $\sqrt n\,\eta \to \kappa$, the term
$\sqrt n\,\eta\,c_1 \to \kappa c_1$, while the remainders satisfy
$\sqrt n\,[O_p(\eta^2) + O_p(\eta\rho_n) + O_p(\eta n^{-1/2})] = o_p(1)$;
the expansion \eqref{eq:master-expansion} with the asymptotically linear
part of $\overline B_{n,\eta}$ absorbed as in regime (R3) then gives the $N(-\kappa c_1, V)$ limit. For
the consistency of $\hat c_1$, write $\hat c_1 - c_1$ as the empirical
fluctuation of the fixed function
$\nabla h_0(T)^{\top}(\mathrm{Id} - T)$ plus nuisance replacement errors,
and bound the latter as in Step 3:
$$
\hat c_1 - c_1
= O_p(n^{-1/2})
+ O_p\bigl(\rho_n\bigr)
+ O_p\bigl(\textstyle\max_k\|\hat T^{(-k)} - T\|_{L_2(P_X^0)}\bigr)
= o_p(1),
$$
where the empirical part uses
$\E_{P_X^0}\|X - T(X)\|^2 = W_2^2 < \infty$ and the conditional Chebyshev
argument. Slutsky's theorem then yields the recentered statement.
\emph{Part (iii).} The debiased estimator has bias term
$\overline B_{n,\eta} - \eta\,\hat c_1$, and by
\eqref{eq:interp-drift-expansion} and the $\hat c_1$ expansion,
$$
\sqrt n\bigl(\overline B_{n,\eta} - \eta\hat c_1\bigr)
= \frac{1}{\sqrt n}\sum_i\psi_T(X_i,A_i)
+ \sqrt n\,O_p(\eta^2)
+ \sqrt n\,O_p(\eta\,\rho_n)
+ \sqrt n\,O_p(\eta\, n^{-1/2}) + o_p(1).
$$
The three explicit remainders are $o_p(1)$ by, respectively,
$\eta = o(n^{-1/4})$, the product condition
\eqref{eq:interp-product-condition}, and $\eta \to 0$; regime (R3) applies
as in part (i).
\end{proof}

\begin{proof}[Proof of \Cref{cor:att_consistency}]
    We adapt the proof from Theorem \ref{thm:att_consistency}, as the estimator changed from using $\hat{T}$ to $\hat{T}_{\eta}=\eta\,\mathrm{Id}+ (1-\eta)\hat{T}$ with $\eta\in[0,1]$.
    $$
    \begin{aligned}
        |\hat{\tau}_{ATT}^{\eta,(k)} - \tau_{ATT}| &\leq |E_{P_{Y,n}^{1,(k)}}[Y] - E_{P_{Y}^{1}}[Y]| +\Bigl|\Bigl( P_{X,n}^{0,(k)}\Bigl(\hat{h}_0^{(-k)}\circ \hat{T}_{\eta}^{(-k)}\Bigr) - P_X^0\bigl(h_0\circ T\bigr)\Bigr)\Bigr|\\
        &= (I) + (II)
    \end{aligned}
    $$
    where
    $$
    \begin{aligned}
        (II) &= \left|
        P_{X,n}^{0,(k)}\bigl(\hat{h}_0^{(-k)}\circ \hat{T}_{\eta}^{(-k)}\bigr) - P_{X,n}^{0,(k)}\bigl(\hat{h}_0^{(-k)}\circ T\bigr)
        \right|\\ 
        &+ 
        \bigl|P_{X,n}^{0,(k)}\bigl(\hat{h}_0^{(-k)}\circ T\bigr) - P_{X,n}^{0,(k)}\bigl({h}_0\circ T\bigr)\bigr|
        +
        \bigl|P_{X,n}^{0,(k)}\bigl({h}_0\circ T\bigr) - P_X^0\bigl({h}_0\circ T\bigr)\bigr|\\
        &:= (II.a) + (II.b) + (II.c)
    \end{aligned}
    $$
    As $(II.b), (II.c)$ are the same as in Thm \ref{thm:att_consistency}, we only focus on $(II.a)$. Under Lipchitz condition and Cauchy-Schwartz
    $$
    \begin{aligned}
        (II.a) &\leq 
        \frac{L^{(-k)}}{|\Omega_0^{(k)}|}\sum_{i\in \Omega_0^k}\left\|
         \hat{T}_{\eta}^{(-k)}(X_i) -  T(X_i)
        \right\|\\
        &= L^{(-k)}\Bigr(\frac{1}{|\Omega_0^{(k)}|}\sum_{i\in \Omega_0^k}\left\|
         \hat{T}_{\eta}^{(-k)}(X_i) -  T(X_i)
        \right\|^2\Bigl)^{1/2}
    \end{aligned}
    $$
    then with independent samples of $\Omega_0^{(k)}$ and $\Omega_0^{(-k)}$, apply LLN to $\Omega_0^{(k)}$ we obtain that $$\frac{1}{|\Omega_0^{(k)}|}\sum_{i\in \Omega_0^k}\left\|
         \hat{T}_{\eta}^{(-k)}(X_i) -  T(X_i)
        \right\|^2 = \|\hat{T}_{\eta}^{(-k)} -  T\|_{L^2(P_X^0)}^2 + o_p(1).$$ 
        Consequently, under Continuous mapping theorem,
        $$
        \left(\frac{1}{|\Omega_0^{(k)}|}\sum_{i\in \Omega_0^k}\left\|
         \hat{T}_{\eta}^{(-k)}(X_i) -  T(X_i)
        \right\|^2\right)^{1/2} = \|\hat{T}_{\eta}^{(-k)} -  T\|_{L^2(P_X^0)} + o_p(1)
        $$
    Note that $W_2^2(P_X^0, P_X^1) = \int \|x-T(x)\|^2 dP_X^0<\infty$ under Brenier's theorem.
    With triangle inequality,
    $$
    \begin{aligned}
        \|\hat{T}_{\eta}^{(-k)} -  T\|_{L^2(P_X^0)} 
    &\leq \eta\|I - T\|_{L^2(P_X^0)} + (1-\eta)\|\hat{T}^{(-k)} - T\|_{L^2(P_X^0)}\\
    &= \eta W_2(P_X^0, P_X^1) + (1-\eta)\cdot o_p(1)
    \end{aligned}
    $$
    This reveals the trade-off for the interpolated estimator. When $W_2(P_X^0,P_X^1)=0$, the two covariate distributions coincide and $\eta$ may take any value in the unit interval. Otherwise the interpolation leaves a persistent bias unless the shrinkage toward the identity vanishes, that is $\eta:=\eta_n\to 0$ as required by condition 4. This matches the population bias bound $L_h\,\eta\, W_2(P_X^0, P_X^1)$ established in the interpolation section. Under condition 4 the estimator is therefore consistent.
    
\end{proof}

The following corollary characterizes the map influence term when the map is the plug-in of an exact-map functional. \Cref{sec:role-of-map} summarizes its consequences.

\begin{corollary}[Cancellation for an exact-map functional]
\label{cor:exact-map-cancellation}
Hold $h_0$ fixed and known, and let $\theta := P_X^1 h_0$. Let $\calm$ be a
local statistical model of pairs of covariate laws containing
$(P_X^0, P_X^1)$. Suppose:
\begin{enumerate}
\item[(i)] $\hat T^{(-k)}$ is the plug-in of a map-valued functional
$\calt(Q_0, Q_1)$ that is exact on the model,
$\calt(Q_0, Q_1)_{\sharp}Q_0 = Q_1$ for every pair $(Q_0, Q_1) \in \calm$,
with $\calt(P_X^0, P_X^1) = T$;
\item[(ii)] $\calt$ is pathwise differentiable along parametric submodels of
$\calm$, and the kernel $\phi_T$ of condition (A1) represents its pathwise
derivative;
\item[(iii)] either $\calm$ is locally nonparametric, or $\psi_T$ lies in
the tangent space of $\calm$ at $(P_X^0, P_X^1)$.
\end{enumerate}
Then the pushforward constraint determines the map influence term,
\begin{equation}
\psi_T(x, a)
= \frac{a}{r}\bigl(h_0(x) - \theta\bigr)
 - \frac{1-a}{1-r}\bigl(h_0(T(x)) - \theta\bigr),
\label{eq:exact-map-compatibility}
\end{equation}
the influence function collapses to
\begin{equation}
\psi_{Y,h_0,T}(p) - \psi_T(x,a)
= \frac{a}{r}\bigl(y - h_0(x) - \tau_{ATT}\bigr),
\label{eq:oracle-exact-map-if}
\end{equation}
and the asymptotic variance satisfies the exact identity
\begin{equation}
\begin{aligned}
V^{\mathrm{or}}_{\mathrm{exact}}
&= \frac{1}{r}\operatorname{Var}\bigl(Y - h_0(X)\,\big|\,A=1\bigr)\\
&= V_{\mathrm{or}}
 - \frac{\operatorname{Var}\{h_0(X)\mid A=1\}}{r(1-r)}
 - \frac{2}{r}\operatorname{Cov}\bigl(h_1(X) - h_0(X),\, h_0(X)\,\big|\,A=1\bigr).
\end{aligned}
\label{eq:oracle-exact-map-variance}
\end{equation}
All three displays are independent of which exact-map functional is used.
\end{corollary}

The corollary separates two effects. Estimating the map through an
exact-map functional removes the control-arm imputation variance
$\operatorname{Var}\{h_0(X)\mid A=1\}/(1-r)$ from the oracle variance
entirely. At the same time, it replaces the treated-arm outcome variance by
the residual variance $\operatorname{Var}(Y - h_0(X)\mid A=1)$. Whether the
net effect is a reduction depends on the treatment effect function.
The identity shows that
$V^{\mathrm{or}}_{\mathrm{exact}} \le V_{\mathrm{or}}$ holds if and only if
$$
\operatorname{Cov}\bigl(h_1(X) - h_0(X),\, h_0(X)\,\big|\,A=1\bigr)
\ge -\frac{\operatorname{Var}\{h_0(X)\mid A=1\}}{2(1-r)}.
$$
Under a constant conditional treatment effect the covariance vanishes and
the reduction is strict, with size
$\operatorname{Var}\{h_0(X)\mid A=1\}/(r(1-r))$.

Hypothesis (iii) is where efficiency enters. It concerns the estimation of
the nuisance parameters of the model $\calm$, not efficiency for
$\tau_{ATT}$. In a parametric model, the plug-in of asymptotically
efficient parameter estimates has an influence function in the tangent
space, so hypothesis (iii) holds. The Gaussian map built from
$(\hat\mu_a, \hat\Sigma_a)$ is an instance, since these estimates are
efficient under Gaussianity. Without
hypothesis (iii), the argument determines only the projection of $\psi_T$
onto the tangent space of $\calm$.

None of this changes the semiparametric efficiency bound for the ATT. In
the unrestricted observed-data model, the efficient influence function of
\citet{hahn1998role} does not depend on the choice of map.
\Cref{cor:exact-map-cancellation} holds $h_0$ fixed. When $h_0$ is
estimated, its contribution must be accounted for as in
\Cref{rem:root-n-regression}.

\begin{proof}[Proof of \Cref{cor:exact-map-cancellation}]
For every regular parametric submodel $t\mapsto(P^0_{X,t}, P^1_{X,t})$ in
$\calm$, the exactness hypothesis (i) gives the identity
\begin{equation}
P^0_{X,t}\bigl[h_0\bigl(\calt(P^0_{X,t}, P^1_{X,t})(X)\bigr)\bigr]
= P^1_{X,t}\,h_0.
\label{eq:submodel-pushforward-identity}
\end{equation}
Differentiate at $t = 0$. The derivative of the left side with the map held
fixed is represented by the influence function
$\frac{1-a}{1-r}\{h_0(T(x)) - \theta\}$; the derivative through the map is
represented by $\psi_T$, which is hypothesis (ii); and the right side is
represented by $\frac{a}{r}\{h_0(x) - \theta\}$. Hence for every submodel
score $s$ in the tangent space of $\calm$,
$$
\E\Bigl[\Bigl(
\tfrac{1-a}{1-r}\{h_0(T(x)) - \theta\} + \psi_T(x,a)
 - \tfrac{a}{r}\{h_0(x) - \theta\}\Bigr)\, s(p)\Bigr] = 0.
$$
The bracketed function is mean zero. Under hypothesis (iii) it is orthogonal
to a tangent space that contains it, in the locally nonparametric case, or
it lies in the tangent space together with $\psi_T$, in the efficient
plug-in case. In either case it vanishes identically, which proves
\eqref{eq:exact-map-compatibility}. Substituting
\eqref{eq:exact-map-compatibility} into $\psi_{Y,h_0,T} - \psi_T$ cancels
the control-arm terms and yields \eqref{eq:oracle-exact-map-if}; its
variance is the first expression in \eqref{eq:oracle-exact-map-variance},
using $\E[Y - h_0(X)\mid A=1] = \tau_{ATT}$.

For the identity, write $\Delta := \psi_{Y,h_0,T} - \psi_T$, so that
$V_{\mathrm{or}} = \operatorname{Var}(\psi_{Y,h_0,T})
= \operatorname{Var}(\Delta) + \operatorname{Var}(\psi_T)
+ 2\operatorname{Cov}(\Delta, \psi_T)$. The cross products between treated
and control terms vanish since $A(1-A) = 0$, and the pushforward identity
gives $\operatorname{Var}_{P_X^0}\{h_0(T(X))\}
= \operatorname{Var}\{h_0(X)\mid A=1\}$. Hence, from the identity
\eqref{eq:exact-map-compatibility},
$$
\operatorname{Var}(\psi_T)
= \frac{\operatorname{Var}\{h_0(X)\mid A=1\}}{r}
+ \frac{\operatorname{Var}\{h_0(X)\mid A=1\}}{1-r}
= \frac{\operatorname{Var}\{h_0(X)\mid A=1\}}{r(1-r)},
$$
and, from \eqref{eq:oracle-exact-map-if},
$$
\begin{aligned}
\operatorname{Cov}(\Delta, \psi_T)
&= \frac{1}{r^2}\,\E\bigl[A\,(Y - h_0(X) - \tau_{ATT})(h_0(X) - \theta)\bigr]\\
&= \frac{1}{r}\operatorname{Cov}\bigl(h_1(X) - h_0(X),\, h_0(X)\mid A=1\bigr),
\end{aligned}
$$
where the last step conditions on $X$, uses $\E[Y\mid X, A=1] = h_1(X)$,
$\E[h_1(X) - h_0(X)\mid A=1] = \tau_{ATT}$ and
$\E[h_0(X) - \theta\mid A=1] = 0$. Rearranging
$\operatorname{Var}(\Delta) = V_{\mathrm{or}} - \operatorname{Var}(\psi_T)
- 2\operatorname{Cov}(\Delta, \psi_T)$ gives
\eqref{eq:oracle-exact-map-variance}.
\end{proof}

\begin{prop}[Exact empirical-plan cancellation]
\label{prop:empirical-plan-cancellation}
Let
$$
\hat P_X^0 = \sum_{i\in\Omega_0} a_i\delta_{X_i},
\qquad
\hat P_X^1 = \sum_{j\in\Omega_1} b_j\delta_{X_j}
$$
be weighted empirical measures with positive weights, and let
$\hat\pi = (\hat\pi_{ij}) \in \Pi(\hat P_X^0, \hat P_X^1)$ be an arbitrary
feasible plan. For every measurable $g$ with finite sums,
\begin{equation}
\sum_{i,j}\hat\pi_{ij}\, g(X_j) = \sum_{j\in\Omega_1} b_j\, g(X_j)
= \hat P_X^1 g.
\label{eq:plan-cancellation}
\end{equation}
Consequently, the same-sample, plan-integrated ATT estimator satisfies the
pathwise identity
\begin{equation}
\hat\tau_{ATT,\hat\pi}
= \sum_{j\in\Omega_1} b_j Y_j
 - \sum_{i,j}\hat\pi_{ij}\,\hat h_0(X_j)
= \sum_{j\in\Omega_1} b_j\bigl(Y_j - \hat h_0(X_j)\bigr):
\label{eq:same-sample-plan-att}
\end{equation}
neither the selected coupling nor its cost affects this number. Consider
instead the barycentric projection
$\overline T_{\hat\pi}(X_i) = \sum_j w_{ij}X_j$ with weights
$w_{ij} = \hat\pi_{ij}/a_i$. The discrepancy is
\begin{equation}
\hat P_X^0\bigl(g\circ\overline T_{\hat\pi}\bigr) - \hat P_X^1 g
= \sum_{i\in\Omega_0} a_i\Bigl[
g\Bigl(\sum_j w_{ij}X_j\Bigr) - \sum_j w_{ij}\,g(X_j)\Bigr],
\label{eq:barycentric-gap}
\end{equation}
which is zero for affine $g$, nonpositive for convex $g$, and, if
$g\in C^2$ with $\sup_x\|\nabla^2 g(x)\|_{\mathrm{op}}\le M$, bounded by
\begin{equation}
\bigl|\hat P_X^0(g\circ\overline T_{\hat\pi}) - \hat P_X^1 g\bigr|
\le \frac{M}{2}\sum_{i,j}\hat\pi_{ij}
\bigl\|X_j - \overline T_{\hat\pi}(X_i)\bigr\|^2.
\label{eq:barycentric-gap-bound}
\end{equation}
\end{prop}

\begin{proof}[Proof of \Cref{prop:empirical-plan-cancellation}]
The second marginal constraint gives
$\sum_{i,j}\hat\pi_{ij}g(X_j) = \sum_j g(X_j)\sum_i\hat\pi_{ij}
= \sum_j b_j g(X_j)$, which is \eqref{eq:plan-cancellation} and, applied to
$g = \hat h_0$, \eqref{eq:same-sample-plan-att}. For the barycentric
statement, Jensen's inequality applied row-wise to the weights $w_{ij}$ gives
the sign for convex $g$ and equality for affine $g$. If $g\in C^2$ with
bounded Hessian, Taylor-expand $g(X_j)$ around
$\overline T_{\hat\pi}(X_i)$: the first-order term vanishes because
$\sum_j w_{ij}\{X_j - \overline T_{\hat\pi}(X_i)\} = 0$, and the remainder
gives \eqref{eq:barycentric-gap-bound}.
\end{proof}

\begin{prop}[Training-balance identity and cross-fit stability]
\label{prop:training-balance-identity}
Let $P_{X,n}^{a,(k)}$ and $P_{X,n}^{a,(-k)}$ denote the arm-$a$ empirical
measures on the evaluation fold and its complement, with equal arm-fold
sizes. Suppose the fitted rule satisfies exact training balance,
$P_{X,n}^{0,(-k)} g_k = P_{X,n}^{1,(-k)}\hat h_0^{(-k)}$, for every $k$ (as
holds for plan-integrated rules and empirical exact maps trained on
$\Omega^{(-k)}$). Define
$$
H_n := \frac{1}{K}\sum_{k=1}^K
P_{X,n}^{1,(-k)}\bigl(\hat h_0^{(-k)} - h_0\bigr),
\qquad
D_n := \frac{1}{K}\sum_{k=1}^K
\bigl(P_{X,n}^{0,(k)} - P_{X,n}^{0,(-k)}\bigr) g_k.
$$
Then the exact identity
\begin{equation}
\hat\tau_{ATT} - \tau_{ATT}
= \bigl(P_{n}^{1} - P^1\bigr)\bigl(Y - h_0(X)\bigr) - H_n - D_n
\label{eq:crossfit-plan-identity}
\end{equation}
holds. For any fixed function $g$,
$K^{-1}\sum_k(P_{X,n}^{0,(k)} - P_{X,n}^{0,(-k)})g = 0$; hence
$D_n = o_p(n^{-1/2})$ whenever, for some common $g_{\star}$,
\begin{equation}
\frac{1}{K}\sum_{k=1}^K
\bigl(P_{X,n}^{0,(k)} - P_{X,n}^{0,(-k)}\bigr)(g_k - g_{\star})
= o_p(n^{-1/2}).
\label{eq:crossfit-stability}
\end{equation}
Exact training balance alone does not imply \eqref{eq:crossfit-stability}.
Hypothesis 1 of \Cref{thm:balance-error-regimes} implies its out-of-fold
half by the conditional Chebyshev argument. The in-fold half additionally
requires an in-sample empirical process bound, for example a complexity
condition on the class containing $g_k - g_{\star}$, and it holds trivially
under sample splitting.
\end{prop}

\begin{proof}
Training balance gives
$P_{X,n}^{0,(k)}g_k
= P_{X,n}^{1,(-k)}\hat h_0^{(-k)}
+ (P_{X,n}^{0,(k)} - P_{X,n}^{0,(-k)})g_k$. Average over $k$, use the
equal-fold identity $K^{-1}\sum_k P_{X,n}^{1,(-k)}f = P_{X,n}^{1}f$ for
fixed $f$ applied to $h_0$, and subtract
$\tau_{ATT} = P^1(Y - h_0(X))$, which is valid under the identification assumptions. The fixed-function
cancellation is the same equal-fold identity applied on the control arm; the
sufficiency of \eqref{eq:crossfit-stability} follows by adding and
subtracting $g_{\star}$, and the final claim by applying conditional
Chebyshev to $(P_{X,n}^{0,(k)} - P_X^0)(g_k - g_{\star})$ and
$(P_{X,n}^{0,(-k)} - P_X^0)(g_k - g_{\star})$ separately.
\end{proof}

\begin{prop}[Exact cross-fitted variance identity]
\label{prop:crossfit-variance-identity}
Define
$\varepsilon_k := (P_{n}^{1,(k)} - P^1)Y
 - (P_{X,n}^{0,(k)} - P_X^0)g_k$ and
$\overline\varepsilon := K^{-1}\sum_k\varepsilon_k$, with equal arm-fold sizes
$n_a/K$. Then
\begin{align}
\operatorname{Var}(\hat\tau_{ATT})
&= \frac{\operatorname{Var}(Y\mid A=1)}{n_1}
 + \frac{K^{-1}\sum_k\E\bigl[\operatorname{Var}_{P_X^0}(g_k)\bigr]}{n_0}
 + \operatorname{Var}(\overline B_n) + R_{\mathrm{cf}},
\label{eq:full-crossfit-variance}\\
R_{\mathrm{cf}}
&= \frac{2}{K^2}\sum_{k<\ell}\E(\varepsilon_k\varepsilon_\ell)
 - 2\operatorname{Cov}(\overline\varepsilon, \overline B_n).
\label{eq:crossfit-remainder}
\end{align}
In ordinary $K$-fold cross-fitting the training samples overlap, so the fold
contributions are not conditionally independent and $R_{\mathrm{cf}}$ need
not vanish; it is zero under sample splitting
(\Cref{prop:honest-split-map-variance}) and is absorbed into the $o_p(1)$
of \eqref{eq:master-expansion} under the hypotheses of
\Cref{thm:balance-error-regimes}.
\end{prop}

\begin{proof}
Each $\varepsilon_k$ is mean zero, and conditional on the data outside fold
$k$, $\operatorname{Var}(\varepsilon_k\mid\Omega^{(-k)})
= K\operatorname{Var}(Y\mid A=1)/n_1
+ K\operatorname{Var}_{P_X^0}(g_k)/n_0$. Expanding
$\operatorname{Var}(\overline\varepsilon - \overline B_n)$ and collecting
the diagonal terms gives
\eqref{eq:full-crossfit-variance}--\eqref{eq:crossfit-remainder}.
\end{proof}

\begin{remark}
The plug-in estimator $\hat\tau_{ATT}$ is not Neyman-orthogonal in $T$, which is why $\psi_T$ enters the asymptotic variance. A 
debiased variant
$$
\tilde\tau_{ATT}^{(k)}
:= \hat\tau_{ATT}^{(k)}
   + \frac{1}{|\Omega^{(-k)}|}\sum_{j\in\Omega^{(-k)}}
            \hat\psi_T^{(-k)}(X_j,A_j),
$$
with $\hat\psi_T^{(-k)}$ a plug-in estimator of $\psi_T$ obtained by substituting $\hat h_0^{(-k)}$ for $h_0$ and $\hat\phi_T^{(-k)}$ for $\phi_T$ in the auxiliary fold, removes the $\psi_T$ term from the asymptotic linear representation. Under standard regularity for $\hat\psi_T^{(-k)}$, the corrected estimator satisfies
$$\sqrt n(\tilde\tau_{ATT}-\tau_{ATT})\Rightarrow\mathcal N(0,V_{\mathrm{or}}),$$
recovering the original variance $V_{\mathrm{or}}=\operatorname{Var}(Y\mid A=1)/r+\operatorname{Var}(h_0(X)\mid A=1)/(1-r)$.

We do not pursue the one-step correction here. The plug-in version is what is computed, and the augmented variance $V$ is the variance of the estimator actually computed. The construction of $\hat\phi_T^{(-k)}$ is non-trivial and depends on the specific OT-map estimator, such as the Brenier plug-in, entropic OT, or semi-discrete map. We give an example for the entropic OT map in \Cref{app:1b-prime}.
\end{remark}

\subsection{Asymptotic Linearity of the Outcome Regression}
\label{app:regression-al}

This appendix states the regression analogue of condition (A1), used by
\Cref{rem:root-n-regression}. When the outcome regression converges at
exactly the root-$n$ rate, its error contributes an influence term $\psi_h$
symmetric to $\psi_T$.

\begin{corollary}[Asymptotically linear outcome-regression estimator]
\label{cor:al-regression}
Let $Q^{\dagger} := T^{\dagger}_{\sharp}P_X^0$; when $T^{\dagger} = T$ is
exact, $Q^{\dagger} = P_X^1$. Suppose conditions (A1) and the $h_0$
smoothness of (A4) hold, together with:
\begin{enumerate}
\item[(A6)]
$\hat h_0^{(-k)}$ admits an asymptotically linear expansion: there exists a
measurable $\phi_h:\calx\times(\{0,1\}\times\calx\times\R)\to\R$ such that
for each fold $k$ and each $x\in\calx$,
$$
\hat h_0^{(-k)}(x) - h_0(x)
= \frac{1}{|\Omega^{(-k)}|}\sum_{j\in\Omega^{(-k)}}\phi_h(x; p_j)
 + r_n^{h,(k)}(x),
\qquad \|r_n^{h,(k)}\|_{L_2(Q^{\dagger})} = o_p(n^{-1/2}),
$$
with $\E[\phi_h(x; p)] = 0$ for $Q^{\dagger}$-a.e.\ $x$ and
$\int_{\calx}\E[\phi_h(x; p)^2]\, dQ^{\dagger}(x) < \infty$;
\item[(A7)] the estimator smoothness of (A4) holds with $L_2(P_X^1)$
replaced by $L_2(Q^{\dagger})$: $\hat h_0^{(-k)}$ is
$L^{(-k)}$-Lipschitz with $L^{(-k)} = O_p(1)$,
$\hat h_0^{(-k)}\in C^{1,1}(\calx)$ a.s.\ with
gradient Lipschitz constant $M^{(-k)} = O_p(1)$, and
$\|\nabla\hat h_0^{(-k)} - \nabla h_0\|_{L_2(Q^{\dagger})} = o_p(1)$.
\end{enumerate}
Define $\psi_h(p) := \E_{Q^{\dagger}}[\phi_h(X'; p)]$, where $X' \sim Q^{\dagger}$ is the inner integration variable. Then
\begin{equation}
\frac{1}{K}\sum_{k=1}^K B_{h,k}
= \frac{1}{n}\sum_{i=1}^n \psi_h(p_i) + o_p(n^{-1/2}),
\label{eq:reg-drift-al}
\end{equation}
with $\psi_h$ mean zero and $\E[\psi_h^2] < \infty$.
\end{corollary}

\begin{example}[OLS verifies (A6)]
\label{ex:ols-al}
In the linear outcome model of the Gaussian example,
$h_0(x) = \tilde x^{\top}\beta$ with
$Y_i = h_0(X_i) + \epsilon_i$, $\epsilon_i \sim N(0,\sigma^2)$ independent,
the OLS estimator satisfies
$\hat h_0^{(-k)}(x) - h_0(x) = \tilde x^{\top}(\hat\beta^{(-k)} - \beta)$ with
$$
\phi_h(x; p)
= \frac{1-a}{1-r}\,\tilde x^{\top}\tilde\Sigma_0^{-1}\tilde x_{p}\,
\bigl(y - h_0(x_p)\bigr),
\qquad
\tilde\Sigma_0 := \E[\tilde X\tilde X^{\top}\mid A=0],
$$
where $(x_p, y)$ are the covariate and outcome of the auxiliary observation
$p$. The conditions of (A6) are verified as follows.
\emph{Mean zero:} $\E[\phi_h(x;p)] = 0$ since
$\E[\epsilon \mid X, A=0] = 0$ by the definition of $h_0$.
Gaussianity of the noise is not needed.
\emph{Integrated second moment:} under homoskedastic noise,
$\E[\phi_h(x;p)^2]
= \tfrac{\sigma^2}{1-r}\,\tilde x^{\top}\tilde\Sigma_0^{-1}\tilde x$,
which is unbounded in $x$ on $\R^d$, so the supremum form of the condition
fails, but
$\int\E[\phi_h(x;p)^2]\,dP_X^1(x)
= \tfrac{\sigma^2}{1-r}\,\E_{P_X^1}\bigl[\tilde X^{\top}
\tilde\Sigma_0^{-1}\tilde X\bigr] < \infty$
since $P_X^1$ is Gaussian; under heteroskedastic noise the same holds with
$\sigma^2$ replaced by $\E[\sigma^2(X)\|\tilde\Sigma_0^{-1/2}\tilde X\|^2
\mid A=0]$-type moments.
\emph{Remainder:} $r_n^{h,(k)}(x) = \tilde x^{\top}v_n$ collects the
inverse-Gram replacement
$(\hat{\tilde\Sigma}_0^{-1} - \tilde\Sigma_0^{-1})$ acting on the
$O_p(n^{-1/2})$ score average and the arm-count ratio correction
$\bigl(\tfrac{|\Omega^{(-k)}|}{n_0^{(-k)}} - \tfrac{1}{1-r}\bigr)$ times the
same centred average; both give $\|v_n\| = O_p(n^{-1})$, hence
$\|r_n^{h,(k)}\|_{L_2(P_X^1)}
\le \|v_n\|\,(\E_{P_X^1}\|\tilde X\|^2)^{1/2} = O_p(n^{-1})
= o_p(n^{-1/2})$; this uses $\E\|\tilde X\|^4 < \infty$ (Gaussian).
The resulting
$\psi_h(p) = \frac{1-a}{1-r}\,\tilde\mu_1^{\top}\tilde\Sigma_0^{-1}
\tilde x_p(y - h_0(x_p))$ with $\tilde\mu_1 = \E[\tilde X\mid A=1]$ is a
non-degenerate root-$n$ contribution,
$\operatorname{Var}(\psi_h)
= \tfrac{\sigma^2}{1-r}\,\tilde\mu_1^{\top}\tilde\Sigma_0^{-1}\tilde\mu_1
> 0$: condition (A3) fails here, while (A6) holds.

Two caveats apply to this example. First, (A6) is an expansion around the
true $h_0$, so it requires the linear model to be correctly specified;
misspecified OLS is asymptotically linear around the best linear predictor
$h_0^{\dagger} \neq h_0$, which the present statement does not cover (a
regression pseudo-limit slot symmetric to $T^{\dagger}$ would, and would
add a bias term to $b^{\dagger}$; we do not pursue it here).
Second, the Lipschitz constant of $\hat h_0^{(-k)}$ is
$\|\hat\beta_1^{(-k)}\| = O_p(1)$, not a deterministic $L$, which is why
(A7) states the Lipschitz condition with $L^{(-k)} = O_p(1)$.
\end{example}

Under (A1), (A2), (A6), (A7) and (A5), the argument in the proof of
\Cref{thm:att_normality} applies with $\psi_B = \psi_T + \psi_h$ in regime
(R3) of \Cref{thm:balance-error-regimes}, giving influence function
$\psi_{Y,h_0,T} - \psi_T - \psi_h$ and variance
$V = V_{\mathrm{or}} + \operatorname{Var}(\psi_T + \psi_h)
- 2\operatorname{Cov}(\psi_{Y,h_0,T},\, \psi_T + \psi_h)$.

\begin{proof}[Proof of \Cref{cor:al-regression}]
Decompose
$B_{h,k} = Q^{\dagger}\bigl(\hat h_0^{(-k)} - h_0\bigr)
 + (Q_k - Q^{\dagger})\bigl(\hat h_0^{(-k)} - h_0\bigr)$.
The second term is $o_p(n^{-1/2})$ by \Cref{lem:cross-term}. For the first,
substitute (A6) and integrate against $Q^{\dagger}$: Fubini is justified by
the second-moment bound on $\phi_h$, the remainder satisfies
$|Q^{\dagger} r_n^{h,(k)}| \le \|r_n^{h,(k)}\|_{L_2(Q^{\dagger})}
= o_p(n^{-1/2})$, and the result is
$$
B_{h,k}
= \frac{1}{|\Omega^{(-k)}|}\sum_{j\in\Omega^{(-k)}}\psi_h(p_j)
 + o_p(n^{-1/2}),
$$
with $\psi_h$ mean zero by the kernel mean-zero condition and Fubini. The
second moment satisfies
$\E[\psi_h^2] \le \int_{\calx}\E[\phi_h(x;p)^2]\,dQ^{\dagger}(x) < \infty$,
by Jensen's inequality
$\psi_h(p)^2 \le \E_{X'\sim Q^{\dagger}}[\phi_h(X';p)^2]$ followed by
Tonelli. Fold aggregation as in
the proof of \Cref{cor:al-transport} gives \eqref{eq:reg-drift-al}.
\end{proof}

\section{Verification of the Conditions on the Transport-Map Estimators}
\label{app:map-conditions}

This appendix is the verification companion to \Cref{sec:role-of-map}. For each class of transport-map estimators, it identifies the condition the estimator satisfies and describes how the estimated map then enters the bias and the variance of the causal estimator. Two classes require no verification. First, any map or plan estimated on the sample on which the functional is evaluated needs no condition at all: by \Cref{prop:empirical-plan-cancellation}, the coupling never enters the estimator, and the map contributes nothing to the bias or the variance. Second, the out-of-sample nonparametric plug-in maps, namely the barycentric projection of an optimal empirical plan and its one-nearest-neighbour extension, are $L_2(P_X^0)$ consistent at rates of order $n^{-1/4}$ up to logarithmic factors for $d \le 4$ and of order $n^{-1/d}$ for $d \ge 5$ \citep[Proposition 14]{manole2021plugin}; see also \citet{deb2021rates}. This verifies the consistency condition 1(b) of \Cref{thm:att_consistency} by citation, and the rate is the only property used. No expansion of the form (A1) is known for these maps, so their contribution to the variance is not identified, and the estimator falls in regime (R1) or (R4) of \Cref{thm:balance-error-regimes}, according to the rate of the bias term $\overline B_n$. The two classes that do admit the expansion (A1) are verified in the following subsections.

\subsection{Condition (A1) for the Gaussian plug-in map}
\label{app:gaussian-a1}

Any map estimator that is a smooth function of finitely many asymptotically linear parameter estimates satisfies condition (A1) by the delta method, so every parametric plug-in map falls under this subsection. We verify the leading case, the Gaussian plug-in map of \Cref{sec:monge_OR}. Write $\theta := (\mu_0, \Sigma_0, \mu_1, \Sigma_1)$, $q_1 := r$ and $q_0 := 1-r$, and recall
$$
T_{\theta}(x) = \mu_1 + A(x-\mu_0),
\qquad
A = A(\Sigma_0, \Sigma_1)
:= \Sigma_0^{-1/2}\bigl(\Sigma_0^{1/2}\Sigma_1\Sigma_0^{1/2}\bigr)^{1/2}\Sigma_0^{-1/2}.
$$
Define the moment influence functions, for $a \in \{0,1\}$,
$$
\begin{aligned}
m_a(X_*, A_*) &:= \frac{\mathbf 1\{A_*=a\}}{q_a}\,(X_* - \mu_a),\\
S_a(X_*, A_*) &:= \frac{\mathbf 1\{A_*=a\}}{q_a}
\bigl[(X_*-\mu_a)(X_*-\mu_a)^{\top} - \Sigma_a\bigr].
\end{aligned}
$$

\begin{prop}[The Gaussian plug-in map satisfies (A1) and (A2)]
\label{app:prop:gaussian-a1}
Let $P_X^a = N(\mu_a, \Sigma_a)$ with $\Sigma_a \succ 0$ for $a \in \{0,1\}$, and let $\hat T^{(-k)} = T_{\hat\theta^{(-k)}}$ be the plug-in map built from the arm means and sample covariances on $\Omega^{(-k)}$. Then condition (A1) holds with limit map $T^{\dagger} = T_{\theta}$, which is the exact Gaussian transport map, so condition (A2) holds as well. The kernel is
\begin{equation}
\phi_T(x; X_*, A_*)
= m_1(X_*, A_*)
+ \bigl[\partial A(S_0, S_1)(X_*, A_*)\bigr](x - \mu_0)
- A\, m_0(X_*, A_*),
\label{app:gaussian-kernel}
\end{equation}
where $\partial A$ is the Fr\'echet derivative of $(\Sigma_0, \Sigma_1) \mapsto A$ at the truth, applied to the covariance influences. The kernel is affine in $x$; consequently
$$
\sup_{x \in \calx} \E\|\phi_T(x; X, A)\|^2 = \infty,
\qquad
\int_{\calx} \E\|\phi_T(x; X, A)\|^2\, dP_X^0(x) < \infty.
$$
\end{prop}

\begin{proof}[Proof sketch]
The arm means and sample covariances on the fold are asymptotically linear with kernels $m_a$ and $S_a$: the H\'ajek denominators are handled by \Cref{lem:delta-method}, and the estimated-mean correction in $\hat\Sigma_a$ is of second order. The map $\theta \mapsto T_{\theta}(x)$ is Fr\'echet differentiable at every $\theta$ with $\Sigma_0, \Sigma_1 \succ 0$, because matrix inversion and the matrix square root are analytic on the positive definite cone; its derivative in the direction $h = (h_{\mu_0}, h_{\Sigma_0}, h_{\mu_1}, h_{\Sigma_1})$ is $h_{\mu_1} + [\partial A(h_{\Sigma_0}, h_{\Sigma_1})](x - \mu_0) - A h_{\mu_0}$, which is affine in $x$. The pointwise delta method then gives \eqref{eq:map-al} with kernel \eqref{app:gaussian-kernel}. On the event that $\hat\theta^{(-k)}$ lies in a fixed neighborhood of $\theta$, whose probability tends to one, the second-order Taylor remainder is bounded by $C\,\|\hat\theta^{(-k)} - \theta\|^2(1 + \|x\|)$ for a constant $C$ depending on $\theta$, so its $L_2(P_X^0)$ norm is $O_p(n^{-1})$, which verifies the remainder requirement of (A1). The kernel is mean zero because $m_a$ and $S_a$ are. For the second moment, affinity in $x$ gives $\E\|\phi_T(x; X, A)\|^2 \le C'(1 + \|x\|^2)$, with $C'$ finite because the Gaussian law has finite fourth moments. The bound is not uniform in $x$, so the supremum form fails, while integrating against $P_X^0$, which has finite second moments, yields the integrated form required by condition (A1).
\end{proof}

In consequence, \Cref{thm:att_normality} applies to the Gaussian plug-in map with the $\psi_T$ term retained. The sample means and covariances are asymptotically efficient under Gaussianity, so hypothesis (iii) of \Cref{cor:exact-map-cancellation} holds, and the exact variance identity \eqref{eq:oracle-exact-map-variance} characterizes the map contribution. Under misspecification, the same argument yields condition (A1) around the Gaussian pseudo-map, which is no longer exact; condition (A2) then fails, and the situation is that of the entropic map in \Cref{app:1b-prime}.

\subsection{Condition (A1) for the empirical entropic map}
\label{app:1b-prime}

This subsection verifies condition (A1) of \Cref{thm:att_normality} for the empirical entropic optimal transport (EOT) map at a fixed regularization $\epsilon > 0$. The verification composes three external results: an empirical-process CLT on a H\"older-ball Donsker class \citep[Theorem 2.7.1]{van1996weak}, the Hadamard differentiability of the EOT potentials with respect to the marginal distributions \citep[Section 4]{goldfeld2024limit}, and the functional delta method \citep[Theorem 3.9.4]{van1996weak}. The composition is not stated as a named result in the cited works, so we give it here and verify the kernel qualifiers of (A1). The functional delta method used here acts on maps between function spaces and is unrelated to the scalar delta-method expansion of \Cref{lem:delta-method}, which handles random denominators.

We work under three primitive hypotheses. (P1): the marginals $P_X^0, P_X^1$ are supported in a common bounded closed ball $\mathcal X \subset \mathbb R^d$. (P2): the transport cost is quadratic, $c(x, y) = \frac{1}{2}\|x - y\|^2$. (P3): the regularization $\epsilon > 0$ is fixed, and $\hat T^{(-k)}$ is the empirical EOT (Sinkhorn) map at level $\epsilon$ on the auxiliary fold. Observations are i.i.d.\ tuples $(X_j, A_j) \sim \nu$ on $\mathcal X \times \{0, 1\}$ with $r = \nu(A=1) \in (0,1)$, the auxiliary fold has size $N := (K-1)n/K$, and $\hat\nu^{(-k)}$ denotes the joint empirical measure on the fold. Fix an integer $s > d/2$ and let $B_s$ be the closed unit ball of $C^s(\mathcal X)$ in the H\"older norm. Under (P2) the EOT map satisfies $T_\epsilon(x) = x - \nabla\varphi_1^{\epsilon}(x)$, where $\varphi_1^{\epsilon}$ is the first EOT potential \citep[Proposition 2]{pooladian2021entropic}. Writing $\Lambda$ for the map from the joint law to the pair of conditional marginals and $\Phi^\epsilon$ for the map from the marginal pair to the EOT potentials, normalized as in \citet[Section 2]{goldfeld2024limit}, the EOT map is the composite
\begin{equation}
\Psi^\epsilon : \nu \longmapsto \mathrm{id}_{\mathcal X} - \nabla\,\Pi_1\,\Phi^\epsilon(\Lambda(\nu)) \in C^{s-1}(\mathcal X; \mathbb R^d),
\label{app:def-Psi}
\end{equation}
where $\Pi_1$ selects the first potential, so that $\hat T^{(-k)} = \Psi^\epsilon(\hat\nu^{(-k)})$ and $T_\epsilon = \Psi^\epsilon(\nu)$.

\begin{prop}[Verification of condition (A1) under (P1)-(P3)]
\label{app:prop:1b}
Under \textnormal{(P1)-(P3)} and $r \in (0, 1)$, the empirical EOT map
$\hat T^{(-k)} = \Psi^\epsilon(\hat\nu^{(-k)})$ satisfies condition (A1) with limit map $T^{\dagger} = T_\epsilon$ and kernel
\begin{equation}
\phi_T(x; X_*, A_*) := [\Psi^{\nu, \epsilon}]'\bigl(\delta_{(X_*, A_*)} - \nu\bigr)(x),
\quad x \in \mathcal X,\ (X_*, A_*) \in \mathcal X \times \{0, 1\},
\label{app:phi-T-def}
\end{equation}
where $[\Psi^{\nu, \epsilon}]'$ is the Hadamard derivative of $\Psi^\epsilon$ at $\nu$. The kernel is mean zero and satisfies the supremum bound $\sup_x \E\|\phi_T(x; X, A)\|^2 < \infty$, which implies the integrated second moment required by condition (A1).
\end{prop}

\begin{proof}
\emph{Step 1: empirical-process CLT.}
Under (P1) and $s > d/2$, the class $B_s \cup \{f \mathbf 1_{A=0}, f\mathbf 1_{A=1} : f \in B_s\}$ is $\nu$-Donsker \citep[Theorem 2.7.1]{van1996weak}. Hence $\sqrt N(\hat\nu^{(-k)} - \nu) \Rightarrow \mathbb G^\nu$ in $\ell^\infty(B_s)$, and the limit is concentrated on the closed span $\mathbb D_0$ of the centred evaluations $\delta_{(x,a)} - \nu$.

\emph{Step 2: Hadamard differentiability of the composite.}
The conditioning map $\Lambda$ is Hadamard differentiable at $\nu$, tangentially to $\ell^\infty(B_s)$ and for $r \in (0,1)$, by the quotient rule \citep[Lemma 3.10.20]{van1996weak}. The potential map $\Phi^\epsilon$ is Hadamard differentiable at the marginal pair, with a linear continuous derivative into $C^s(\mathcal X) \times C^s(\mathcal X)$; this is the main technical input, established in \citet[Section 4]{goldfeld2024limit} through the Schr\"odinger system and a Banach-space implicit function theorem. The remaining components of \eqref{app:def-Psi} are linear and continuous. By the chain rule \citep[Theorem 3.9.6]{van1996weak}, $\Psi^\epsilon$ is Hadamard differentiable at $\nu$ tangentially to $\mathbb D_0$, as a map into $C^{s-1}(\mathcal X; \mathbb R^d)$.

\emph{Step 3: functional delta method.}
Apply the functional delta method \citep[Theorem 3.9.4]{van1996weak} with $\hat\nu^{(-k)}$, rate $\sqrt N$, and limit $\mathbb G^\nu$. The linearization step inside its proof, which is the statement condition (A1) requires, gives
\begin{equation}
\hat T^{(-k)} - T_\epsilon = [\Psi^{\nu, \epsilon}]'\bigl(\hat\nu^{(-k)} - \nu\bigr) + R_n^{T, (k)},
\quad \|R_n^{T, (k)}\|_{C^{s-1}} = o_p(N^{-1/2}),
\label{app:linearization-Cs1}
\end{equation}
in $C^{s-1}(\mathcal X; \mathbb R^d)$.

\emph{Step 4: kernel form.}
Decompose
$$
\hat\nu^{(-k)} - \nu = \frac{1}{N}\sum_{j \in \Omega^{(-k)}}\bigl(\delta_{(X_j, A_j)} - \nu\bigr).
$$
By linearity of $[\Psi^{\nu, \epsilon}]'$ on $\mathbb D_0$,
$$
[\Psi^{\nu, \epsilon}]'\bigl(\hat\nu^{(-k)} - \nu\bigr) = \frac{1}{N}\sum_{j \in \Omega^{(-k)}}\phi_T(\cdot; X_j, A_j),
$$
with $\phi_T$ as in \eqref{app:phi-T-def}. Substituting into \eqref{app:linearization-Cs1},
\begin{equation}
\hat T^{(-k)}(x) - T_\epsilon(x) = \frac{1}{N}\sum_{j \in \Omega^{(-k)}}\phi_T(x; X_j, A_j) + R_n^{T, (k)}(x),
\quad \|R_n^{T, (k)}\|_{C^{s-1}} = o_p(N^{-1/2}).
\label{app:linearization-Cs1-kernel}
\end{equation}

\emph{Step 5: embedding and rate translation.}
Under (P1), $\mathcal X$ is bounded, so the inclusion
$$
C^{s-1}(\mathcal X; \mathbb R^d) \hookrightarrow L_2(P_X^0; \mathbb R^d)
$$
is continuous with operator norm at most one. Hence
$\|R_n^{T, (k)}\|_{L_2(P_X^0)} = o_p(N^{-1/2})$, and since
$N = (K-1)n/K$ with $K$ fixed, this rate is $o_p(n^{-1/2})$. The expansion
\eqref{app:linearization-Cs1-kernel} is then the statement of condition
(A1) with $r_n^{T,(k)} := R_n^{T,(k)}$.

\emph{Step 6: kernel qualifiers.}
Mean zero follows from linearity and continuity of the derivative, since $\E[\delta_{(X,A)} - \nu] = 0$ in $\ell^\infty(B_s)$, so $\E[\phi_T(x; X, A)] = [\Psi^{\nu, \epsilon}]'(0)(x) = 0$. For the second moment, let $L$ be the operator norm of $[\Psi^{\nu, \epsilon}]'$, finite by continuity. Every $f \in B_s$ satisfies $\|f\|_\infty \le 1$, so $\|\delta_{(X_*, A_*)} - \nu\|_{\ell^\infty(B_s)} \le 2$ and $\|\phi_T(\cdot; X_*, A_*)\|_{C^{s-1}} \le 2L$. The embedding of Step 5 then gives $\sup_x \E\|\phi_T(x; X, A)\|^2 \le 4L^2 < \infty$.
\end{proof}

Two caveats apply. First, \citet{goldfeld2024limit} state central limit theorems for the empirical EOT potentials and map, but not the asymptotically linear expansion itself; the expansion is the linearization step inside the proof of the functional delta method, which condition (A1) requires and the CLT alone does not supply. Second, the embedding in Step 5 uses the bounded common support of (P1); without it, the H\"older-norm rate does not transfer to $L_2(P_X^0)$. The construction applies at any fixed $\epsilon > 0$, with the operator norm $L$ and the kernel $\phi_T$ depending on $\epsilon$; the case of vanishing $\epsilon$ is outside its scope. In the terminology of \Cref{cor:exact-map-cancellation}, the kernel \eqref{app:phi-T-def} represents the pathwise derivative of the map-valued functional $\Psi^\epsilon$; the corollary itself does not apply to the entropic map, because $T_\epsilon$ is not exact. In consequence, the central limit theorem applies with the limit centered at the pseudo-true parameter $\tau_{ATT} - b^{\dagger}$, as discussed in \Cref{sec:role-of-map}.

\end{document}